%% file: paper.tex
\documentclass{article}
\usepackage{arxiv}
\usepackage{natbib}
\usepackage{booktabs} 
\usepackage{algorithm}
\usepackage[noend]{algpseudocode}
\usepackage{tikz}
\usetikzlibrary{calc}

\usepackage{amsthm}
\usepackage{mathtools}
\usepackage{amsmath}
\usepackage{amssymb}
\usepackage{comment}
\usepackage{url}

\title{Selling Information while Being an Interested Party}

\usepackage{thm-restate}

\newtheorem{definition}{Definition}

\newtheorem{theorem}{Theorem}
\newtheorem{corollary}{Corollary}

\author{
	Matteo Castiglioni\\
	Politecnico di Milano\\
	\texttt{matteo.castiglioni@polimi.it}
	\And
	Francesco Bacchiocchi\\
	Politecnico di Milano\\
	\texttt{francesco.bacchiocchi@polimi.it}
	\And
	Alberto Marchesi\\
	Politecnico di Milano\\
	\texttt{alberto.marchesi@polimi.it}
	\And
	Giulia Romano\\
	Politecnico di Milano\\
	\texttt{giulia.romano@polimi.it}
	\And
	Nicola Gatti\\
	Politecnico di Milano\\
	\texttt{nicola.gatti@polimi.it}
}

\input{defs}

\begin{document}
	\maketitle
\input{content/abstract}

\input{content/intro}

\input{content/prelim}

\input{content/type_reporting}

\input{content/principal_agent}

\input{content/hardness}
\input{content/const_actions}

\input{content/const_states}

\input{content/general}
\input{content/const_types}

\bibliographystyle{ACM-Reference-Format}
\bibliography{biblio}

\clearpage
\appendix

\input{content/appendix_type_reporting}

\input{content/appendix_hardness_first}

\input{content/appendix_const_actions}

\input{content/appendix_const_states}

\input{content/appendix_hardness_second}

\input{content/appendixGeneral}

\input{content/appendix_const_types}

\end{document}

%% file: defs.tex
\usepackage{color}
\usepackage{pifont}
\usepackage[normalem]{ulem} 

\definecolor{mygreen}{rgb}{0.0, 0.5, 0.0}
\definecolor{myorange}{rgb}{0.55, 0.62, 1}
\newcommand{\nb}[3]{{\colorbox{#2}{\bfseries\sffamily\scriptsize\textcolor{white}{#1}}}{\textcolor{#2}{\sf\small\textsf{#3}}}}

\newcommand{\mat}[1]{\nb{Matteo}{purple}{#1}}

\usepackage[utf8]{inputenc} 
\usepackage[T1]{fontenc}    
\usepackage{hyperref}       
\usepackage{url}            
\usepackage{booktabs}       
\usepackage{amsfonts}       
\usepackage{nicefrac}       
\usepackage{microtype}      
\usepackage{xcolor}         

\usepackage{amssymb}
\usepackage{mathtools}
\usepackage{amsmath}
\usepackage{tikz}
\usepackage{pgfplots}
\usetikzlibrary{shapes.geometric}
\pgfplotsset{compat=newest}
\usepgfplotslibrary{groupplots}
\usepgfplotslibrary{dateplot}
\usepackage{todonotes}
\usepackage{dsfont}
\usepackage{enumitem}
\usepackage{tikz}
\usepackage{amsthm}
\usepackage{thm-restate}
\usepackage{booktabs}
\usepackage{multirow}
\usepackage{newfloat}
\usepackage{wrapfig}
\usepackage{dsfont}
\usepackage{comment}
\usepackage{multirow}
\usepackage{mleftright}
\usepackage{bm}
\usepackage{comment}
\usepackage{bbm}
\usepackage{subfig}
\usepackage{placeins}
\usepackage{colortbl}

\usepackage{pgf,tikz,pgfplots}
\pgfplotsset{compat=1.15}
\usepackage{mathrsfs}
\usetikzlibrary{arrows}

\usepackage[switch]{lineno}
\makeatletter
\AtBeginDocument{%
	\@ifpackageloaded{amsmath}{%
		\newcommand*\patchAmsMathEnvironmentForLineno[1]{%
			\expandafter\let\csname old#1\expandafter\endcsname\csname #1\endcsname
			\expandafter\let\csname oldend#1\expandafter\endcsname\csname end#1\endcsname
			\renewenvironment{#1}%
			{\linenomath\csname old#1\endcsname}%
			{\csname oldend#1\endcsname\endlinenomath}%
		}%
		\newcommand*\patchBothAmsMathEnvironmentsForLineno[1]{%
			\patchAmsMathEnvironmentForLineno{#1}%
			\patchAmsMathEnvironmentForLineno{#1*}%
		}%
		\patchBothAmsMathEnvironmentsForLineno{equation}%
		\patchBothAmsMathEnvironmentsForLineno{align}%
		\patchBothAmsMathEnvironmentsForLineno{flalign}%
		\patchBothAmsMathEnvironmentsForLineno{alignat}%
		\patchBothAmsMathEnvironmentsForLineno{gather}%
		\patchBothAmsMathEnvironmentsForLineno{multline}%
	}{}
}
\makeatother

\DeclareMathOperator*{\A}{\mathcal{A}}

\DeclareMathOperator*{\argmax}{argmax}

\newcommand{\R}{\mathcal{R}}

\allowdisplaybreaks

\newcommand{\br}{b^{k}_{\xi, \pi}}
\newcommand{\ebr}{b^{k, \epsilon}_{\xi, \pi}}
\newcommand{\defeq}{\coloneqq}

\newcommand{\sset}{\mathcal{S}}

\newcommand{\p}{\xi}
\newcommand{\pset}{{\Xi}}

\newcommand{\pr}{\textnormal{Pr}}

\newcommand{\xvec}{\mathbf{x}}

\newcommand{\K}{\mathcal{K}}

\newcommand{\NPHard}{$\mathsf{NP}$-hard}

\newcommand{\supp}{\textnormal{supp}}

\newcommand{\send}{\mathsf{s}}
\newcommand{\avec}{\boldsymbol{a}}

\newcommand{\Expec}{\mathbb{E}}

\newcommand{\brset}{\mathcal{B}}

\newcommand{\nAct}{m}

\newcommand{\nState}{d}

\newcommand{\nVar}{n_\textnormal{var}}
\newcommand{\nEq}{n_\textnormal{eq}}
\newcommand{\epsLM}{\zeta}
\newcommand{\deltaLM}{\delta}
\newcommand{\normaliz}{\tau}

\newcommand{\eps}{\epsilon}

\makeatletter
\newcommand{\subalign}[1]{%
	\vcenter{%
		\Let@ \restore@math@cr \default@tag
		\baselineskip\fontdimen10 \scriptfont\tw@
		\advance\baselineskip\fontdimen12 \scriptfont\tw@
		\lineskip\thr@@\fontdimen8 \scriptfont\thr@@
		\lineskiplimit\lineskip
		\ialign{\hfil$\m@th\scriptstyle##$&$\m@th\scriptstyle{}##$\hfil\crcr
			#1\crcr
		}%
	}%
}
\makeatother

\makeatletter
\newcommand{\pushright}[1]{\ifmeasuring@#1\else\omit\hfill$\displaystyle#1$\fi\ignorespaces}
\newcommand{\pushleft}[1]{\ifmeasuring@#1\else\omit$\displaystyle#1$\hfill\fi\ignorespaces}
\newcommand{\specialcell}[1]{\ifmeasuring@#1\else\omit$\displaystyle#1$\ignorespaces\fi}
\makeatother

%% file: content/abstract.tex
\begin{abstract}
	We study the algorithmic problem faced by an information holder (\emph{seller}) who \emph{wants to optimally sell such information} to a budged-constrained decision maker (\emph{buyer}) that has to undertake some action. Differently from previous works addressing this problem, we consider the case in which the seller is an \emph{interested party}, as the action chosen by the buyer does \emph{not} only influence their utility, but also seller's one. This happens in many real-world settings, where the way in which businesses use acquired information may positively or negatively affect the seller, due to the presence of externalities on the information market. The utilities of both the seller and the buyer depend on a random \emph{state of nature}, which is revealed to the seller, but it is unknown to the buyer. Thus, the seller's goal is to (partially) sell their information about the state of nature to the buyer, so as to concurrently maximize revenue and induce the buyer to take a desirable action.

	We study settings in which buyer's budget and utilities are determined by a random buyer's \emph{type} that is unknown to the seller. In such settings, an optimal \emph{protocol} for the seller must propose to the buyer a \emph{menu} of information-revelation policies to choose from, with the latter acquiring one of them by paying its corresponding \emph{price}. Moreover, since in our model the seller is an interested party, an optimal protocol must also prescribe the seller to pay back the buyer contingently on their action.

	First, we show that the problem of computing a seller-optimal protocol can be solved in polynomial time. This result relies on a quadratic formulation of the problem, which we solve by means of a linear programming relaxation. Next, we switch the attention to the case in which a seller's protocol employs a single information-revelation policy, rather than proposing a menu. In such a setting, we show that computing a seller-optimal protocol is $\mathsf{APX}$-hard, even when either the number of actions or that of states of nature is fixed. We complement such a negative result by providing a quasi-polynomial-time approximation algorithm that, given any $\rho > 0$ and $\epsilon > 0$ as input, provides a multiplicative approximation $\rho$ of the optimal seller's expected utility, by only suffering a negligible $2^{-\Omega(1/\rho)}+\epsilon$ additive loss. Such an algorithm runs in polynomial time whenever either the number of buyer's actions or that of states of nature is fixed. In order to derive our results, we draw a connection between our information-selling problem and principal-agent problems with observable actions. Finally, we complete the picture of the computational complexity of finding seller-optimal protocols without menus by providing additional results for the specific setting in which the buyer has limited liability, and by designing a polynomial-time algorithm for the case in which buyer's types are fixed.
\end{abstract}

%% file: content/intro.tex
\section{Introduction}

Nowadays, there is a terrific amount of information being collected on the Web and other online platforms.
Such information ranges from consumer preferences, \emph{e.g.}, in e-commerce and streaming websites, to credit reports and location histories. 
As a result, recent years have witnessed the born and exponential blowout of markets where specialized companies sell information that is valuable to other businesses, such as advertisers, retailers, and loan providers.

Very recently, information markets have also received the attention of the algorithmic game theory research community.
However, while works addressing classical settings such as auctions~\citep{daskalakis2022learning}, signaling~\citep{dughmi2019algorithmic}, and contract design~\citep{dutting2019simple} are now proliferating, only few papers studied the problem of {information selling}, with ~\citep{babaioff2012}~and~\citep{xu2020} constituting two notable examples.

We study the algorithmic problem faced by an information holder (\emph{seller}) who \emph{wants to optimally sell such information} to a budged-constrained decision maker (\emph{buyer}) that has to undertake some action.
Differently from previous works addressing such a problem (see, \emph{e.g.},~\citep{xu2020}), we consider the case in which the seller is an \emph{interested party}, as the action chosen by the buyer does \emph{not} only influence their utility, but also seller's one.
This happens in many real-world settings, where the way in which businesses use acquired information may positively or negatively affect the seller, due to the presence of externalities on the information market.
The utilities of both the seller and the buyer depend on a \emph{state of nature} that is drawn according to a commonly-known probability distribution.
The realized state of nature is revealed to the seller, while it remains unknown to the buyer.
Thus, the seller's goal is to (partially) sell their information about the state of nature to the buyer, so as to concurrently maximize revenue and induce the buyer to take a desirable action.

We study settings in which buyer's budget and utilities are determined by a random buyer's \emph{type} that is unknown to the seller.
In such settings, in order to optimally sell information, the seller has to commit upfront to a \emph{protocol} working as follows.
First, the seller proposes to the buyer a \emph{menu} of information-revelation policies to choose from, and the latter acquires an expected-utility-maximizing one according to their (private) type, by paying its corresponding \emph{price}.
By building on the \emph{Bayesian persuasion} framework introduced by~\citet{kamenica2011bayesian}, an information-revelation policy is implemented as a \emph{signaling scheme}, which is a randomized mapping from states of nature to signals issued to the buyer.
Then, the realized state of nature is disclosed to the seller, who reveals information about it to the buyer according to the acquired signaling scheme.
Finally, the buyer selects a best-response action according to the just acquired information, and the seller pays back the buyer with a \emph{payment} which depends on both the chosen action and the signal that has been previously sent by the seller.
Our protocol extends the one of~\citet{xu2020} by adding a final payment from the seller to the buyer.
As we show later, this is crucial in order to design seller-optimal protocols in our setting where the seller is an interested party, since the latter is \emph{not} only concerned with revenue, but also with the buyer's action.
Moreover, the addition of payments from the buyer to the seller is also reasonable in many real-world scenarios.
For instance, think of a case in which the information holder asks the buyer to deposit additional money, and this is given back to them only if the performed action respects some given rules on which the two parties agreed upfront.

\subsection{Original Contributions}

After introducing all the needed concepts in Section~\ref{sec:preliminaries}, we start providing our results in Section~\ref{sec:protocol_selection}, where we analyze the case of general protocols in which the seller proposes a menu of signaling schemes to the buyer.
We show that a seller-optimal protocol can be computed in polynomial time.
In order to do that, we first formulate the problem of finding a seller-optimal protocol as a quadratic problem.
%
%
%
Then, we show that one can focus on \emph{direct} and \emph{persuasive} signaling schemes, which are those that send signals corresponding to action recommendations for the buyer and properly incentivize the latter to follow such recommendations.
This in turn allows us to restrict the attention to protocols that ask the buyer to pay their entire budget upfront and, then, pay back the buyer only if they take the recommended action.
These results allow us to formulate a suitable linear relaxation of the quadratic problem.
A similar technique has been employed in generalized principal-agent problems~\citep{gan2022optimal}, where it is possible to show that an optimal solution to the linear relaxation can be efficiently cast to an approximately-optimal solution to the quadratic problem.
Indeed, \citet{Castiglioni2022Randomized} show that, even in the special case of hidden-action principal-agent problems, obtaining an optimal solution to the quadratic problem is \emph{not} possible in general, since the principal's optimization problem may \emph{not} admit a maximum.
Surprisingly, in our information-selling setting, we prove that an optimal solution to our linear relaxation, which can be computed in polynomial time, can be used to recover a seller-optimal protocol in polynomial time.
As a byproduct, this also shows that, in our setting, the seller's problem always admits a maximum.

In the second part of the paper, we switch the attention to the case of protocols \emph{without menus}, in which the seller does \emph{not} propose a menu of signaling schemes to the buyer, but they rather commit to a single signaling scheme.
This is the case in many real-world applications, where it is unreasonable that a buyer is asked to choose an information-revelation policy among a range of options.
%
%
%
Computing a seller-optimal protocol without menus begets considerable additional computational challenges, since, intuitively, the seller has no way of extracting information about the buyer's private type, as instead it is the case when proposing a menu to choose from.

In Section~\ref{sec:principal_agent}, we draw a connection between the problem of computing a seller-optimal protocol without menus and principal-agent problems with \emph{observable actions}.
These are problems in which a principal commits to an action-dependent payment scheme in order to incentivize an agent to take some costly, observable action, in order to maximize their expected utility.
We prove that observable-action principal-agent problems are a special case of our information-selling problem, and that, in such problems, computing an expected-utility-maximizing payment-scheme for the principal is $\mathsf{APX}$-hard.
In particular, these results show that our information-selling problem is $\mathsf{APX}$-hard even when the number of states of nature is fixed and the buyer has \emph{limited liability}, and, thus, the seller cannot charge a price for a signaling scheme upfront.
%
%
We also provide some preliminary technical results on observable-action principal-agent problems, which are useful in order to prove some of our main claims in the paper, while also being of independent interest.

In Section~\ref{sec:selection_limited_liability}, we show how to circumvent the $\mathsf{APX}$-hardness for settings in which the seller employs protocols without menus and the buyer has limited liability.
%
%
%
These special settings are of interested on their own, as a similar model has been recently addressed by~\citet{dughmi2019}.
%
%
%
We focus on special cases where one of the parameters characterizing a problem instance is fixed.
In particular, we study what happens if we fix the number of buyer's actions, 
showing that the problem admits a PTAS.
%
Moreover, we prove that, when instead the number of states of nature is fixed, there exists a polynomial-time \emph{bi-criteria} approximation algorithm that, given any $\rho > 0$ and $\epsilon > 0$ as input, provides a multiplicative approximation $\rho$ of the optimal seller's expected utility, by only suffering a
$2^{-\Omega(1/\rho)}+\epsilon$ additive loss.
Notice that such a loss is exponentially small in $\frac{1}{\rho}$, and, thus, it is negligible even for reasonably large values of $\rho$. 
As shown by~\citet{castiglioni2021Contract}, such an approximation result is tight for hidden-action principal-agent problems.
It remains an open problem to establish whether such an approximation guarantee is also tight for principal-agent problems with observable actions, which are a special case of our information-selling problem.
%
%

\begin{table}[!htp]
	\caption{Summary of the results provided in the paper. Each cell specifies, on the first line, the computational complexity of finding a seller-optimal protocol, while, additionally, on the second line, it specifies the approximation guarantees that we can obtain in polynomial time, where \textsc{opt} denotes the seller's expected utility in an optimal protocol. The approximation guarantees that are shaded in gray can only be obtained by means of a quasi-polynomial-time algorithm.}
	\label{table-results}
	\centering
	{\renewcommand{\arraystretch}{1.2}
		\setlength{\tabcolsep}{2pt}
		\begin{tabular}{c|c|c|c|c|}
			& general &  fixed \# actions &  fixed \# states &  fixed \# types \\
			\hline
			Protocols \emph{with menus}
			& $\mathsf{P}$
			& $\mathsf{P}$
			& $\mathsf{P}$
			& $\mathsf{P}$  \\
			\hline
			\multirow{2}{*}{\shortstack[c]{Protocols \emph{w/o menus} \\ Buyer \emph{w. limited liability}}}
			& $\mathsf{APX}$-hard 
			& ---
			&  $\mathsf{APX}$-hard
			& \multirow{2}{*}{$\mathsf{P}$}  \\
			& \cellcolor{gray!15} $\rho \textsc{opt} \hspace{-0.5mm}-\hspace{-0.5mm} 2^{-\Omega(\nicefrac{1}{\rho})}\hspace{-0.5mm}-\hspace{-0.5mm}\epsilon$
			& PTAS
			&  $\rho \textsc{opt} \hspace{-0.5mm}-\hspace{-0.5mm} 2^{-\Omega(\nicefrac{1}{\rho})}\hspace{-0.5mm}-\hspace{-0.5mm}\epsilon$
			& 
			\\
			\hline
			\multirow{2}{*}{\shortstack[c]{Protocols \emph{w/o menus} \\ Buyer {w/o} limited liability}}
			& $\mathsf{APX}$-hard 
			& $\mathsf{APX}$-hard
			& $\mathsf{APX}$-hard
			& \multirow{2}{*}{$\mathsf{P}$} \\
			& \cellcolor{gray!15} $\rho \textsc{opt} \hspace{-0.5mm}-\hspace{-0.5mm} 2^{-\Omega(\nicefrac{1}{\rho})}\hspace{-0.5mm}-\hspace{-0.5mm}\epsilon$
			& $\rho \textsc{opt} \hspace{-0.5mm}-\hspace{-0.5mm} 2^{-\Omega(\nicefrac{1}{\rho})}\hspace{-0.5mm}-\hspace{-0.5mm}\epsilon$
			& $\rho \textsc{opt} \hspace{-0.5mm}-\hspace{-0.5mm} 2^{-\Omega(\nicefrac{1}{\rho})}\hspace{-0.5mm}-\hspace{-0.5mm}\epsilon$
			& \\
			\hline
	\end{tabular}}
\end{table}

In conclusion, in Section~\ref{sec:selection_no_limited_liability} we study the problem of computing seller-optimal protocols without menus in general settings in which the buyer does \emph{not} have limited liability, and, thus, the seller can charge a price for a signaling scheme.
%
%
%
We first prove a stronger negative result, by showing that, in such a setting, the problem of computing a seller-optimal protocol is $\mathsf{APX}$-hard even if the number of buyer's actions is fixed.
%
%
Then, we show how to circumvent such a negative result by providing a quasi-polynomial-time \emph{bi-criteria} approximation algorithm that, given any $\rho > 0$ and $\epsilon > 0$ as input, provides a multiplicative approximation $\rho$ of the optimal seller's expected utility, plus a
$2^{-\Omega(1/\rho)}+\epsilon$ additive loss.
%
We prove that, when either the number of buyer's action or that of states of nature is fixed, such an algorithm runs in polynomial time.
Finally, we show that, when the number of buyer's types is fixed, the problem admits a polynomial-time algorithm.
This also implies that the seller's optimization problem for protocols without menus always admits a maximum.
%

We summarize the results provided in this paper in Table~\ref{table-results}.
All the proofs are in the Appendix.

\subsection{Related Works}

The study of algorithmic ways of selling information to an imperfectly-informed buyer has received some attention in the past.
\citet{babaioff2012} initiated the study by considering a buyer with an unlimited budget.
They provide an exponentially-sized \emph{linear program} (LP) for computing an optimal mechanism for selling information, and they efficiently solve it
through the ellipsoid method.
The main drawback of the approach presented by~\citet{babaioff2012} is that an
optimal mechanism may require a significant money transfer from the buyer to the seller and \emph{viceversa}, in order to only
achieve a small, overall net transfer.
\citet{xu2020} complement the results in~\citep{babaioff2012} by studying the problem of selling information when both the buyer and the seller are budget-constrained.
Moreover, they also consider a setting in which the buyer's budget is private, and the seller needs to elicit it in the mechanism.
\citet{xu2020} show that the addition of budget constraints considerably simplifies the problem of computing an optimal mechanism, since it can be formulated as a polynomially-sized LP.
%

The problem of selling information has also been addressed by \citet{bergemann2018}, who study the case of binary actions and states of nature, characterizing a revenue-maximizing mechanism in such a setting.
Furthermore, \citet{bergemann2022} extend the analysis to the case in which there are more than two actions and binary states of nature.
%
In contrast, \citet{Liu2021} study a revenue-maximizing mechanism for selling information when the stochasticity of the state of nature only affects a subset of the actions of the decision maker.

Our problem is also related to the \emph{Bayesian persuasion} framework originally introduced by~\citet{kamenica2011bayesian}, where an informed sender wants to influence the behavior of a self-interested receiver via the strategic provision of information.
\citet{dughmi2019} generalize the classical framework by considering the case in which there are monetary transfers between the sender and the receiver.
%
Our information-selling setting in which the buyer has limited liability generalizes the model of~\citet{dughmi2019} by also introducing buyer's types.
%


Finally, let us remark that our information-selling problem shares critical features with \emph{Bayesian principal-agent problems} (see, \emph{e.g.},~\citep{castiglioni2021Contract,alon2021contracts,alon2022bayesian,guruganesh2021contracts,DBLP:conf/sigecom/CastiglioniM022} for some references).
Indeed, as we show in Section~\ref{sec:principal_agent}, the problem of computing a seller-optimal protocol generalizes particular principal-agent problems in which the agent's action is observable.
Such a connection between the two settings is also demonstrated in terms of results.
In particular, notice that~\citet{castiglioni2021Contract} design bi-criteria approximation algorithms whose guarantees are similar to those provided in this paper.
Moreover, \citet{gan2022optimal} show how to find optimal protocols in generalized principal-agent problems by using a linear relaxation of the principal's optimization problem, which is quadratic.
%
%
%


%% file: content/prelim.tex
\section{Preliminaries}\label{sec:preliminaries}

We study the problem faced by an information holder (\emph{seller}) selling information to a budget-constrained decision maker (\emph{buyer}).
%
The information available to the seller is collectively termed \emph{state of nature} and encoded as an element of a finite set $\Theta \defeq \{\theta_i\}_{i=1}^d$ of $d$ possible states, while the set of the $m$ actions available to the buyer is $\mathcal{A} \defeq \{a_i\}_{i=1}^m$.
The buyer is also characterized by a private \emph{type}, which is unknown to the seller and belongs to a finite set $\mathcal{K} \defeq \{k_i\}_{i=1}^n$ of $n$ possible types.
Each buyer's type $k \in \mathcal{K}$ is characterized by a utility function $u^k_\theta:\mathcal{A} \to [0,1]$ associated to each state $\theta \in \Theta$ and a budget $b_k \in \mathbb{R}_{+}$ representing how much they can afford to pay.
%
%
In our model, the seller's utility is \emph{not} only determined by how much the buyer pays for acquiring information, but it also depends on the buyer's action.
Specifically, for every state $\theta \in \Theta$, the sender gets an additional utility contribution determined by a function $u^s_{\theta}:\mathcal{A} \to [0,1]$. 
We assume that both the seller and the buyer know the probability distribution $\mu \in \Delta_{\Theta}$ according to which the state of nature is drawn, as well as the probability distribution $\lambda \in \Delta_{\mathcal{K}}$ determining the buyer's type.\footnote{In this work, given a finite set $X$, we let $\Delta_{X}$ be the set of all the probability distributions defined over the elements of $X$.}
We let $\mu_\theta$ be the probability assigned to state $\theta \in \Theta$, while $\lambda_k$ is the probability of type $k \in \mathcal{K}$.

As in~\citep{xu2020}, we assume w.l.o.g. that information revelation happens only once during the seller-buyer interaction.
Thus, as it is the case in \emph{Bayesian} persuasion~\citet{kamenica2011bayesian}, the seller reveals information to the buyer by committing to a \emph{signaling scheme} $\phi$, which is a randomized mapping from states of nature to signals being issued to the buyer.
Formally, $\phi : \Theta \to \Delta_{\mathcal{S}}$, where $\mathcal{S}$ is a finite set of signals. We denote by $\phi_\theta \in \Delta_{\mathcal{S}}$ the probability distribution employed when the state of nature is $\theta \in \Theta$, with $\phi_\theta(s)$ being the probability of sending $s \in \mathcal{S}$.

\subsection{Protocols with Menus}\label{sec:preliminaries_menus}

An information-selling \emph{protocol} for the seller is defined as follows.
%
%
The seller first proposes a \emph{menu} of signaling schemes to the buyer, with each signaling scheme being assigned with a \emph{price}.
Then, the buyer chooses a signaling scheme and pays its price upfront, before information is revealed.\footnote{Notice that proposing a menu of signaling schemes is equivalent to asking the buyer to report their type and then choosing a signaling scheme based on that, as it is the case in~\citep{xu2020}.}
The seller also commits to action-dependent payments, which are made by the seller in favor of the buyer after information is revealed and the latter has taken an action.
This is in contrast with what happens in the protocol introduced by~\citet{xu2020}, where there are no action-dependent money transfers.
Intuitively, such payments are needed in order to incentivize the agent to play an action that is profitable for the seller, and, thus, they are \emph{not} needed in the setting of~\cite{xu2020} where the seller's utility function is only determined by how much the buyer pays for acquiring information.
Formally, we define a seller's protocol as follows:
\begin{definition}[Seller's protocol]\label{def:protocol}
	A \emph{protocol} for the seller is a tuple $ \{ (\phi^k, p_k,\pi_k ) \}_{k \in \mathcal{K}} $, where:
	\begin{itemize}
		\item $\{\phi^k\}_{k \in \mathcal{K}}$ is a \emph{menu of signaling schemes} $\phi^k: \Theta \to\Delta_{\mathcal{S}}$, one for each receiver's type $k \in \mathcal{K}$; 
		\item $\{p_k\}_{k \in \mathcal{K}}$ is a \emph{menu of prices}, with $p_k \in \mathbb{R}_+$ representing how much the seller charges the buyer for selecting the signaling scheme $\phi^k$;\footnote{Assuming $p_k \geq 0$ is w.l.o.g., since, intuitively, the seller is never better off paying the buyer before they played any action.}
		%
		\item $\{\pi_k\}_{k \in \mathcal{K}}$ is a \emph{menu of payment functions}, which are defined as $\pi_k: \mathcal{S} \times \mathcal{A} \to \mathbb{R}_+$ with $\pi_k(s,a)$ encoding how much the seller pays the buyer whenever the latter plays action $a \in \mathcal{A}$ after selecting the signaling scheme $\phi^k$ and receiving signal $s \in \mathcal{S}$.\footnote{The assumption that $\pi_k(s,a) \geq 0$ is w.l.o.g., since the buyer does \emph{not} commit to following the protocol, and, thus, $\pi_k(s,a) < 0$ would result in the buyer leaving the protocol without paying after taking an action.}
	\end{itemize}
\end{definition}
%
%

The seller and the buyer interact as follows:
(i) the seller commits to a protocol $ \{ (\phi^k, p_k,\pi_k  ) \}_{k \in \mathcal{K}} $;
(ii) the buyer selects a signaling scheme $\phi^{k}$ and pays $p_{k}$ to the seller (with $k \in \mathcal{K}$ possibly different from their true type);
(iii) the seller observes the realized state of nature $\theta \sim \mu$, draws a signal $s \sim \phi^{k}_\theta$ according to the selected signaling scheme, and communicates $s$ to the
buyer; 
(iv) given the signal $s$, the buyer infers a posterior distribution $\xi^s \in \Delta_\Theta$ over states of nature, where the probability $\xi^s_\theta$ of state $\theta \in \Theta$ is computed with the \emph{Bayes} rule, as follows:
$$
\xi^s_\theta \defeq \frac{\mu_\theta \, \phi^k_{\theta}(s)}{\sum_{\theta' \in \Theta} \mu_{ \theta'} \phi^k_{\theta'}(s)};
$$
(v) given the posterior $\xi^s$, the buyer selects an action $a \in \mathcal{A}$; and
(vi) the seller pays $\pi_{k}(s, a)$ to the buyer.
As in the model by~\citet{xu2020}, we assume that the seller is committed to following the protocol, while the buyer is \emph{not}, \emph{i.e.}, the buyer is free of leaving the interaction at any point.

In step (v), after observing a signal $s \in \mathcal{S}$ and computing the posterior $\xi^s$, the buyer plays a \emph{best response} by choosing an action $a \in \mathcal{A}$ maximizing their expected utility.
Formally:
\begin{definition}[$\epsilon$-Best-response]\label{def:br_set}
	Let $\epsilon \geq 0$.
	Given a signal $s \in \mathcal{S}$, the induced posterior $\xi^s \in \Delta_{\Theta}$, and a payment function $\pi: \mathcal{S} \times \mathcal{A} \to \mathbb{R}_+$, the \emph{$\epsilon$-best-response set} of a buyer of type $k \in \mathcal{K}$ is:
	%
	\[
		\mathcal{B}^{k,\epsilon}_{\xi^s, \pi} \coloneqq \left\{ a \in \A : \sum_{\theta\in\Theta}\xi^s_\theta \, u_\theta^k(a) + \pi(s,a) \ge \max_{a' \in \A} \,\sum_{\theta\in\Theta}\xi^s_\theta \, u_\theta^k(a') + \pi(s,a') - \epsilon \right\}.
	\]
	%
	%
	%
	We let $b_{\xi^s, \pi}^{k,\epsilon} \in \mathcal{B}^{k,\epsilon}_{\xi^s, \pi}$ be an \emph{$\epsilon$-best response} played by the buyer.
	%
	The \emph{best-response set} $\mathcal{B}^{k}_{\xi^s, \pi}$ of a buyer of type $k \in \mathcal{K}$ is defined for $\epsilon = 0$, while $b_{\xi^s, \pi}^{k} \in \mathcal{B}^{k}_{\xi^s, \pi}$ is a \emph{best response} played by the buyer.\footnote{When the buyer is indifferent among multiple best responses (respectively, $\epsilon$-best responses), we always assume that they break ties in favor of the seller, choosing an action in $ \mathcal{B}^k_{\xi^s, \pi}$ (respectively, $ \mathcal{B}^{k,\epsilon}_{\xi^s, \pi}$) maximizing the seller's expected utility.}
	%
	%
\end{definition}
%
%
In the following, we will oftentimes work in the space of the distributions over posteriors.
In that case, given a posterior $\xi \in \Delta_{\Theta}$, we abuse notation and write $\mathcal{B}^k_{\xi, \pi}$, $\mathcal{B}^{k,\epsilon}_{\xi, \pi}$, $b_{\xi, \pi}^{k,\epsilon}$, and $b_{\xi, \pi}^k$.

The seller's goal is to implement an optimal (\emph{i.e.}, utility-maximizing) protocol $ \{ (\phi^k, p_k,\pi_k ) \}_{k \in \mathcal{K}} $.
We focus on seller's protocols that are \emph{incentive compatible} (IC) and \emph{individually rational} (IR).\footnote{By a revelation-principle-style argument (see~\citep{shoham2008multiagent} for some examples), focusing on IC and IR protocols is w.l.o.g. when looking for an optimal protocol.}
Specifically, a seller's protocol is IC if for every pair of buyer's types $k, k' \in \mathcal{K}$:
\begin{equation*}
\sum_{s \in \mathcal{S}}  \sum_{\theta \in \Theta}  \mu_\theta \phi^k_{\theta}(s) \left[ u^k_\theta(b_{\xi^s, \pi_k}^k) + \pi_k(s,b_{\xi^s, \pi_k}^k) \right] - p_k \ge 
 \sum_{s \in \mathcal{S}} \max_{a \in \mathcal{A}} \sum_{\theta \in \Theta} \mu_\theta \phi^{k'}_{\theta}(s) \left[ u^k_\theta(a) + \pi_{k'}(s, a) \right] - p_{k'} ,
\end{equation*}
while it is IR if for every buyers' type $k \in \mathcal{K}$:
\begin{equation*}
\sum_{s \in \mathcal{S}} \sum_{\theta \in \Theta} \mu_\theta \phi^k_{\theta}(s) \left[ u^k_\theta(b_{\xi^s, \pi_k}^k) +  \pi_k(s,b_{\xi^s,\pi_k}^k) \right] - p_k \ge 
 \max_{a \in \mathcal{A}}\sum_{\theta \in \Theta} \mu_\theta  u^k_\theta(a) .
\end{equation*}
Intuitively, an IC protocol incentivizes the buyer to select the signaling scheme $\phi^{k}$ corresponding ot their true type $k \in \mathcal{K}$, while an IR protocol ensures that the buyer gets more utility by acquiring information rather than leaving the protocol before step (ii) and playing an action without information.
Then, the seller's expected utility is computed as follows:
\[
\sum_{k \in \mathcal{K}} \lambda_k \left[ \sum_{s \in\mathcal{S}} \sum_{\theta \in \Theta} \mu_\theta \phi^{k}_\theta(s) \left[ u_\theta^s(b^k_{\xi^s,\pi_k}) -\pi_k(s, b^k_{\xi^s,\pi_k}) \right]+ p_k  \right].
\]

A crucial component of our results is that we can restrict the attention to protocols that are \emph{direct} and \emph{persuasive}.
We say that protocol is {direct} if it uses signaling schemes whose signals correspond to action recommendations for the buyer, namely $\mathcal{S} = \mathcal{A}$, while a direct protocol is said to be {persuasive} whenever playing the recommended action is always a best response for the buyer. 
%

\subsection{Protocols without Menus}\label{sec:preliminaries_nomenus}

In the second part of the paper, we study the case of seller's \emph{protocols without menus}, in which the seller does \emph{not} propose a menu of signaling schemes to the buyer, but they rather commit to a single signaling scheme and a single payment function.\footnote{From the point of view of~\citet{xu2020}, this is equivalent to assuming that there is no type reporting stage.}
%
%
This allows us to simplify the definition of a protocol (see Definition~\ref{def:protocol}), by denoting a seller's protocol without menus as a tuple $(\phi,p,\pi)$, where $\phi: \Theta \to\Delta_{\mathcal{S}}$ is a signaling scheme, $p \in \mathbb{R}_+$ is a price for such a signaling scheme, representing how much the seller charges the buyer to reveal information to them, 
%
%
and $\pi: \mathcal{S} \times \mathcal{A} \to \mathbb{R}_+$ is a payment function.
%
The seller-buyer interaction unfolds as in the general case with menus, but, in this case, step (ii) only involves the payment of price $p \in \mathbb{R}_+$ on buyer's part.

Some of our results on protocols without menus address the special case in which the buyer has \emph{limited liability}, which means that the buyer has no budget, and, thus, the seller cannot charge a price for a signaling scheme upfront.
Formally, this amounts to asking that $b_k = 0$ for all $k \in \K$.
Notice that, while such a special case may seem of scarce appeal for the problem of selling information, it is indeed interesting on its own, as it is similar to the model studied by~\citet{dughmi2019}.
Indeed, our model can be seen as a generalization of the one in~\citep{dughmi2019}, which adds buyer's private types.
Moreover, in the general case in which the buyer has \emph{no} limited liability, our model additionally builds on top of that of~\citet{dughmi2019} by adding the possibility for the seller to ask the buyer a payments before information is revealed.
%

For protocols without menus, IC constraints are \emph{not} needed anymore, while IR constraints are still required in order to ensure that the buyer is incentivized to acquire information from the principal.
%
%
Given a protocol without menus $(\phi,p,\pi)$, only some of the buyer's types are actually incentivized to participate in the protocol, \emph{i.e.}, all the types whose corresponding IR constraint is satisfied.
%
%
Formally, a protocol determines a subset $ \mathcal{R}_{\phi,p,\pi} \subseteq \mathcal{K} $ of buyer's types such that, for every $k \in  \mathcal{R}_{\phi,p,\pi}$, it holds that: (i) a buyer of type $k$ has enough budget to buy information, namely $b_k \ge p $; and (ii) the IR constraint is satisfied for a buyer of type $k$.\footnote{Whenever the expected utility of a buyer's type is the same by participating in the protocol as \emph{not} doing that, we assume that they take the option maximizing the seller's expected utility.}
%
%
In particular, point (ii) can be formally stated by saying that the following condition is satisfied for every $k \in  \mathcal{R}_{\phi,p,\pi}$:
\begin{equation*}
	\sum_{s \in \mathcal{S}} \sum_{\theta \in \Theta} \mu_\theta \phi_{\theta}(s) \left[ u^k_\theta(b_{\xi^s, \pi}^k) +  \pi(s,b_{\xi^s,\pi}^k) \right] - p \ge \max_{a \in \A} \sum_{\theta \in \Theta} \mu_\theta  u^k_\theta(a) .
\end{equation*}
%
Moreover, given a protocol without menus $(\phi, \pi, p)$, the seller's expected utility is given by: 
\begin{equation*}
\sum_{k \in \mathcal{R}_{\phi,p,\pi}} \lambda_k \left[  \sum_{s \in\mathcal{S}} \sum_{\theta \in \Theta} \mu_\theta \phi_\theta(s) \left[ u_\theta^s(b^k_{\xi^s,\pi}) -\pi(s, b^k_{\xi^s,\pi}) \right] + p  \right] + \sum_{k \not \in \mathcal{R}_{\phi,p,\pi}}\lambda_k \sum_{\theta \in \Theta} \mu_\theta  u^k_\theta(b^k_\mu),
\end{equation*}
where $b^k_{\xi} \in \arg \max_{a \in \mathcal{A}} \sum_{\theta \in \Theta} \xi_\theta  u^k_\theta(a)$ is a best response for a buyer's type $k \in \K$ that only considers the posterior $\xi \in \Delta_{\Theta}$, where, as customary, ties are broken in favor of the seller.
Notice that a buyer's type $k \notin \mathcal{R}_{\phi,p,\pi}$ is among those who decide to do \emph{not} acquire information from the seller, and, thus, they play a best response to the probability distribution $\mu$ (instead of a posterior).

Finally, when dealing with protocols without menus, it will be useful to directly work with \emph{distributions over posteriors} induced by signaling schemes, rather than with signaling schemes~\citep{kamenica2011bayesian}.
A signaling scheme $\phi: \Theta \to \Delta_{\mathcal{S}}$ induces a distribution $\gamma$ over $\Delta_{\Theta}$, which has a support $\text{supp}(\gamma) \defeq \left\{ \xi^s \mid s \in \mathcal{S}  \right\}$ and satisfies the following conditions:
\begin{equation}
\label{eq:consistent}
\sum_{\xi \in \text{supp}(\gamma)} \gamma_\xi \, \xi_\theta = \mu_\theta \quad \forall \theta \in \Theta,
\end{equation}
where $\gamma_\xi \in [0,1]$ is the probability that $\gamma$ assigns to the posterior $\xi \in \text{supp}(\gamma)$.
Thus, instead of working with signaling schemes $\phi$, one can w.l.o.g. work with distributions $\gamma$ over $\Delta_{\Theta}$ that are \emph{consistent} with the probability distribution $\mu$, \emph{i.e.}, they satisfy the condition in Equation~\eqref{eq:consistent}.

When working with distributions over posteriors $\gamma$ rather than with signaling schemes $\phi$, with a slight abuse of notation, we denote a seller's protocol without menus as $(\gamma,p,\pi)$, by identifying a signaling scheme with its induced distribution over posteriors $\gamma$.
Similarly, we slightly abuse notation in payment functions, by assuming that they are defined over posteriors rather than signals.
Formally, we let $\pi: \Delta_{\Theta} \times \A \to \mathbb{R}_+$, with $\pi(\xi,a)$ denoting how much the buyer pays back the seller when the induced posterior is $\xi \in \Delta_{\Theta}$ and they play action $a \in \A$.

%% file: content/type_reporting.tex
\section{Computing a Seller-optimal Protocol with Menus}\label{sec:protocol_selection}
We begin by studying the problem of computing a seller-optimal protocol in which the seller has the ability of proposing a menu of signaling schemes and payment functions to the buyer.
Formally, the problem of computing an optimal IC and IR protocol with menus can be formulated as follows:
\begin{subequations}\label{eqn:LP_type_reporting1}
	\begin{align}
	 \sup_{\substack{\phi_\theta^k(s) \geq 0 \\ p_k \geq 0 \\ \pi_k(s,a) \ge 0}} & \,\, \sum_{k \in \mathcal{K}} \lambda_k \left[ \sum_{s \in\mathcal{S}} \sum_{\theta \in \Theta} \mu_\theta \phi^{k}_\theta(s) \left[ u_\theta^s(b^k_{\xi^s,\pi_k}) -\pi_k(s, b^k_{\xi^s,\pi_k}) \right]+ p_k  \right] \quad \quad \textnormal{s.t.}\quad\quad\quad\quad\quad\quad\,\,\,\, \label{opt_reporting_c1} \\
	& \sum_{s \in \mathcal{S}}  \sum_{\theta \in \Theta}  \mu_\theta \phi^k_{\theta}(s) \left[ u^k_\theta(b_{\xi^s, \pi_k}^k) + \pi_k(s,b_{\xi^s, \pi_k}^k) \right] - p_k \nonumber\\
	& \specialcell{\quad \quad \quad  \quad \quad \ge 
	\sum_{s \in \mathcal{S}} \max_{a \in \mathcal{A}} \sum_{\theta \in \Theta} \mu_\theta \phi^{k'}_{\theta}(s) \left[ u^k_\theta(a) + \pi_{k'}(s, a) \right] - p_{k'}  \hfill   \forall k \in \mathcal{K}, \forall k' \in \mathcal{K}} \label{opt_reporting_c2} \\
	& \specialcell{\sum_{s \in \mathcal{S}} \sum_{\theta \in \Theta} \mu_\theta \phi^k_{\theta}(s) \left[ u^k_\theta(b_{\xi^s, \pi_k}^k) +  \pi_k(s,b_{\xi^s,\pi_k}^k) \right] - p_k \ge 
	 \max_{a \in \mathcal{A}}\sum_{\theta \in \Theta} \mu_\theta  u^k_\theta(a) 
	\hfill \forall k \in \mathcal{K}}\label{opt_reporting_c3}\\
	& \specialcell{\sum_{s \in \mathcal{S}} \phi^k_{\theta}(s) = 1
		\hfill \forall k \in \mathcal{K},\forall \theta \in \Theta.}\label{opt_reporting_c4}
	\end{align}
\end{subequations}
Notice that Problem~\eqref{eqn:LP_type_reporting1} is defined in terms of $\sup$ rather than $\max$ since, as it is the case in principal-agent problems (see, \emph{e.g.},~\citep{Castiglioni2022Randomized,gan2022optimal}), it is \emph{not} in general immediate to establish whether the seller's optimization problem always admits a maximum or not.
Indeed, in the following we show that our problem always admits a maximum.

%
%
As a first step, we prove that we can focus w.l.o.g on protocols which are direct and persuasive.
%
%
%
\begin{restatable}{lemma}{typerepfirst}\label{lem:typerep1}
	%
	Given any IC and IR seller's protocol, it is always possible to recover an IC and IR seller's protocol that is direct and persuasive, and it provides the seller with the same expected utility.
\end{restatable}
Intuitively, Lemma~\ref{lem:typerep1} follows from the fact that, given any signaling scheme $\phi^{k}$ and price function $\pi_{k}$ corresponding to some type $k \in \mathcal{K}$, if two signals induce the same best response for a buyer of type $k$, then it is possible to merge the two signals in a single one, recovering a new signaling scheme and a new price function for type $k$ that achieve the same seller's expected utility.
%
By doing such a procedure for every buyer's type until there are no two signals inducing the same best response for that type, we obtain a protocol that is direct and persuasive, and it has the same seller's expected utility as the original protocol.
%
%
Notice that, since in direct protocols it holds $\S = \A$, whenever we write $\pi_k(a,a')$ for $a,\ a' \in \A$, the first action $a$ is the seller's recommendation (signal), while the second action $a'$ is the one actually played by the buyer.

As a second crucial step, we exploit Lemma~\ref{lem:typerep1} in order to show that, given an IC and IR protocol that is direct and persuasive, there exists another IC and IR protocol which is still direct and persuasive, it achieves the same seller's expected utility, and it is such that: (i) for every $k \in \K$, the price $p_k$ of $\phi^k$ is equal to entire budget $b_k$ of a buyer of type $k$, and (ii) the buyer is \emph{not} paid back (\emph{i.e.}, they get a null payment) if they deviate from the seller's action recommendation.
Formally:
\begin{restatable}{lemma}{typerepsecond}\label{lem:typerep2}
	Given an IC and IR protocol $ \{ (\phi^k, p_k,\pi_k ) \}_{k \in \mathcal{K}} $ that is direct and persuasive, it is always possible to recover an IC and IR protocol $ \{( \phi^k, \tilde p_k,\tilde \pi_k ) \}_{k \in \mathcal{K}} $ such that: it is direct and persuasive, it provides the same seller's expected utility as the original protocol, and, for every buyers' type $k \in \mathcal{K}$, it satisfies $\tilde p_k=b_k$ and $\tilde \pi_k(a,a')=0$ for all $a \neq a' \in \mathcal{A}$.
\end{restatable}
As a direct consequence of Lemma~\ref{lem:typerep2}, we can compactly denote $\pi_{k}(a,a)$ as $\pi_{k}(a)$ for every $a \in \A$, since we can focus w.l.o.g. on payment functions such that $\pi_{k}(a,a')=0$ for all $a \not = a'$.

We are now ready to introduce an LP with polynomially-many variables and constraints that is a linear relaxation of Problem~\eqref{eqn:LP_type_reporting1}.  
%
In order to formulate the LP, we exploit Lemmas~\ref{lem:typerep1}~and~\ref{lem:typerep2} to restrict the attention to direct and persuasive protocols, prices such that $p_k=b_k$ for every $k \in \K$, and payments such that $\pi_{k}(a,a')=0$ for every $k \in \K$ and $a \neq a' \in \A$.
Moreover, we encode the terms $\sum_{\theta \in \Theta} \mu_\theta \phi^k_\theta(a) \pi_k(a)$ as single variables $l_k(a)$.
Then, the LP reads as follows:
\begin{subequations}\label{fig:lp_reporting}
	\begin{align}
	\max_{\substack{ \phi^{k}_\theta(a) \geq 0 \\  l_k(a) \geq 0 \\y_{k,k',a} \geq 0 }} &  \,\,\sum_{k \in \mathcal{K}} \lambda_k \sum_{a \in\mathcal{A}} \left[  \sum_{\theta \in \Theta} \mu_\theta \phi^{k}_\theta(a) u_\theta^s(a) -  l_k(a) \right] + b_k \quad \textnormal{s.t.} \quad\quad\quad\quad\quad\quad\quad\quad\quad\quad\quad\quad\quad\,\,\,\,\\
	&\specialcell{ \sum_{a \in \mathcal{A}}  \left[  \sum_{\theta \in \Theta}  \mu_\theta \phi^k_{\theta}(a)  u^k_\theta(a) +  l_k(a) \right]  - b_k \ge \sum_{a \in \mathcal{A}} y_{k,k',a} -b_{k'} \hfill \forall k \in \mathcal{K}, \forall k' \in \mathcal{K}} \label{lp_reporting_c1} \\
	& \specialcell{y_{k,k',a} \geq \sum_{\theta \in \Theta} \mu_\theta \phi_\theta^{k'}(a) u^k_\theta(a) + l_{k'}(a) \hfill \forall k \in \mathcal{K}, \forall k' \in \mathcal{K},\forall a \in \mathcal{A} }\label{lp_reporting_c2} \\
	& \specialcell{y_{k,k',a} \geq \sum_{\theta \in \Theta} \mu_\theta \phi_\theta^{k'}(a) u^k_\theta(a') \hfill \forall k \in \mathcal{K}, \forall k' \in \mathcal{K},\forall a \neq a' \in \mathcal{A}} \label{lp_reporting_c3}\\
	& \specialcell{\sum_{a \in \mathcal{A}} \left[ \sum_{\theta \in \Theta} \mu_\theta \phi^k_{\theta}(a) u^k_\theta(a) +  l_k(a) \right] - b_k \ge \sum_{\theta \in \Theta} \mu_\theta  u^k_\theta(a') \hfill \forall k \in \mathcal{K}, \forall a' \in \mathcal{A} }\label{lp_reporting_c4}\\
	&\specialcell{ \sum_{\theta \in \Theta}  \mu_\theta \phi^k_{\theta}(a)  u^k_\theta(a) +  l_k(a) \geq \sum_{\theta \in \Theta}  \mu_\theta \phi^k_{\theta}(a)  u^k_\theta(a') \hfill \forall  k \in \mathcal{K}, \forall a \neq a' \in \mathcal{A} } \label{lp_reporting_c5} \\
	& \specialcell{  \sum_{a \in \mathcal{A}} \phi_\theta^k(a) = 1 \hfill \forall k \in \mathcal{K}, \forall \theta \in \Theta} \label{lp_reporting_c6}.
	\end{align}
\end{subequations}
In LP~\eqref{fig:lp_reporting}, Constraints~\eqref{lp_reporting_c1}--\eqref{lp_reporting_c3} ensure that the protocol is IC, Constraints~\eqref{lp_reporting_c4} enforce that it is IR, while Constraints~\eqref{lp_reporting_c5} guarantee that the protocol is persuasive.
%

Given how LP~\eqref{fig:lp_reporting} is obtained from Problem~\eqref{eqn:LP_type_reporting1}, it is \emph{not} immediately clear how, given a feasible solution to LP~\eqref{fig:lp_reporting}, one can recover a protocol that is a solution to Problem~\eqref{eqn:LP_type_reporting1} with seller's expected utility equal to the value of the solution to LP~\eqref{fig:lp_reporting}.
Indeed, in a solution to LP~\eqref{fig:lp_reporting}, a variable $l_k(a)$ could be strictly positive even when the variables $\phi^k_\theta(a)$ are equal to zero.
In such a case, it is \emph{not} possible to immediately recover a value for $\pi_k(a)$ starting from a solution to LP~\eqref{fig:lp_reporting}, since $l_k(a)$ encodes $\sum_{\theta \in \Theta} \mu_\theta \phi^k_\theta(a) \pi_k(a)$, from which computing $\pi_k(a)$ would require a division by zero.

In the following, we show how, given an optimal solution to LP~\eqref{fig:lp_reporting}, it is indeed possible to build in polynomial time a seller-optimal protocol with menus.
%
%
First, we prove a preliminary result:
%
%
\begin{restatable}{lemma}{typerethird}\label{lem:typerep3}
The optimal value of LP~\eqref{fig:lp_reporting} is at least as large as the supremum in Problem~\eqref{eqn:LP_type_reporting1}.
\end{restatable}

Then, we show that, given a solution to LP~\eqref{fig:lp_reporting}, it is possible to recover in polynomial time an IC and IR protocol with at least the same value.
Formally:
%
\begin{restatable}{lemma}{typerefourth}\label{lem:typerep4}
	Given a feasible solution to LP~\eqref{fig:lp_reporting}, it is possible to recover in polynomial time an IC and IR protocol whose seller's expected utility is greater than or equal to the value of the solution to LP~\eqref{fig:lp_reporting}. 
	%
\end{restatable}
%
Intuitively, Lemma~\ref{lem:typerep4} is proved by showing that, given a feasible solution to LP~\eqref{fig:lp_reporting}, it is possible to efficiently construct a new solution in which, whenever some variable $l_k(a)>0$, then there exists at least one state of nature $\theta \in \Theta$ for which $\phi_\theta^k(a) > 0$, \emph{i.e.}, action $a$ is recommended with strictly positive probability.
Moreover, such a procedure does \emph{not} detriment the objective function value and retains the IC and IR conditions.
Then, from the new solution, one can recover a protocol that is a valid solution to Problem~\eqref{eqn:LP_type_reporting1}, by letting $\pi_k(a) = l_k(a) / \sum_{\theta \in \Theta} \mu_\theta \phi^k_\theta(a)$ for all $k \in \K$ and $a \in \A$.

Finally, by exploiting Lemmas~\ref{lem:typerep3}~and~\ref{lem:typerep4}, we can design a polynomial-time algorithm that finds a seller-optimal protocol with menus.
Indeed, the algorithm can simply optimally solve LP~\eqref{fig:lp_reporting} (in polynomial time), and use Lemma~\ref{lem:typerep4} to recover an IC and IR protocol having at least the same value.
Tanks to Lemma~\ref{lem:typerep3}, such a protocol is optimal for the seller.
%
%
\begin{restatable}{theorem}{typerepthm}\label{thm:typerep}
	There exists a polynomial-time algorithm that computes a protocol with menus that maximizes the seller's expected utility.
\end{restatable}
Theorem~\ref{thm:typerep} also shows as a byproduct that Problem~\eqref{eqn:LP_type_reporting1} always admits a maximum.

Let us remark that the idea of formulating a linear relaxation of a quadratic problem by introducing a new variable has already been used in generalized principal-agent problems by~\citet{gan2022optimal}.
However, in such a setting, the linear relaxation cannot be used to solve the principal's optimization problem exactly, but only to recover a desirable approximation of an optimal solution.
This is because the problem may \emph{not} admit a maximum, as shown by~\citet{Castiglioni2022Randomized} even in the special case of hidden-action principal-agent problems.
Surprisingly, in our information-selling setting, the linear relaxation can be used to find an (exact) optimal solution.
Intuitively, this is possible since, in our setting, the seller observes the action undertaken by the buyer, while in hidden-action principal-agent problems the principal does \emph{not} directly observe the agent's action.
%

%% file: content/principal_agent.tex
\section{Drawing a Connection with Principal-agent Problems}\label{sec:principal_agent}

In this section, we show that our information-selling problem is intimately related to a particular class of principal-agent problems.
Specifically, we show that the problem of computing a seller-optimal protocol without menus is a generalization of the problem of computing an optimal contract in principal-agent problems in which the principal observes the action undertaken by the agent.
%

In Section~\ref{sec:principal_agent_def}, we formally introduce principal-agent problems with observable actions.
Then, in Section~\ref{sec:principal_agent_res}, we show how such problems are related to our information-selling problem, and we prove an hardness result for them which carries over to our problem
Finally, in Section~\ref{sec:principal_agent_add}, we provide some preliminary technical results that will be useful in the following sections.

\subsection{Principal-agent Problem with Observable Actions}\label{sec:principal_agent_def}

We start by formally defining an instance of \emph{(Bayesian) observable-action principal-agent problem}.\footnote{Notice that observable-action principal-agent problems are a special case of \emph{Bayesian} hidden-action principal-agent problems. Indeed, this can be easily seen by taking an instance of the hidden-action problem in which outcomes correspond one-to-one with agent's actions, and each action deterministically determines its corresponding outcome.}
%
%
For ease of exposition, we reuse some of the notation already introduced in Section~\ref{sec:preliminaries}, in order to denote elements that in observable-action principal-agent problems have the same role as in our information-selling setting.
The \emph{agent} has a finite set $\K$ of possible \emph{types}, and a type $k \in \K$ is drawn with probability $\lambda_k$ according to a known distribution $\lambda \in \Delta_{\K}$.
Each agent's type $k \in \K$ has a set $\A$ of \emph{actions}, with each action having a type-dependent \emph{cost} $c^k_a \in [0,1]$.
The \emph{principal} is characterized by a \emph{reward} $r_a \in [0,1]$ for every agent's action $a \in \A$.
Moreover, the principal can commit to a \emph{contract}, which can be encoded by a function $\pi : \A \to \mathbb{R}_+$ defining a payment $\pi(a)$ from the principal to the agent for every possible agent's action $a \in \A$.
Given a contract, an agent of type $k \in \K$ plays a \emph{best response} $b^k_\pi \in \A$, defined as $b^k_\pi \in  \arg \max_{a \in \A} \left\{ \pi(a)-c^k_a \right\}$, where, as usual, we assume that ties are broken in favor of the principal.
Finally, the principal's goal is to commit to a contract maximizing their expected utility, which is defined as $\sum_{k \in \K} \lambda_k [ r_{b^k_\pi} -\pi(b^k_\pi) ] $.
%
%

\subsection{From Selling Information to Observable-action Principal-agent Problems}\label{sec:principal_agent_res}

Next, we show that our information-selling problem in the case in which protocols are without menus and the buyer has limited liability (\emph{i.e.}, $b_k = 0$ for all $k \in \K$) is strongly related to the problem of finding an optimal (\emph{i.e.}, expected-utility-maximizing) contract in observable-action principal-agent problems.
%
%
Specifically, we show that, given a posterior $\xi \in \Delta_{\Theta}$,
%
designing a payment function $\pi: \Delta_{\Theta} \times \A \to \mathbb{R}_+$ that maximizes the seller's expected utility conditioned on the fact that the induced posterior is $\xi$ is equivalent to finding an optimal contract in a suitably-defined principal-agent problem with observable actions.
Formally, for ease of presentation, we introduce the following notion of payment function that is optimal for the seller in a given posterior:
%

\begin{definition}[Optimal payment function in a posterior]
	Given a posterior $\xi \in \Delta_\Theta$, we say that a payment function $\pi: \Delta_{\Theta} \times \A \to  \mathbb{R}_{+}$ is \emph{optimal} in $\xi$ if the following holds:
	\begin{equation}\label{prob:payment_func}
		\pi \in \argmax_{\pi'} \,\, \sum_{k \in \K} \lambda_k \left[ \sum_{\theta \in \Theta} \xi_\theta \,  u^s_\theta(b^k_{\xi,\pi'})-\pi(\xi,b^k_{\xi,\pi'}) \right].
	\end{equation}
%
\end{definition}
%
Notice that, in Problem~\eqref{prob:payment_func}, the price $p$ of the signaling scheme $\phi$ does \emph{not} appear in the seller's expected utility, since we are restricted to settings in which the buyer has limited liability, and, thus, it is always the case that $p=0$.
For the same reason, we can safely assume that all the buyer's types satisfy IR constraints.
Then, we can state the following crucial result:
\begin{restatable}{lemma}{lemmaprincipal}\label{lem:prin2sing}
	%
	%
	%
	Given a posterior $\xi \in \Delta_{\Theta}$, solving Problem~\eqref{prob:payment_func} is equivalent to computing a contract maximizing the principal's expected utility in an instance of observable-action principal-agent problem such that, for every agent's type $k \in \K$ and action $a \in \A$, the following holds:
	%
	%
	\[
		c^k_a=\sum_{\theta \in \Theta} \xi_\theta \left[  u^k_\theta(b_\xi^k) -u^k_\theta(a) \right] \quad \text{and} \quad
		r_a = \sum_{\theta \in \Theta} \xi_\theta \, u^s_\theta(a).
	\]
	%
	Moreover, finding an optimal contract in any instance of observable-action principal-agent problem can be reduced in polynomial time to computing a seller-optimal protocol without menus in a problem instance in which the buyer has limited liability and there is only one state of nature.
	%
	%
\end{restatable}
The first statement in Lemma~\ref{lem:prin2sing} implies that, given an instance of our information-selling problem in which the buyer has limited liability and there is only one state of nature, it is possible to compute a seller-optimal protocol without menus by finding an optimal contract in an instance of observable-action principal-agent problem defined as in the lemma (notice that such an instance can be easily built in polynomial time).
%
%
Thus, by Lemma~\ref{lem:prin2sing}, we can easily prove the following:
\begin{theorem}\label{thm:equiv}
	Restricted to instances in which the buyer has limited liability and there is only one state of nature, computing a seller-optimal protocol without menus is equivalent to the problem of finding an optimal contract in general instances of the observable-action principal-agent problem.
\end{theorem}

While the computational complexity of finding optimal contracts in hidden-action principal-agent problems is well understood (see, \emph{e.g.},~\citep{castiglioni2021Contract}), to the best of our knowledge, there are no results on problems with observable actions.
In following theorem, we prove a strong hardness result for them: there exists a constant $\alpha<1$ such that designing a contract which provides the principal with at least an $\alpha$ fraction of the expected utility in an optimal contract is computationally intractable.
%
Formally:
%
\begin{restatable}{theorem}{theoremHardnessState}\label{thm:hardness_principla}
	In observable-action principal-agent problems, the problem of computing a contract maximizing the principal's expected utility is $\mathsf{APX}$-hard.
\end{restatable}
Then, Theorem~\ref{thm:equiv} immediately gives the following result:
\begin{corollary}\label{cor:hardness}
	The problem of computing a seller-optimal protocol without menus is $\mathsf{APX}$-hard, even when the buyer has limited liability and the number of states of nature $d$ is fixed.
\end{corollary}
As we show in the following sections (see Theorems~\ref{thm:fixedStates}~and~\ref{thm:fixedactions2}), whenever the number of states of nature is fixed, the problem of computing a seller-optimal protocol without menus admits a polynomial-time algorithm providing a particular \emph{bi-criteria} approximation of the seller's expected utility in an optimal protocol.
Such an approximation is similar to the the bi-criteria guarantees provided by~\citet{castiglioni2021Contract} for \emph{Bayesian} hidden-action principal-agent problems.
By Theorem~\ref{thm:equiv}, our polynomial-time bi-criteria approximation algorithm for the setting in which the buyer has limited liability (Theorem~\ref{thm:fixedStates}) can be easily adapted to work with observable-action principal-agent problems. 
Theorem~7 in~\citep{castiglioni2021Contract} shows that, for hidden-action problems, such bi-criteria approximations are tight.
We leave as an open problem to establish whether these are also tight in our observable-action principal-agent problems or one can obtain better guarantees in polynomial time for our specific case.
%
%

\subsection{Additional Preliminary Technical Results}\label{sec:principal_agent_add}

We conclude the section by recalling two already-known results on hidden-action principal-agent problems.
Clearly, these also hold for the specific case of observable-action principal-agent problems.
%
%
By Theorem~\ref{thm:equiv}, such results can be easily cast to our information-selling problem.
Indeed, we also show that one of them can be strengthen in our setting.
%

The first result that we are going to introduce makes use of {linear contracts}, which are payment schemes that pay the agent a given fraction of the principal's reward.
Formally, in observable-action principal-agent problems, a contract $\pi: \A \to \mathbb{R}_+$ is said to be \emph{linear} if there exists a $\beta \in [0,1]$ such that $\pi(a) =\beta \, r_a $ for all $ a \in \A$.
%
%
%
Despite their simplicity, linear contracts provide good approximations with respect to general ones.
In particular, the following holds:
%
\begin{theorem} [Essentially Theorem~3 by~\citet{castiglioni2021Contract}] \label{thm:linear}
	In an observable-action principal-agent problem, for any $\rho \in (0, \nicefrac{1}{2}]$, there exists a linear contract $\pi :\A \to \mathbb{R}_{+}$ such that:
	\[
		\sum_{k \in \K} \lambda_k  \left[ r_{b^k_{\pi}}-\pi(b^k_{\pi}) \right] \ge  \rho \, \max_{\pi'} \sum_{k \in \K} \lambda_k \left[ r_{b^k_{\pi'}}-\pi'(b^k_{\pi'}) \right] -2^{\Omega(1/\rho)}. 
	\]
	Moreover, such a linear contract is defined by a parameter $\beta=1-2^{-i}$, for some $i \in \{1, \dots,\lfloor\nicefrac{1}{2\rho}\rfloor\}$. 
\end{theorem}
We will make use of a stronger version of Theorem~\ref{thm:linear}, which applies to our setting and directly follows from the analysis of~\citet{castiglioni2021Contract} and Lemma~\ref{lem:prin2sing}.
Formally:
\begin{restatable}{corollary}{corolPD} \label{cor:linear}
	%
	Given a posterior $\xi \in \Delta_{\Theta}$, for any $\rho \in (0, \nicefrac{1}{2}]$, 
	there exists a payment function $\pi :\Delta_{\Theta} \times \A \to \mathbb{R}_{+}$ such that $\pi(\xi,a) = \beta \, \sum_{\theta \in \Theta} \xi_\theta \, u^s_\theta(a)$ for every $a \in \A$, where $\beta \in [0,1]$ is an (action-independent) parameter, and, additionally, the following holds:
	\begin{align*}
		\sum_{k \in \K} \lambda_k  \sum_{\theta \in \Theta} \xi_\theta& \left[ u^s_{\theta}( b^k_{\xi,\pi})-\pi(s,b^k_{\xi,\pi})\right] \ge \rho  \sum_{k \in \K} \lambda_k \max_{a \in A} \sum_{\theta \in \Theta}\xi_\theta  \left[ u^s_\theta(a) +  u^k_{\theta}(a)-u^k_{\theta}(b^k_{\xi})\right]-2^{\Omega(1/\rho)}.
	\end{align*}
	Moreover, such a parameter $\beta$ is equal to $1-2^{-i}$ for some $i \in \{1, \dots,\lfloor\nicefrac{1}{2\rho}\rfloor\}$. 
\end{restatable}

Finally, we recall a useful result that establishes a connection between agent's best responses and \emph{approximate} best responses in principal-agent problems.
Intuitively, such a result states that, given a contract under which the agent is allowed to play an $\epsilon$-best response (for some $\epsilon \geq 0$), it is always possible to recover a new contract in which the agent must play an (exact) best response, by only incurring in a small loss in the principal's expected utility.
Formally, given $\epsilon \geq 0$ and a contract $\pi :\A \to \mathbb{R}_{+}$, for every $k \in \K$, we let $\mathcal{B}^{k,\epsilon}_{\pi} \subseteq \A$ be the set of \emph{$\epsilon$-best-response} actions for an agent of type $k$.
Such a set is made by all the actions $a \in \A$ such that $\pi(a)-c^k_a\ge \max_{a' \in \A} \left\{ \pi(a')-c^k_{a'} \right\} - \epsilon$.
We denote by $b^{k,\epsilon}_\pi\in\mathcal{B}^{k,\epsilon}_{\pi} $ an $\epsilon$-best-response action that is actually played by an agent of type $k$, assuming that ties are broken in favor of the principal, as usual.
%
%
Then:
%
%
\begin{theorem}[Essentially Proposition~A.4 by~\citet{dutting2021complexity}] \label{thm:dutting2}
	Given $\epsilon \geq 0$, an instance of observable-action principal-agent problem and, and a contract $\pi :\A \to \mathbb{R}_{+}$, 
	%
	%
	there exists a contract $\pi' :\A \to \mathbb{R}_{+}$ such that $\pi'(a)=(1-\sqrt{\epsilon}) \, \pi(a)+\sqrt{\epsilon} \, r_a$ for every $a \in \A$, and the following holds:
	\[
		\sum_{k \in \K} \lambda_k \left[ r_{b^{k}_{\pi'}}- \pi'(b^{k}_{\pi'})\right]\ge \sum_{k \in \K} \lambda_k \left[ r_{b^{k,\epsilon}_{\pi}}- \pi(b^{k,\epsilon}_{\pi})\right]-2\sqrt{\eps}.
	\]
	%
	%
\end{theorem}
Theorem~\ref{thm:dutting2} can be easily cast to our setting by means of Lemma~\ref{lem:prin2sing}. Formally:
\begin{corollary} \label{thm:dutting}
	%
	Given $\epsilon \geq 0$, a posterior $\xi \in \Delta_{\Theta}$, and a payment a function $\pi: \Delta_{\Theta} \times \A \to \mathbb{R}_{+}$,
	%
	%
	there exists a payment function $\pi' : \Delta_{\Theta} \times \A \to \mathbb{R}_{+}$ such that  $\pi'(\xi,a)=(1-\sqrt{\epsilon}) \, \pi(\xi,a)+\sqrt{\epsilon} \, \sum_{\theta\in \Theta} \xi_{\theta} u^s_\theta(a)$ for every $a \in \A$, and the following holds:
	 \[ 
	 \sum_{k \in \K} \lambda_k \left[\sum_{\theta\in \Theta} \xi_\theta \, u^s_\theta(b^{k}_{\xi,\pi'})- \pi'(\xi,b^{k}_{\xi,\pi'})\right]\ge \sum_{k \in \K} \lambda_k \left[\sum_{\theta\in \Theta} \xi_\theta \, u^s_\theta(b^{k,\epsilon}_{\xi,\pi})- \pi(\xi,b^{k,\epsilon}_{\xi,\pi})\right]-2\sqrt{\eps}.
	\]
	%
\end{corollary}
Corollary~\ref{thm:dutting} will be crucial to provide our results in the following sections.

%% file: content/const_actions.tex
\section{Computing a Seller-optimal Protocol without Menus:\\The Case of a Buyer with Limited Liability}\label{sec:selection_limited_liability}

In this section, we study the problem of computing a seller-optimal protocol without menus when the buyer has limited liability, \emph{i.e.}, each buyer's types $k \in \K$ has budget $b_k=0$.
As remarked in Section~\ref{sec:preliminaries_nomenus}, such a setting is of interest on its own, since it is a generalization of the one addressed by~\citet{dughmi2019}.
Moreover, the technical results derived in this section will be useful to deal with the general problem in which the buyer has \emph{no} limited liability.
%
%
We show how to circumvent the $\mathsf{APX}$-hardness result that we established in Corollary~\ref{cor:hardness}, first, in Section~\ref{sec:fixing_actions}, by fixing the number of buyer's actions, and then, in Section~\ref{sec:const_states_first}, by fixing the number of states of nature.

%
In this section, since the buyer has limited liability, we can assume w.l.o.g. that $p = 0$, so that we can compactly denote a protocol with a pair $(\gamma,\pi)$, rather than with $(\gamma,p,\pi)$.
%

\subsection{Fixing the Number of Buyer's Actions}\label{sec:fixing_actions}

First, we show that, whenever the buyer has limited liability and the number of buyer's actions $m$ is fixed, the problem of computing a seller-optimal protocol without menus admits a PTAS, \emph{i.e.}, we can design a protocol whose seller's expected utility is arbitrarily close to that of an optimal protocol in time polynomial in the instance size.

In order to design our PTAS, we start by observing that, since $p=0$, the IR constraints are satisfied by all the protocols $(\gamma,\pi)$.
This allows us to formulate the problem of computing a seller-optimal protocol without menus as the following optimization problem:\footnote{Notice that, as it is the case for Problem~\eqref{eqn:LP_type_reporting1} in Section~\ref{sec:protocol_selection}, it is \emph{not} immediately clear \emph{a priori} whether the problem of computing a seller-optimal protocol without menus admits a maximum or \emph{not}. Thus, in principle we should start by defining the problem with a $\sup$ rather than a $\max$. However, in Section~\ref{sec:fixing_types}, we provide a (possibly exponential-time) algorithm which finds a seller-optimal protocol without menus in general settings, and this implies that a maximum always exists.}
%
\begin{subequations}\label{eq:consAct}
\begin{align}
\max_{\substack{\gamma_\xi \ge 0 \\ \pi(\xi,a) \ge 0}} & \,\ \sum_{k \in \K}  \lambda_k   \sum_{\xi \in \text{supp}(\gamma)} \gamma_\xi \left[ \sum_{\theta \in \Theta} \xi_\theta \, u^s_\theta(b^{k}_{ \xi,\pi}) - \pi(\xi, b^{k}_{ \xi,\pi})\right] \quad \text{s.t.} \label{eq:utility}\\
& \sum_{ \xi \in \text{supp}(\gamma)} \gamma_\xi \, \xi_\theta = \mu_\theta & \forall \theta \in \Theta.
\end{align}
\end{subequations}
Notice that Problem~\eqref{eq:consAct} is defined over general distributions over posteriors $\gamma$, whose support $\text{supp}(\gamma)$ may be \emph{not} finite.
Thus, as we show in the following, the crucial result that we need to design a PTAS is the possibility of restricting the attention to finite sets of posteriors.

%
We need to introduce a particular class of posteriors, which are called \emph{$q$-uniform posteriors}.
\begin{definition}[$q$-Uniform posterior]
	A posterior $\xi \in \Delta_{\Theta}$ is \emph{$q$-uniform} if it can
	be obtained by averaging the elements of a multi-set defined by $q \in \mathbb{N}_{>0}$ canonical basis vectors of $\mathbb{R}^d$.
\end{definition}
In the following, we denote by $\Xi_q \subseteq \Delta_{\Theta}$ (for a given $q \in \mathbb{N}_{>0}$) the finite set of all the $q$-uniform posteriors.
As it is easy to check, such a set satisfies $|\Xi_q|\le \min\{d^q, q^d\}$.

In order to derive our PTAS, as a first preliminary result we show that, given any posterior $\xi^* \in \Delta_\Theta$, payment function $\pi: \Delta_\Theta \times \mathcal{A}\to \mathbb{R}_{+} $, and $\epsilon > 0$, there always exists a signaling scheme $\gamma$ supported on $\Xi_q$ which induces posterior $\xi^*$ on average and guarantees a seller’s expected utility close to that provided by the posterior $\xi^*$ (assuming the buyer plays an $\epsilon$-best response).
Formally:
\begin{restatable}{lemma}{constactionsfirst}\label{lem:actions_quniform}
	Given any $\epsilon, \alpha > 0$, a posterior $ \xi^* \in \Delta_{\Theta}$, and a payment function $\pi: \Delta_\Theta \times\mathcal{A} \to \mathbb{R}_{+} $, there always exists a signaling scheme $\gamma \in \Delta_{\Xi_q}$ with $q={2\log(2m/\alpha)}/{\epsilon^2}$ such that:
	\begin{equation*}
		 \sum_{\xi \in \Xi_q} \gamma_\xi \left[ \sum_{\theta \in \Theta}  \xi_\theta \, u^s_\theta(b^{k,\epsilon}_{\xi,\pi}) - \pi(\xi, b^{k,\epsilon}_{\xi,\pi}) \right] \ge \sum_{\theta \in \Theta} \xi^*_\theta \, u^s_\theta(b^k_{\xi^*,\pi}) - \pi(\xi^*, b^k_{\xi^*,\pi}) - \alpha,
	\end{equation*}
	for every buyer's type $k \in \mathcal{K}$, where we let $\pi:\Delta_\Theta \times \mathcal{A} \to \mathbb{R}_{+} $ be a payment function that is optimal in every posterior $\xi \in \Xi_q$ when the buyer plays an $\epsilon$-best response, i.e., $\pi$ solves Problem~\eqref{prob:payment_func} for every $\xi \in \Xi_q$ with $b^{k}_{\xi,\pi}$ replaced by $b^{k,\epsilon}_{\xi,\pi}$.
	%
	%
	Furthermore, the signaling scheme $\gamma$ satisfies:
	\begin{equation*}
		\sum_{\p\in\pset_q} \gamma_\p \, \p_\theta=\p_\theta^\ast  \quad \quad \forall \theta \in \Theta.
	\end{equation*}
\end{restatable}
Lemma~\ref{lem:actions_quniform} guarantees that, by decomposing each posterior $\xi \in \Delta_\Theta$ as a convex combination of the elements of $\Xi_q $, the seller's expected utility decreases by at most $\alpha$.
This implies that, assuming the buyer plays an $\epsilon$-best response, it is possible to work with signaling schemes (and thus payment functions) supported on $\Xi_q$, by only slightly degrading the seller's expected utility. 

Another component that we need for our PTAS is an algorithm that, given a $q$-uniform posterior, computes an optimal payment function in that posterior (\emph{i.e.}, a payment function solving Problem~\eqref{prob:payment_func} for such a posterior).
By Theorem~\ref{thm:equiv}, it is easy to see that such an algorithm has to solve a problem that is equivalent to computing an optimal contract in observable-action principal-agent problems.
Thus, by Theorem~\ref{thm:hardness_principla}, such a problem is $\mathsf{APX}$-hard in general.
%
%
Next, we show that, whenever the number of buyer's actions $m$ is fixed, the $\mathsf{APX}$-hardness result can be circumvented, and, thus, we can provide an algorithm that solves the desired task and runs in polynomial time.
%
Formally:
\begin{restatable}{lemma}{constactionsecond}\label{lem:price_fun}
	Restricted to instances where the buyer has limited liability and the number of buyer's actions $m$ is fixed, there exists a polynomial-time algorithm that, given a posterior $ \xi \in \Delta_\Theta$ as input, computes the payments $\pi(\xi,a)$ for $a \in \A$ of a payment function $\pi : \Delta_\Theta \times \mathcal{A} \to \mathbb{R}_{+}$ optimal in $\xi$.
	%
\end{restatable}
Notice that, in order to get a payment function $\pi:\Delta_\Theta \times \mathcal{A} \to \mathbb{R}_{+} $ that is optimal in every posterior $\xi \in \Xi_q$, it is sufficient to apply Lemma~\ref{lem:price_fun} for each $\xi \in \Xi_q$, then putting together all the computed payments $\pi(\xi,a)$ in order to obtain the overall payment function $\pi$.

The final piece that we need to complete the design of our PTAS is a way of coming back to work with buyer's best responses, rather than using $\epsilon$-best responses.
Indeed, this is possible thanks to Corollary~\ref{thm:dutting}, which allows us to modify the payment function in all the induced posteriors, so as to achieve the desired result by only losing a small amount $2 \sqrt{\eps}$ of the seller's expected utility.
%
%

%
%
%
Now, we are ready to design our PTAS that works whenever the buyer has limited liability and the number of buyer's actions $m$ is fixed.
By Lemma~\ref{lem:actions_quniform}, we can focus on signaling schemes supported over $q$-uniform posteriors, for a suitably-defined $q \in \mathbb{N}_{>0}$.
Moreover, thanks to Lemma~\ref{lem:price_fun}, we can compute a payment function that is optimal in all the $q$-uniform posteriors, by running the polynomial-time algorithm in Lemma~\ref{lem:price_fun} for each $q$-uniform posterior in $\Xi_q$.
%
By Corollary~\ref{thm:dutting}, such an optimal payment function achieves a seller's expected utility that is close to that obtained by a payment function which is optimal in every $q$-uniform posterior when considering $\epsilon$-best responses, thus allowing for the application of the result in Lemma~\ref{lem:actions_quniform}.
In conclusion, our PTAS works by solving a modified version of LP~\eqref{eq:consAct}, where we set $\supp(\gamma) :=\Xi_q$ in Equation~\eqref{eq:utility}, and we take as payment function the one returned by applying Lemma~\ref{lem:price_fun} in each posterior $\xi \in \Xi_q$.
%
%
%
It is easy to see that the overall procedure requires time polynomial in the instance size when the number of actions $m$ is fixed, since $|\Xi_q| \leq d^q$ and $q={2\log(2m/\alpha)}/{\epsilon^2}$ as prescribed by Lemma~\ref{lem:actions_quniform}.
However, the overall running time depends exponentially in $\alpha > 0$, which the seller's expected utility approximation provided by the algorithm.
%
%
%
%
%
%
This allows us to prove the following result:
\begin{restatable}{theorem}{constactionsthmfirst}\label{le:sp}
	Restricted to instances where the buyer has limited liability and the number of buyer's actions $m$ is fixed, the problem of computing a seller-optimal protocol without menus admits a PTAS. 
	%
\end{restatable}
%
%

Finally, we show that a similar approach can be employed to derive a quasi-polynomial time algorithm providing a \emph{bi-criteria} approximation of the seller's expected utility in an optimal protocol, even when the number of buyer's actions $m$ is arbitrary.
Indeed, in our PTAS, the computation of an optimal payment function in a given $q$-uniform posterior can be done in polynomial time only when the number of actions is fixed.
%
%
While in general the problem is $\mathsf{APX}$-hard, an approximately-optimal price function can be computed in polynomial time by applying Corollary~\ref{cor:linear}.
Moreover, since $q={2\log(2m/\alpha)}/{\epsilon^2}$ and $|\Xi_q| \leq d^q$, the enumeration over the $q$-uniform posteriors can be performed in time quasi polynomial in th number of actions $m$.
This gives the following result:
\begin{restatable}{theorem}{QPTAS}\label{QPTAS}
	Restricted to instances in which the buyer has limited liability, there exists an algorithm that, given any $\alpha,\eps>0$ and $ \rho \in (0,\nicefrac{1}{2}]$ as input, returns a protocol without menus achieving a seller's expected utility greater than or equal to $ \rho \, \textnormal{\textsc{opt}} - 2^{-\Omega(1/\rho)} -(\alpha+2\sqrt{\epsilon})$, where $\textnormal{\textsc{opt}}$ is the seller's expected utility in an optimal protocol.
	Moreover, the algorithm runs in time polynomial in $\mathcal{I}^{\log m}$---where $\mathcal{I}$ is the size of the problem instance and $m$ is the number of buyer's actions---,
	%
	and the seller's expected utility in the returned protocol is greater than or equal to $\textnormal{\textsc{optlin}}-(\alpha+2\sqrt{\epsilon})$,
	where $\textnormal{\textsc{optlin}}$ is the best expected utility achieved by a protocol parametrized by $\beta$ as in Corollary~\ref{cor:linear}.
%
\end{restatable}
The second part of the statement will be useful in deriving our results for the problem of computing seller-optimal protocols without menus in the general case in which the buyer has \emph{no} limited liability.
Intuitively, it states that, even if our approximation algorithm only provides a \emph{bi-criteria} approximation of a seller-optimal protocol, the returned protocol achieves a seller's expected utility which is arbitrarily close to that achievable by using payment functions that define the payments as a given fraction of the seller's expected utility.
%

%% file: content/const_states.tex
\subsection{Fixing the Number of States of Nature}\label{sec:const_states_first}

Next, we study the case in which the buyer has limited liability and the number of states of nature $d$ is fixed.
We prove that, in such a setting, it is possible to compute a \emph{bi-criteria} approximation of an optimal protocol without menus similar to that in Theorem~\ref{QPTAS}, but in polynomial time.
Notice that such a result circumvents the $\mathsf{APX}$-hardness one provided in Corollary~\ref{cor:hardness}, as the latter is based on a reduction working with instances with only one state of nature.
%
%
%

Similarly to Section \ref{sec:fixing_actions}, we first show that it is possible to employ signaling schemes supported on the set $\Xi_q$ of $q$-uniform posteriors (for a suitably-defined $q \in \mathbb{N}_{>0}$), by only suffering an arbitrarily small, additive loss in terms of seller's expected utility.
While the following result is similar to the one obtained in Lemma~\ref{lem:actions_quniform}, it is based on different techniques and, in particular, on the fact that the seller's expected utility is Lipschitz continuous in the buyers' posterior.
%
%
Formally:
\begin{restatable}{lemma}{conststatesfirst}\label{lem:states_quniform}
	Given any $\alpha > 0$, a posterior $ \xi^* \in \Delta_{\Theta}$, and a payment function $\pi: \Delta_{\Theta} \times \mathcal{A} \to \mathbb{R}_{+} $ that is optimal in every posterior $\xi \in \Xi_q$ with $q=\lceil {9d}/{\alpha^2} \rceil $, there exists a signaling scheme $\gamma \in \Delta_{\Xi_q}$:
	\begin{equation*}
		\sum_{\xi \in \Xi_q} \gamma_\xi \left[ \sum_{\theta \in \Theta}  \xi_\theta \, u^s_\theta(b^{k}_{\xi,\pi}) - \pi(\xi, \br) \right] \ge \sum_{\theta \in \Theta} \xi^*_\theta u^s_\theta(b^k_{\xi^*,\pi}) - \pi(\xi^*, b^k_{\xi^*,\pi}) - \alpha,
	\end{equation*}
	for every receiver's type $k \in \mathcal{K}$.
	Furthermore, the signaling scheme $\gamma$ satisfies:
	\begin{equation*}
	\sum_{\p\in\pset_q} \gamma_\p \, \p_\theta=\p_\theta^\ast  \quad  \forall \theta \in \Theta.
	\end{equation*}
	%
\end{restatable}

%
Similarly to the case of a fixed number of actions, we employ Lemma~\ref{lem:states_quniform} to restrict the attention to signaling schemes (and thus payment functions) supported on $\Xi_q$.
Moreover, in this case, we can apply Corollary~\ref{cor:linear} in each $q$-uniform posterior in order to compute in polynomial time a payment function that provides a \emph{bi-criteria} approximation of the optimal seller's expected utility in such a posterior.
%
%
Finally, we design an algorithm that solves a modified version of LP~\eqref{eq:consAct}, where we set $\supp(\gamma) := \Xi_q$ in Equation~\eqref{eq:utility}, and we take as payment function the one obtained by putting together those computed by means of Corollary~\ref{cor:linear} for each $\xi \in \Xi_q$. 
Finally, the overall procedure requires polynomial time, since $|\Xi_q| \le q^d $ and the number of states of nature $d$ is fixed, and achieves a \emph{bi-criteria} approximation of the seller's expected utility in an optimal protocol.
%
%
Formally:
\begin{restatable}{theorem}{conststatethmsecond} \label{thm:fixedStates}
	Restricted to instances in which the buyer has limited liability and the number of states of nature $d$ is fixed, there exists an algorithm that, given $\alpha>0$ and $ \rho \in (0,\nicefrac{1}{2}]$ as input, returns in polynomial time a protocol without menus achieving a seller's expected utility greater than or equal to $\rho\, \textnormal{\textsc{opt}}-2^{-\Omega(1/\rho)}-\alpha$, where $\textnormal{\textsc{opt}}$ is the seller's expected utility in an optimal protocol.
	Moreover, the seller's expected utility in the returned protocol is greater than or equal to $\rho \, \textnormal{\textsc{optlin}}-2^{-\Omega(1/\rho)}-\alpha$
	where $\textnormal{\textsc{optlin}}$ is the best expected utility achieved by a protocol parametrized by $\beta$ as in Corollary~\ref{cor:linear}.
	%
\end{restatable}
Similarly to Theorem~\ref{QPTAS}, the second part of the statement will be useful for deriving our results in the general case in which the buyer has \emph{no} limited liability.
%

%% file: content/general.tex
\section{Computing a Seller-optimal Protocol without Menus:\\The General Case}\label{sec:selection_no_limited_liability}

%

We conclude our analysis by considering the problem of computing a seller-optimal protocol without menus in general instances in which the buyer has \emph{no} limited liability.
Thus, in such a setting, the seller also decides a price $p \in \mathbb{R}_+$ for the signaling scheme proposed to the buyer.  
%

First, we provide a negative result for general instances that is stronger than the one established in Corollary~\ref{cor:hardness}.
In particular, the latter result states that the seller's optimization problem is $\mathsf{APX}$-hard even in the special case in which the buyer has limited liability and there is only one state of nature, relying on a reduction employing instances with an arbitrary number of actions $m$.
Indeed, for the specific case in which the buyer has limited liability and the number of actions $m$ is fixed, Theorem~\ref{le:sp} provides a PTAS.
Next, we show that, in general instances where the buyer may \emph{not} have limited liability, the problem is $\mathsf{APX}$-hard even when the number of buyer's actions $m$ is fixed.
%
%
%
To prove such an hardness result, we employ a result by~\citet{Guruswami2009} (see Theorem~\ref{thm:garu} below), which is about the following \emph{promise problem} related to the satisfiability of a fraction of linear equations with rational coefficients and variables restricted to the hypercube.\footnote{In the definition in \citep{Guruswami2009}, the vector ${\hat x}$ can be non-binary. However,~\citet{Guruswami2009} use a binary vector ${\hat x}$ in their proof and, thus, their hardness result also holds for our definition.}
\begin{definition}[\textnormal{\textsc{Lineq-Ma}$(1 -\epsLM, \deltaLM)$} by~\citet{Guruswami2009}]\label{def:lineq_ma}
	For any two constants $\epsLM, \deltaLM \in \mathbb{R}$ satisfying $0 \le \deltaLM \le 1 -\epsLM \le 1$, \textsc{Lineq-Ma}$(1 -\epsLM, \deltaLM)$ is the following promise problem: Given a set of linear equations $Ax=c$ over variables $x \in \mathbb{Q}^{\nVar}$, with coefficients $A \in \mathbb{Q}^{\nEq \times \nVar}$ and $c \in \mathbb{Q}^{\nEq}$, distinguish between the following two cases:
	\begin{itemize}
		\item there exists a vector $ \hat x \in \{0,1\}^{\nVar}$ that satisfies at least a fraction $1 -\epsLM$ of the equations; 
		\item every possible vector $ x \in \mathbb{Q}^{\nVar}$ satisfies less than a fraction $\deltaLM$ of the equations.
	\end{itemize}
\end{definition}
\begin{theorem}[\citet{Guruswami2009}]\label{thm:garu}
	For all the constants $\epsLM, \deltaLM \in \mathbb{R}$ which satisfy $0 \le \deltaLM \le 1 -\epsLM \le 1$, the problem \textnormal{\textsc{Lineq-Ma}$(1-\epsLM,\deltaLM)$} is \NPHard. 
\end{theorem}
Then, Theorem~\ref{def:lineq_ma} allows us to prove the following hardness result:
\begin{restatable}{theorem}{hardnessfirst}\label{thm:hardness1}
	The problem of computing a seller-optimal protocol without menus is $\mathsf{APX}$-hard, even when the number of buyer's actions $m$ is fixed.
	%
\end{restatable}

In the following, we show how to circumvent the hardness result in Theorem~\ref{thm:hardness1}, by providing, in Section~\ref{sec:fixing_actions_with_budget}, a quasi-polynomial-time \emph{bi-criteria} approximation algorithm and, in Section~\ref{sec:fixing_types}, a polynomial-time (exact) algorithm working when the number of buyer's types is fixed.

\subsection{A General Quasi-polynomial-time Bi-criteria Approximation Algorithm}\label{sec:fixing_actions_with_budget}

In order to circumvent the negative result presented in Theorem~\ref{thm:hardness1}, we design a 
%
quasi-polynomial-time algorithm that computes a protocol without menus providing a \emph{bi-criteria} approximation of the seller's expected utility in an optimal protocol.
Formally, our algorithm guarantees a multiplicative approximation $\rho$ of the optimal utility, by only suffering an additional $2^{-\Omega(1/\rho)} + \alpha$ additive loss.
Moreover, we show that our algorithm runs in polynomial time whenever either the number of buyer's actions $m$ or that of states of nature $d$ is fixed. 

In order to prove the approximation guarantees of our algorithm, we rely
%
on Theorems~\ref{QPTAS}~and~\ref{thm:fixedStates}, and we decompose the seller's expected utility in an optimal protocol without menus into the sum of three different terms.
Our algorithm works by computing three protocols without menus, each one approximating one of the three terms. 
Choosing the best protocol among the three provides the desired approximation guarantees.
The following is an intuition of how each term composing the optimal seller's expected utility is approximated by our algorithm:
\begin{itemize}
\item The \emph{first term} is related to the seller's expected utility collected from buyer's types for which the IR constraints are \emph{not} satisfied.
Such a utility term can be trivially achieved by a protocol that charges no price, discloses no information, and never pays back the buyer.
\item The \emph{second term} is related to the best seller's expected utility which can be extracted from a buyer's action.
This is related to the optimal seller's expected utility in a setting with limited liability, since, in that case, the seller's expected utility is determined by the buyer's action only.
%
Thus, the second utility term can be approximated by using either the algorithm provided in Theorem~\ref{QPTAS} or that given in Theorem~\ref{thm:fixedStates}.\footnote{Notice that we cannot employ Theorem~\ref{le:sp} in place of Theorem~\ref{QPTAS}, since the latter guarantees to achieve a seller's expected utility that is arbitrarily close to that of the best protocol employing payment functions parametrized by $\beta$, and this is needed in order to derive the guarantees of our algorithm. Such a guarantee is \emph{not} provided by Theorem~\ref{le:sp}, which only predicates on the quality of the returned protocol with respect to an optimal protocol without menus.}
%
%
%
%
\item The \emph{third term} is related to the seller's expected utility obtained by the transfers between the seller and the buyer, which include the charged price and the final payment.
Such a utility term can be approximated by using a protocol that reveals all the information to the buyer while charging a carefully-chosen price for that.
\end{itemize}
Formally, we prove the following main result:
\begin{restatable}{theorem}{constactionsthmsecond}\label{thm:fixedactions2}
	There exists an algorithm that, given any $\alpha>0$ and $\rho\in (0,\nicefrac{1}{6}]$ as input, computes a protocol without menus whose seller's expected utility is greater than or equal to $\rho \, \textnormal{\textsc{opt}}-2^{-\Omega(1/\rho)}-\alpha$, where $\textnormal{\textsc{opt}}$ is the seller's expected utility in an optimal protocol.
	%
	Moreover, the algorithm runs in time polynomial in $\mathcal{I}^{\log m}$---where $\mathcal{I}$ is the size of the problem instance---when it is implemented with the algorithm in Theorem~\ref{QPTAS} as a subroutine, while it runs in time polynomial in $\mathcal{I}^{d}$ when it is implemented with the algorithm in Theorem~\ref{thm:fixedStates} as a subroutine.
	%
\end{restatable}

%% file: content/const_types.tex
\subsection{Fixing the Number of Buyer's Types}\label{sec:fixing_types}


Next, we study the problem of computing a seller-optimal protocol without menus when the number of buyer's types is fixed, showing that it is possible to design a polynomial-time algorithm.
As a byproduct, the existence of such an algorithm shows that, for protocols without menus, the seller's optimization problem always admits a maximum.
%

As a preliminary result, we show that it is always possible to focus on protocols without menus $(\phi,p,\pi)$ that employ signals belonging to the set $\mathcal{A}^{n}$, and define signaling schemes $\phi: \Theta \to \Delta_{\A^n}$ and payment functions $\pi: \A^n \times \A \to \mathbb{R}_+$ such that, for every signal $\boldsymbol{a} \in \A^n$ and $k \in \K$, it holds $a_k \in \mathcal{B}^k_{\xi^{\boldsymbol{a} },\pi}$, where $a_k \in \A$ denotes the action corresponding to type $k$ in $\boldsymbol{a} $.
%
%
%
Intuitively, in such protocols, a signal specifies an action recommendation for each buyer's type, so that the buyer is always incentivized to follow such recommendations.
With a slight abuse of notation, we say that protocols without menus $(\phi,p,\pi)$ as described above are \emph{generalized-direct} and \emph{generalized-persuasive}.
Formally, we prove the following result:
\begin{restatable}{lemma}{constypesfirst}\label{lem:constypes1}
	Given a seller's protocol without menus, there always exists another protocol without menus which is generalized-direct and generalized-persuasive, and achieves the same seller's expected utility as the original protocol.
\end{restatable}
In order to prove the lemma, we observe that, given a protocol, if two signals induce the same best response for every buyer's type, it is always possible to merge the two signals, retaining the same expected utility for both the seller and the buyer.
Then, by iterating such a process, we recover a signaling scheme and a payment function employing $\A^{n}$ as set of signals . 

As a second crucial step, we show that we can focus on protocols without menus $(\phi,p,\pi)$ whose price $p$ is equal to the budget $b_k$ of one buyer's type $k \in \K$.
Formally:
%
%
\begin{restatable}{lemma}{constypesecond}\label{lem:constypes2}
	Given a protocol without menus, there always exists another protocol $(\phi, p, \pi)$ such that $p=b_k$ for some $k \in \mathcal{K}$, while achieving the same seller's expected utility as the original protocol.
\end{restatable}
%

Finally, equipped with Lemma \ref{lem:constypes1} and Lemma \ref{lem:constypes2}, we are ready to provide our polynomial-time algorithm.
Intuitively, since we can restrict the attention to protocols without menus $(\phi, p, \pi)$ that are generalized-direct and generalized-persuasive, and whose prices $p$ belong to the set $ \{b_{k}\}_{k \in \K}$, we can solve the seller's problem by iterating over all the possible price values $p \in \{b_{k}\}_{k \in \K}$ and, for each of them, over all the possible subsets $\mathcal{R}\subseteq \K\cap \{ k \in\K : b_k \geq p \}$ of buyer's types that satisfy the IR constraint.
%
%
This can be done in polynomial time since the number of buyer's types $n$ is fixed.
Then, for every price value $p \in \{b_{k}\}_{k \in \K}$ and set $\mathcal{R}\subseteq \K\cap \{ k \in\K : b_k \geq p \}$, it is sufficient to solve the following optimization problem:
%
\begin{subequations}\label{quad:LP_kfixed}
	\begin{align}
	\sup_{\substack{\phi \ge 0 \\ \pi \ge 0}} & 
	\,\, \sum_{k \in \mathcal{R}} \lambda_k\sum_{\boldsymbol{a} \in \mathcal{A}^{n}} \sum_{\theta \in \Theta} \mu_\theta \phi_\theta(\avec) \left[u_\theta^s(a_k) -\pi(\avec, a_k)\right] + \sum_{k \notin \mathcal{R}} \lambda_k\sum_{\theta \in \Theta} \mu_\theta u_\theta^s(b^k_{\mu}) \quad\textnormal{s.t.}\,\,\,\quad\quad\quad\quad\quad \\ 
	& \sum_{\theta \in \Theta} \mu_\theta \phi_{\theta}(\avec) \left[u^k_\theta(a_k) + \pi(\avec,a_k) \right] \ge  \sum_{\theta \in \Theta} \mu_\theta \phi_{\theta}(\avec) \left[u^k_\theta(a') +\pi(\avec,a')\right] \nonumber\\
	&\specialcell{\hfill \forall k \in \mathcal{R}, \forall \boldsymbol{a} \in \mathcal{A}^{n},  \forall a' \not = a_k \in \mathcal{A} }\\
	&\specialcell{  \sum_{\boldsymbol{a} \in \mathcal{A}^{n} } \sum_{\theta \in \Theta} \mu_\theta \phi_{\theta}(\avec) \left[u^k_\theta(a_k) + \pi(\avec,a_k) \right] - b_k \ge \sum_{\theta \in \Theta} \mu_\theta  u^k_\theta(b^k_\mu) \hfill \forall k \in \mathcal{R}} \\
	&\specialcell{\sum_{\boldsymbol{a} \in \mathcal{A}^{n} } \sum_{\theta \in \Theta} \mu_\theta \phi_{\theta}(\avec) \left[u^k_\theta(a_k) + \pi(\avec,a_k) \right] - b_k \le \sum_{\theta \in \Theta} \mu_\theta  u^k_\theta(b_\mu^k)\hfill \forall k \not \in \mathcal{R}}\\
	&\specialcell{ \sum_{\boldsymbol{a} \in \mathcal{A}^{n} } \phi_\theta(\avec)=1 \hfill \forall \theta \in \Theta.}
	\end{align}
\end{subequations}

%

By using techniques similar to those used in Section~\ref{sec:protocol_selection} for protocols with menus, we can show that Problem~\eqref{quad:LP_kfixed} is solvable in polynomial time by means of a suitable-defined LP.
%
%
This allows us to state our last results:
\begin{restatable}{theorem}{fixedtypes}\label{thm:fixed_types}
	Restricted to instances in which the number of buyer's types $n$ is fixed, the problem of computing a seller-optimal protocol without menus admits a polynomial-time algorithm.
	%
\end{restatable}
\begin{corollary}
	The problem of computing a seller-optimal protocol without menus always admits a maximum.
\end{corollary}

%

%% file: content/appendix_type_reporting.tex
\section{Proofs Omitted from Section~\ref{sec:protocol_selection}}

\typerepfirst*
\begin{proof}
	Let  $ \{\left(\phi^k, p_k,\pi_k \right) \}_{k \in \mathcal{K}} $ be an IC and IR protocol such that there exist two signals $s_1,s_2 \in \mathcal{S}$ inducing the same best response $\bar a \in A$ for a given receiver's type $ \bar k \in \mathcal{K} $, $i.e.$, $b_{\xi^{s_1}}^{\bar k}=b_{\xi^{s_2}}^{\bar k}=\bar a$.
	
	In the following, we show how to replace $\phi^{\bar k}$ and $\pi^{\bar k}$ with a new signaling scheme $\bar\phi^{\bar k}$ and a new payment function $\bar\pi^{\bar k}$ by merging $s_1,s_2$ into a single signal $\bar s$. Formally we define:
	$ \bar\phi^{\bar k}_{\theta}( \bar s)=\phi^{k}_{\theta}(s_1)+\phi^{k}_{\theta}(s_2)$ for each $\theta \in \Theta$. Similarly, we define:
	\begin{equation*}
	\bar\pi^{\bar k}(\bar s,\bar a)=\frac{\sum_{\theta \in \Theta} \mu_\theta \phi^{\bar k}_{\theta}(s_1) \pi_{\bar k}(s_1,\bar a) +\sum_{\theta \in \Theta} \mu_\theta \phi^{\bar k}_{\theta}(s_2) \pi_{\bar k}(s_2,\bar a)}{\sum_{\theta \in \Theta} \mu_\theta (\phi^{\bar k}_{\theta}(s_1) + \phi^{\bar k}_{\theta}(s_2))},
	\end{equation*}
	 Finally, we does not change all the other components of the protocol \emph{i.e.}, we leave these components of $\{\left(\bar \phi^k, \bar p_k,\bar \pi_k \right) \}_{k \in \mathcal{K}} $ equal to the one in $\{\left(\phi^k, p_k,\pi_k \right) \}_{k \in \mathcal{K}}$.  To prove the lemma we show that the protocol $ \{\left(\bar\phi^k, \bar p_k,\bar\pi_k \right) \}_{k \in \mathcal{K}} $ achieves the same seller's expected utility of $ \{\left(\phi^k, p_k,\pi_k \right) \}_{k \in \mathcal{K}} $, while satisfying the IC and IR constraints. As a first step, we observe that:
	\begin{align*}
	& \sum_{s \in \mathcal{S} \setminus \{s_1,s_2\}} \sum_{\theta \in \Theta} \mu_\theta \bar \phi^{\bar k}_{\theta}(s) \left[ u^{\bar k}_\theta(b^{\bar k}_{\xi^s, \bar \pi}) + \bar\pi_{\bar k}(s, b^{\bar k}_{\xi^s, \bar\pi}) \right] + \sum_{\theta \in \Theta} \mu_\theta \bar\phi^{\bar k}_{\theta}(\bar s) [u^{\bar k}_\theta(\bar a) + \bar\pi_{\bar k}(\bar s,\bar a) ]- \bar p_{\bar k} = \\
	& \sum_{s \in \mathcal{S} \setminus \{s_1,s_2\}} \sum_{\theta \in \Theta} \mu_\theta \phi^{\bar k}_{\theta}(s) \left[ u^{\bar k}_\theta(b^{\bar k}_{\xi^s, \pi}) + \pi_{\bar k}(s, b^{\bar k}_{\xi^s, \pi}) \right] + \sum_{\theta \in \Theta} \mu_\theta \phi^{\bar k}_{\theta}(s_1)[ u^{\bar k}_\theta( \bar a) + \pi_{\bar k}(s_1,\bar a)] + \\ & \hspace{50mm} +	\sum_{\theta \in \Theta} \mu_\theta \phi^{\bar k}_{\theta}(s_2) [u^{\bar k}_\theta( \bar a) + \pi_{\bar k}(s_2,\bar a)] - p_{\bar k}.
	\end{align*}
	The latter equality holds by linearity and proves that the protocol $ \{(\bar \phi^k, \bar p_k,\bar \pi_k ) \}_{k \in \mathcal{K}} $ preserves the left-hand sides of the IR and IC constraints. Moreover, thanks to the convexity of the max operator, we can show that:
	\begin{align*}
	& \max_{a\in \mathcal{A}} \sum_{\theta \in \Theta} \mu_\theta \phi^{ \bar k}_{\theta}(s_1) [u^{\bar k}_\theta(a) + \pi_{\bar k}(s_1,a)] +	\max_{a \in \mathcal{A}} \sum_{\theta \in \Theta} \mu_\theta \phi^{\bar k}_{\theta}(s_2) [ u^{\bar k}_\theta(a) + \pi_{\bar k}(s_2,a)] - p_{\bar k}\ge \\
	& \max_{a\in \mathcal{A}} \sum_{\theta \in \Theta} \mu_\theta \bar\phi^{\bar k}_{\theta}(\bar s)  [u^{\bar k}_\theta(a) + \bar\pi_{\bar k}(\bar s,a)] - p_{\bar k}.
	\end{align*}
	Then, by summing over the set $(\mathcal{S} \cup \{\bar s\})\setminus \{s_1,s_2\}$, we notice that the value of the right-hand side of the IC constraints achieved by the protocol $ \{(\bar \phi^k, \bar p_k,\bar \pi_k ) \}_{k \in \mathcal{K}} $ is less or equal to the the value achieved by $ \{( \phi^k, p_k, \pi_k ) \}_{k \in \mathcal{K}} $. Due to that, we can easily conclude that the new protocol preserves the IC and the IR constraints. Finally, by observing that the following equality holds:
	\begin{equation*}
	\sum_{\theta \in \Theta} \mu_\theta \bar\phi^k_{\theta}(\bar s) [u^{\bar k}_\theta(\bar a) + \bar\pi_{\bar k}(\bar s,\bar a) ] = \sum_{\theta \in \Theta} \mu_\theta \phi^{\bar k}_{\theta}(s_1)[ u^{\bar k}_\theta( \bar a) + \pi_{\bar k}(s_1,\bar a)] +	\sum_{\theta \in \Theta} \mu_\theta \phi^{\bar k}_{\theta}(s_2) [u^{\bar k}_\theta( \bar a) + \pi_{\bar k}(s_2,\bar a)],
	\end{equation*}
	we can easily prove that the two protocols achieve the same seller's expected utility. Then, by iterating this procedure for each buyer's type and couple of signals until there are no two signals inducing the same best response for that type, we get a protocol that employs direct and persuasive signals. 
\end{proof}

\typerepsecond*
\begin{proof}
	
	We first prove that by setting $\tilde \pi_k(a,a')=0$ for each $a \ne a' \in A$ and $k \in \K$, the seller's expected utility does not change and the IC and IR constraints are preserved. First, we show that by taking $\tilde \pi_k(a,a')=0$ for each $a \ne a' \in \mathcal{A}$ and $k \in \K$, the signaling schemes $\phi^k$ remain persuasive. Indeed, we have that:
	\begin{align*}
	\sum_{\theta \in \Theta} \mu_\theta \phi^k_{\theta}(a) u^k_\theta(a) + \pi_k(a,a) & \ge \sum_{\theta \in \Theta} \mu_\theta \phi^k_{\theta}(a') u^k_\theta(a') + \pi_k(a,a')\\ 
	&  \ge \sum_{\theta \in \Theta} \mu_\theta \phi^k_{\theta}(a') u^k_\theta(a').
	\end{align*}
	Moreover, the left-hand side of all the equations defining the IR and IC constraints does not change since it depends only from the payments $\tilde \pi_k(a,a)= \pi_k(a,a)$ for each $a ' \in \mathcal{A}$ that remains unchanged. Furthermore, by setting $\tilde \pi_k(a,a')=0$ for each $a \ne a' \in \mathcal{A}$ and $k \in \K$, the right-hand sides of the IC constraints achieve smaller or equal values, since, intuitively, we are setting to zero non-negative values. Hence, the IC constraints are satisfied. Finally, by observing that the left-hand sides of the IR constraints and the seller's expected utility do not embed terms $\tilde \pi_k(a,a')$ with $a \ne a' \in \mathcal{A}$ and $k \in \K$, we conclude the first part of the proof.
	
	In the second part of the proof we show that it is always possible to define a protocol in which the seller asks to each buyer's type to deposit all their budget at the beginning of the interaction, achieving the same seller's expected utility and satisfying the constraints. Formally, we show that given a protocol $ \{\left(\phi^k, p_k,\pi_k \right) \}_{k \in \mathcal{K}} $  the protocol $ \{\left(\phi^k, \tilde p_k,\tilde\pi_k \right) \}_{k \in \mathcal{K}} $, with $\tilde p_k=b_k$ and $\tilde\pi_k(a) = \pi_k(a,a) + b_k-p_k$ for each $k \in \mathcal{K}$ and $a \in \mathcal{A}$, achieves the same seller's expected utility.
	Indeed by linearity we have:
	\begin{align*}
	\sum_{k \in \mathcal{K}} \lambda_k \Big[ \sum_{a \in\mathcal{A}} \sum_{\theta \in \Theta} \mu_\theta \phi^{k}_\theta(a) u_\theta^s(a) -\pi_k(a)  + p_k \Big] & \hspace{-0.1cm}=\hspace{-0.1cm}
	\sum_{k \in \mathcal{K}} \lambda_k \Big[ \sum_{a \in\mathcal{A}} \sum_{\theta \in \Theta} \mu_\theta \phi^{k}_\theta(s)  u_\theta^s(a) -\pi_k(a) + b_k - b_k + p_k \hspace{-0.05cm} \Big]\\
	& \hspace{-0.1cm}=\hspace{-0.1cm}
	\sum_{k \in \mathcal{K}}  \lambda_k \Big[ \sum_{a \in\mathcal{A}} \sum_{\theta \in \Theta} \mu_\theta \phi^{k}_\theta(a)  u_\theta^s(a) -\tilde\pi_k(a) + b_k \Big].
	\end{align*}
	With  similar arguments it is easy to check  that the protocol $ \{\left(\phi^k, \tilde p_k,\tilde\pi_k \right) \}_{k \in \mathcal{K}} $ satisfies the IC and IR constraints, concluding the lemma.
\end{proof}

\typerethird*
\begin{proof}
We show that for each protocol $ \{\left(\phi^k, p_k,\pi_k \right) \}_{k \in \mathcal{K}} $ that is feasible for Problem~\eqref{eqn:LP_type_reporting1}, we can derive a solution to LP~\eqref{fig:lp_reporting} with at least the same value. This is sufficient to prove the statement.

 By Lemma~\ref{lem:typerep1} and Lemma~\ref{lem:typerep2}, we focus without loss of generality on protocols $ \{\left(\phi^k, p_k,\pi_k \right) \}_{k \in \mathcal{K}} $  that are direct and persuasive.
 Then, we can build a solution $(\bar \phi,\bar l,\bar y)$ to LP~\eqref{fig:lp_reporting} letting 
 \[\bar l_k(a)=\sum_{\theta \in \Theta} \mu_\theta \phi_\theta(a) \pi_k(a)\] 
 for each $k \in \K$ and $a \in \mathcal{A}$, and 
 \[\bar y_{k,k',a} = \max\left\{ \sum_{\theta \in \Theta}  \mu_\theta \phi^{k'}_{\theta}(a)  \left (u^{k}_\theta(a) + \pi_{k'}(a,a)\right) ,\max_{a'\neq a} \sum_{\theta \in \Theta}  \mu_\theta \phi^{k'}_{\theta}(a')  u^{k}_\theta(a') \right\} \]
 for each $k,k' \in \K$ and $a \in \mathcal{A}$.
 Moreover, we let $\bar \phi=\phi$.
It is easy to verify that the solution $(\bar \phi,\bar l, \bar y)$ results feasible for LP~\eqref{fig:lp_reporting}.  
\end{proof}

\typerefourth*
\begin{proof}
As a first step we show that from a feasible solution $(\phi,l,y)$ to LP~\eqref{fig:lp_reporting} we can recover another solution with at least the same value in which if $\pi_k(a)>0$ then there exists a $\theta \in \Theta$ such that $\phi^k_\theta (a)>0$.  
Specifically, given a $k \in \K$ and an $a \in \mathcal{A}$ such that $\phi^k_\theta (a)=0$ for all $\theta \in \Theta$ and $l_k(a)>0$,  let $\bar a \in \mathcal{A}$ be $(\bar a, \bar \theta)$ be any couple of an action and a state such that $\phi^k_{\bar \theta} (\bar a)>0$.
 We now define a new feasible solution $(\bar \phi, \bar l, \bar y)$ as follows:
 \begin{itemize}
 	\item $\bar l_k(a) = 0$
 	\item $\bar l_k(\bar a) = l_k(a) + l_k(\bar a)$
 	\item $\bar{y}_{k',k,a} = 0 \quad \forall k' \in \K $
 	\item $\bar y_{k',k,\bar a} = y_{k',k,\bar a} + l_{k}(a) \quad  \forall k' \in \K $,
\end{itemize}
while we leave all the other terms variables equal to the ones in $(\phi,l,y)$.
It is easy to check that the solution $(\bar \phi,\bar l,\bar y)$achieves the same seller's expected utility while satisfying the constrains. 
 Applying  this procedure for each couple $(k,a)$ such that $\phi^k_\theta (a)=0$ for all $\theta \in \Theta$ and $l_k(a)>0$, we obtain a new solution $(\tilde \phi,\tilde l, \tilde y)$ such that if $\tilde l_k(a)>0$ then there exists a $\theta \in \Theta$ such that $\tilde \phi^k_\theta (a)>0$.
 
To recover a feasible protocol we just need to set payments as follows. For each  couple $(k,a)$, if there exists a $\theta \in \Theta$ such that $\tilde \phi^k_\theta(a)>0$, we set $\tilde \pi_k(a)= \tilde  l_k(a)/(\sum_{\theta \in \Theta}\mu_\theta \tilde \phi^k_\theta(a))$.
Otherwise, we set $\tilde \pi_k(a)=0$. 
Notice that the ratio $\tilde  l_k(a)/(\sum_{\theta \in \Theta}\mu_\theta \tilde \phi^k_\theta(a))$ is always well defined.
Moreover, it is easy to see that $\{\tilde \phi^k,\tilde p_k,\tilde \pi_k\}_{k \in \K}$, with $\tilde p_k=b_k$ for each $k \in \K$ is a feasible solution to Problem~\eqref{eqn:LP_type_reporting1} with at least the same value as $(\phi,l,y)$ for LP~\eqref{fig:lp_reporting}  . This concludes the proof.
\end{proof}

\typerepthm*
\begin{proof}
The algorithm solves LP~\ref{fig:lp_reporting}. 
By Lemma~\ref{lem:typerep3}, this solution has value
greater or equal to the supremum of Program~\ref{eqn:LP_type_reporting1}. Then, exploiting Lemma~\ref{lem:typerep4}, we can recover in polynomial-time a protocol with at least the same utility, \emph{i.e.}, an optimal one.
\end{proof}

%% file: content/appendix_hardness_first.tex
\section{Proofs Omitted from Section~\ref{sec:principal_agent}}
\lemmaprincipal*
\begin{proof}
	We start proving the first part of the statement.
	Given an instance of Problem~\eqref{prob:payment_func}, we build an instance of the observable-action principal-agent problem	with $c^k_a=\sum_{\theta \in \Theta} \xi_\theta \left[  u^k_\theta(b_\xi^k) -u^k_\theta(a) \right]$ and $r_a = \sum_{\theta \in \Theta} \xi_\theta \, u^s_\theta(a)$.
	To prove the equivalence between the two settings, we first show that the set of best-responses for the two problems coincides.
	Indeed, given a payment function $\pi$ and a type $k\in \K$, let $a \in \mathcal{B}^k_{\xi,\pi}$. Then, for each $a'\in \A$
	\begin{align*}
		\pi(a) - c^k_a&= \pi(a) - \sum_{\theta \in \Theta} \xi_\theta \left[  u^k_\theta(b_\xi^k) -u^k_\theta(a) \right]\\
		& = \pi(a) + \sum_{\theta \in \Theta} \xi_\theta u^k_\theta(a)  -\sum_{\theta \in \Theta} \xi_\theta u^k_\theta(b_\xi^k) \\
		&\ge \pi(a') + \sum_{\theta \in \Theta} \xi_\theta u^k_\theta(a')  -\sum_{\theta \in \Theta} \xi_\theta u^k_\theta(b_\xi^k) \\
		& = \pi(a') - c^k_{a'},
	\end{align*}
	showing that $a \in  \mathcal{B}^k_{\pi}$.
	Similarly, we can prove that if $a \in  \mathcal{B}^k_{\pi}$, then $a \in \mathcal{B}^k_{\xi,\pi}$
	This implies that the set of best responses are equivalent for each payment function $\pi$, \emph{i.e.}, $\mathcal{B}^k_{\pi} =\mathcal{B}^k_{\xi,\pi}$.
	Hence, 
	\begin{align*}
	\argmax_{\pi} \sum_{k \in \K}\lambda_k \left[ r_{b^k_\pi}- \pi(b^k_\pi)\right]  &=\argmax_{\pi} \sum_{k \in \K}\lambda_k \left[ \sum_{\theta \in \Theta} \xi_\theta \, u^s_\theta(b^k_{\pi})  - \sum_{\theta \in \Theta} \xi_\theta \left[  u^k_\theta(b_\xi^k) -u^k_\theta(b^k_\pi) \right]\right]\\
	&=\argmax_{\pi} \sum_{k \in \K}\lambda_k \left[ \sum_{\theta \in \Theta} \xi_\theta \, u^s_\theta(b^k_{\pi})  + \sum_{\theta \in \Theta} \xi_\theta u^k_\theta(b^k_\pi) \right]\\
	&=\argmax_{\pi} \sum_{k \in \K}\lambda_k \left[ \sum_{\theta \in \Theta} \xi_\theta \, u^s_\theta(b^k_{\xi,\pi})  + \sum_{\theta \in \Theta} \xi_\theta u^k_\theta(b^k_{\xi,\pi}) \right],
	\end{align*}
	showing the equivalence between the two problems.
	This proves the first part of the statement.
	
	We now show that from an instance of the observable-action principal-agent problem we can always build an instance of the selling-information problem without menus and with only a single state of nature $\theta$.
	In particular, we set $u^s_{\theta}(a)=r_a$ for each $a \in \mathcal{A}$, and $u^k_\theta(a)=1-c_a^k$ for each $a \in \mathcal{A}$ and $k \in \K$. 
	Following a analysis similar to the first part of the proof, we can show that the two problems are equivalent.
	This concludes the proof.
\end{proof}

\theoremHardnessState*

\begin{proof}
	
	We reduce from vertex cover in cubic graphs.
	Formally, it is \NPHard\ to approximate the size of the minimum vertex cover in cubic graphs with an approximation $(1+\varepsilon)$, for a given constant $\varepsilon>0$~\cite{APXAlimonti}.
	Let $\eta= \varepsilon/7$.
	We show that an $(1-\eta)$-approximation to the principal-agent problem with observable actions can be used to provide a $(1+\varepsilon)$ approximation to vertex cover, concluding the proof.
	
	Consider an instance of vertex cover $(V,E)$ with nodes $V$ and edges $E$. Let $\rho=|V|$ and $\ell=|E|$.
	Given a vertex $v \in V$, we let $E(v)$ be the set of edges $e$ such that $v$ is one of the extreme of the edge $e$. Similarly, given an edge $e \in E$, let $V(e)$ be the set of vertexes $v$ such that $e$ is an edge with extreme $v$.
	We build an instance of the principal-agent problem with observable actions as follows. For each vertex $v \in V$, there exists an agent's type $k_v$, while for each $e \in E$, there exists a type $k_e$.
	For each vertex $v \in V$, there exists an action $a_v$ and an additional action  $a_-$.
	The cost of a type $k_e$, $e \in E$, is $c^{k_e}_{a_-}=0$, $c^{k_e}_{a_v}=\frac{1}{2}$ if $e \in E(v)$ and $1$ otherwise.
	The cost of a type $k_v$, $v \in V$, is $c^{k_v}_{a_{v}}=0$, and $1$ otherwise.
	Finally, the principal's utility is equal to $1$ if the action is in $\{a_{v}\}_{v \in V}$, while is equal to $0$ otherwise, \emph{i.e.}, $u^s_{a_v}=1$ for each $v \in V$ and $u^s_{a_-}=0$.
	All the types are equally probable, \emph{i.e.}, $\lambda_k=\frac{1}{\rho+\ell}$ for each $k \in \K$.
	
	First, we show that if there exists a vertex cover $V^\star$ of size $\nu$, the value of the problem is at least $\frac{(\rho-\nu) + \frac{\nu}{2} + \frac{\ell}{2}}{\rho+\ell}$. 
	Consider the payment function such that $\pi(a_v)=\frac{1}{2}$ if $v\in V^\star$ and $0$ otherwise.
	A type $k_v$ with $v \notin V^\star$ plays the action $a_v$ and receives a payment of $0$. A type $k_v$ with $v \in V^\star$ plays the action $a_v$ and receives a payment of $\frac{1}{2}$. A type $k_e$ plays an action $a_v$ such that $e \in E(v)$ (this action exists by construction) and receives a payment of $\frac{1}{2}$. 
	It is easy to see that the expected seller's utility is $\frac{(\rho-\nu) + \frac{\nu}{2} + \frac{\ell}{2}}{\rho+\ell}$.
	
	Suppose that there exists an algorithm that provides a $1-\eta$ approximation. This implies that the algorithm returns a solution, \emph{i.e.}, a payment function $\pi$,  with value at least $(1-\eta)\frac{\rho-\frac{k}{2}+\frac{\ell}{2}}{\rho+\ell}$.
	We show how to exploit the payment function $\pi$ to build a vertex cover of size at most $(1+\varepsilon) \nu$ in polynomial time.
	In particular, given $\pi$ we recover a vertex cover $\bar V$ of the desired size as follows.
	First, it is easy to see that we can set the payment $\pi(a_-)=0$ and payments $\pi(a_v)\in \{0,\frac{1}{2}\}$ for each $v \in V$ without decreasing the utility. Intuitively, payments are useful only to change the best response of a type $k_e$, $e \in E$, from $a_-$ to $a_v$, $v \in V$. To do so, it is sufficient a payment of $\frac{1}{2}$.
	Then, let $\bar E$ be the set of edges $e \in E$ such that the best response of $k_e$ is $a_-$, \emph{i.e.},  $\bar E= \{e \in E: b^{k_e}_{\pi}=a_-	\}$. 
	Consider a edge $e \in \bar E$ and a vertex $v \in V(e)$. Since $e \in \bar E$ the payment  $\pi(a_v)=0$ and no type $k_e$ plays action $a_v$.
	 Hence, if we modify the payments by letting $\pi(a_v)=\frac{1}{2}$ on the action $a_v$ we have three effects: i) the type $k_e$ changes the best response to $a_v$, ii) some other types $e' \in E$ could change from action $a_-$ to $a_v$,\footnote{Notice that there could be other actions $a_v'$ with $\pi(a_{v'})=\frac{1}{2}$ that provides the same utility of $a_{v}$ for both the principal and the agent.} iii) the payment of type $k_v$ increases by $\frac{1}{2}$.
	Overall the principal's total utility increases by $\frac{1}{2}\lambda_{k_e}$ since $k_e$ changes from action $a_-$ with payment $0$ to action $a_v$ with payment $\frac{1}{2}$, and it decreases by $-\frac{1}{2}\lambda_{k_v}$ as the payment to type $k_v$ increases by $\frac{1}{2}$. Moreover, if other types $k_{e'}$, $e'\neq e$ change from $a_-$ with payment $0$ to action $a_v$ with payment $\frac{1}{2}$ the principal's utility increases.
	This implies that the principal's utility does not decrease with this procedure.
	Hence, repeating this procedure we can build a payment function $\pi$ with the same utility such that all the agent's type plays actions $a_v$, $v \in V$.
	Then, let $\bar V$ be the set of vertexes with at least one agent of type $k_e$, $e \in E$ that plays this action, \emph{i.e.}, $\bar V= \{v \in V: b^{k_e}_\pi=a_v, e \in E\}$. We show that $\bar V$ is an vertex cover of size at most $(1+\epsilon)\nu$, concluding the proof.
	Notice that we can set payment $\pi(a_v)=0$ for each $v \in V \setminus \bar V$ without decreasing the seller's utility. Removing the payment does not change any best response since the only type playing the action $a_v$ is the type $k_v$, the utility $u^{k_v}(a_v)=1$, and the utility of playing any other action $a\neq a_v$ is  $u^{k_v}(a)+\pi(a)\le \frac{1}{2}$.
	Then, the principal's utility is  \[\frac{(\rho-|\bar V|)-\frac{1}{2}|\bar V|+\frac{1}{2} \ell}{\rho+\ell}=\frac{\rho-\frac{1}{2}|\bar V|+\frac{1}{2} \ell}{\rho+\ell},\]
	since $\rho-|\bar V|$ types $k_v$, $v \in V$, play $a_v$ and receive  payment $0$, $|\bar V|$ types $k_v$, $v \in V$, play $a_v$ and receive payment $\frac{1}{2}$, and all the types $k_e$, $e \in E$,  play an action $a_v$, $v \in V$ and receive payment $\frac{1}{2}$.
	Then, since the principal's utility is by assumption $\frac{\rho-\frac{1}{2}|\bar V|+\frac{1}{2}\ell}{\rho+\ell}\ge(1-\eta)\frac{\rho-\frac{\nu}{2}+\frac{\ell}{2}}{\rho+\ell} $,
	it holds 
	\[|\bar V|\le \eta (2\rho+ \ell)+\nu\le (1+7 \eta)\nu=(1+7\eta)\nu=(1+\varepsilon)\nu,\]
	where the second inequality comes from $\rho= \frac{2}{3}\ell$ and $\ell \le 3\nu$. 
	This concludes the proof.
\end{proof}

\corolPD*

\begin{proof}
	As a first step, we notice that even if \citet{castiglioni2021Contract} show that linear contracts provide the desired approximation with respect to optimal contracts, their proof can be extended to show that linear contracts provide the same approximation with respect to the optimal social welfare, \emph{i.e.}, $\sum_{k \in \K} \lambda_{k} \max_{a \in \A} [r_a-c^k_a]$. To prove this result, it is sufficient to follow all the steps of Theorem~3 of  \citep{castiglioni2021Contract} except for the one in which Observation~1 is employed to upperboud the value of the optimal contract with  the social welfare.
	Finally, we can modify this result to hold in our setting exploiting Lemma~\ref{lem:prin2sing}.
\end{proof}

%% file: content/appendix_const_actions.tex
\section{Proofs Omitted from Section~\ref{sec:selection_limited_liability}}

\constactionsfirst*
\begin{proof}
	Let $\tilde\p\in {\pset_q}$ be the empirical mean of $q$ \emph{i.i.d.} samples drawn according to $\p^\ast \in \Delta_\Theta$, where each $\theta\in\Theta$ has probability $\p^\ast_\theta$ of being sampled. 
	Therefore, $\tilde\p \in \Xi_q$ is a random vector supported on $q$-uniform posteriors with expectation $\p^\ast \in \Delta_\Theta$. 
	Moreover, let $\gamma\in\Delta_{\pset_q}$ be a probability distribution such as, for each $\p\in\pset_q$, it holds $\gamma_\p\defeq \pr(\tilde\p=\p)$. 
	We build a new payment function $\tilde \pi$ such that for each $\xi \in \Delta_\Theta$ and $a \in \A$, we have $\tilde\pi(\xi,a)=\pi(\xi^*,a)$
	Moreover, we let $\pset_{q,\epsilon}$ be the set of posteriors such that $\p\in\pset_{q,\epsilon}$ if and only if for each $a\in\A$ it holds:
	\begin{equation}\label{eq:pset_q_eps}
	\left|\sum_{\theta\in\Theta}\left(\p_\theta u^k_\theta(a)-\p^*_\theta u_\theta^k(a)\right)\right|\leq \frac{\epsilon}{2}.
	\end{equation}
	Then, for each $\p\in\pset_{q,\epsilon}$, we have that $\brset_{\p^\ast, \pi}^k\subseteq \brset^{k,\epsilon}_{\p, \pi}$.
	In particular, for any $a^\ast\in\brset_{\p^\ast, \pi}^k$, $\p\in\pset_{q,\epsilon}$ and $a\in \A$:
	\begin{align*}
	\sum_{\theta\in\Theta}\p_\theta u_\theta^k(a^\ast) + \tilde \pi(\xi, a^\ast)
	& \geq \sum_{\theta\in\Theta}\p_\theta^\ast u_\theta^k(a^\ast) + \tilde\pi(\xi^\ast, a^\ast) -\frac{\epsilon}{2} 
	& \textnormal{(By Eq.~\eqref{eq:pset_q_eps} and the definition of $\brset_{\p^\ast, \pi}^k$)}\\
	& \geq \sum_{\theta\in\Theta} \p_\theta^\ast u_\theta^k(a) + \tilde\pi(\xi^\ast, a) -\frac{\epsilon}{2}\\
	&\geq \sum_{\theta\in\Theta}\p_\theta u_\theta^k(a)+\tilde \pi(\xi, a) -\epsilon&\textnormal{(By Equation~\eqref{eq:pset_q_eps})}
	\end{align*}
	which is precisely the definition of $\brset_{\p^\ast, \pi}^{k,\epsilon}$.
	
	 For each $a\in \A$, let $\tilde t_a^k \defeq \sum_{\theta\in\Theta}\tilde\p_\theta u_\theta^k(a) + \tilde \pi(\tilde \xi, a)$ and $t_a^k\defeq \sum_{\theta\in\Theta} \p^\ast_\theta u^k_\theta(a)+ \tilde \pi(\xi^*, a)$.	
	By the Hoeffding's inequality we have that, for each $a\in\A$,
	\begin{align}\label{eq:lemma_hoeffding_t}
	\pr(|\tilde{t}^k_a  - \Expec[\tilde{t}^k_a]|\geq \frac{\epsilon}{2})\leq 2e^{-2q(\epsilon/2)^2}=2e^{- \log(2\nAct/\alpha)}\leq \frac{\alpha}{\nAct}.
	\end{align}
	%
	Moreover, Equation~\eqref{eq:pset_q_eps} and the union bound yield the following:
	\begin{align*}
	\sum_{\p\in\pset_{q,\epsilon}} \gamma_\p &= \pr(\tilde\xi\in \pset_{q,\epsilon})\\
	& = \pr (\bigcap_{a\in\A} \left| \tilde t_a^k- t_a^k \right|\leq \frac{\epsilon}{2})\\
	& \geq 1-\sum_{a\in\A} \pr(\left| \tilde t_a^k- t_a^k \right|\geq \frac{\epsilon}{2}) \\
	& \geq  1-\alpha. & \textnormal{(By Equation~\eqref{eq:lemma_hoeffding_t})}
	\end{align*}
	Let $\bar{\xi}$ be a $\nState$-dimensional vector defined as $\bar\alpha_\theta\defeq \sum_{\p\in\pset_q\setminus\pset_{q,\epsilon}}\gamma_\p\p_\theta$. By definition and for the previous result we have:
	$
	\sum_{\theta\in\Theta}\bar{\xi}_\theta\leq \alpha.
	$ Finally, we can show:
	\begin{align*}
	\sum_{\p\in\pset_{q}}\gamma_\p \sum_{\theta \in \Theta} &\xi_\theta  u^s_\theta( b^{k,\epsilon}_{\xi,\pi'})  - \pi'(\xi, b^{k,\epsilon}_{\xi,\pi'}) \\
	\geq&  \sum_{\p\in\pset_{q}}\gamma_\p \sum_{\theta \in \Theta} \xi_\theta u^s_\theta( b^{k,\epsilon}_{\xi,\pi}) - \pi(\xi, b^{k,\epsilon}_{\xi,\pi})\\ \geq &
	\sum_{\p\in\pset_{q,\epsilon}} \gamma_\p \sum_{\theta \in \Theta} \xi_\theta u^s_\theta( b^{k,\epsilon}_{\xi,\pi}) - \pi(\xi, b^{k,\epsilon}_{\xi,\pi}) \\
	\geq & \specialcell{\sum_{\p\in\pset_{q,\epsilon}} \gamma_\p \sum_{\theta\in\Theta} \p_\theta u_\theta^\send(b_{\p^\ast,\pi}^k)-\pi(\xi^*, b_{\p^\ast,\pi}^k) 
	\hfill \textnormal{($\brset_{\p^\ast, \pi}^{k}\subseteq \brset_{\p, \pi}^{k,\epsilon}$ for each $\p\in\pset_{q,\epsilon}$)}}\\
	= &  \sum_{\theta\in\Theta}\Big( u_\theta^\send(b_{\p^\ast,\pi}^k) -\pi(\xi^*, b_{\p^\ast,\pi}^k) \Big)\Big( \sum_{\p\in\pset_{q,\epsilon}} \gamma_\p\p_\theta \Big) \\
	= & \specialcell{\sum_{\theta\in\Theta} \Big( u_\theta^\send(b_{\p^\ast,\pi}^k) -\pi(\xi^*, b_{\p^\ast,\pi}^k) \Big)\Big( \sum_{\p\in\pset_q} \gamma_\p\p_\theta - \bar\xi_\theta \Big)
	\hfill \textnormal{(By definition of $\bar{\alpha}$)} }\\
	 = &\sum_{\theta\in\Theta} \Big( \xi_\theta u_\theta^\send(b_{\p^\ast,\pi}^k) -\pi(\xi^*, b_{\p^\ast,\pi}^k) \Big)\Big( \sum_{\p\in\pset_q} \gamma_\p\p_\theta \Big)- \sum_{\theta\in\Theta} \Big( u_\theta^\send(b_{\p^\ast,\pi}^k) -\pi(\xi^*, b_{\p^\ast,\pi}^k) \Big)\bar\xi_\theta \\
	\geq & \specialcell{\sum_{\theta\in\Theta} \Big( u_\theta^\send(b_{\p^\ast,\pi}^k) -\pi(\xi^*, b_{\p^\ast,\pi}^k) \Big)\Big( \sum_{\p\in\pset_q} \gamma_\p\p_\theta \Big)-  \sum_{\theta\in\Theta} \bar\xi_\theta \hfill \textnormal{(Utilities in $[0,1]$)}}\\
	\geq & \sum_{\theta\in\Theta}  \xi^*_\theta u_\theta^\send(b_{\p^\ast,\pi}^k) -\pi(\xi^*, b_{\p^\ast,\pi}^k) - \alpha.
	\end{align*}
	
	Finally, by definition of $\gamma$, we have that, for each $\theta\in\Theta$:
	\[
	\sum_{\p\in\pset_q} \gamma_\p\p_\theta=\p_\theta^\ast.
	\]
	This concludes the proof.
\end{proof}

\constactionsecond*
\begin{proof}
	Given a posterior $\xi \in \Delta_\Theta$ and a tuple $a \in \mathcal{A}^{|\mathcal{K}|}$ we let $\Pi_{a} \subseteq \mathbb{R}_+^m$ be the set of payment functions $\pi$ such that for each $k \in \K$ it holds $a_k \in \mathcal{B}^k_{\xi,\pi}$.
	Given an $a \in \mathcal{A}^{|\mathcal{K}|}$, the problem of computing an optimal payment function restricted to payment functions in $\Pi_a$ can be formulated as follows:
	\begin{subequations}
		\begin{align*}
		\min_{\pi} & \,\, \sum_{k \in \mathcal{K}} \lambda_k \pi(\xi, a_k)   \,\,\,\, \text{s.t.} \quad\quad\quad\quad\quad\quad\quad\quad\quad\quad\quad\quad\quad\quad\quad\quad\quad\quad\quad\quad\\
		& \specialcell{\sum_{\theta \in \Theta}\xi_\theta u_\theta^k(a_k)+\pi(\xi, a_k) \ge \sum_{\theta \in \Theta}\xi_\theta u_\theta^k(a')+\pi(\xi, a')  \hfill \forall a' \in \mathcal{A}, k \in \mathcal{K}}\\
		&\specialcell{\pi(\xi,a') \ge 0  \hfill \forall a' \in \A.}
		\end{align*}
	\end{subequations}
	We observe that, for each tuple $a \in \mathcal{A}^{|\mathcal{K}|}$, the vertexes of the regions $\Pi_a \subseteq \mathbb{R}_+^m$ are identified by $m$ of the common $O(nm^2+m)$ constraints:  
	\begin{subequations}
				\begin{align*}
	&\specialcell{\sum_{\theta \in \Theta}\xi_\theta u_\theta^k(a')+\pi(\xi, a') \ge \sum_{\theta \in \Theta}\xi_\theta u_\theta^k(a'')+\pi(\xi, a'') \hfill \forall a'\ne a''  \in \mathcal{A}, \forall k \in \mathcal{K} }\\
	&\specialcell{ \pi(,\xi,a') \ge 0 \hfill \forall a' \in \mathcal{A}.}
	\end{align*}
	\end{subequations}
	
	Hence, the total number of vertexes defining all the regions $\Pi_a$, $a \in \mathcal{A}^{|\mathcal{K}|}$, is at most $ \binom{n m^2+m}{m}= O((n m^2+m)^{m})$.
	Finally, since the objective function is linear in $\Pi_a$ for each tuple $a \in \mathcal{A}^{|\mathcal{K}|}$, given the optimal tuple of induced actions $a^*\in \mathcal{A}^{|\mathcal{K}|}$ the optimum is attained in one of the vertexes of $\Pi_{a^*} $. Moreover, there are overall $\big|\bigcup_{a\in \mathcal{A}^{|\mathcal{K}|}} V(\Pi_a) \big|= O((k m^2+m)^{m})$ vertexes, where $V(\cdot)$ denotes the set of vertexes of the polytope. Hence, when $m$ is  fixed, it is possible to enumerate in polynomial time over all the vertexes in $\bigcup_{a\in \mathcal{A}^{|\mathcal{K}|}} V(\Pi_a)$   and compute the optimal payment function.
\end{proof}

\constactionsthmfirst*
\begin{proof}
	Given two arbitrary constants $\alpha,\epsilon >0$ we let $(\gamma, \pi)$ be an optimal protocol. We show that an $(\alpha+2\sqrt{\eps})$-optimal protocol $(\gamma^*, \pi^*)$  can be computed in polynomial time. As a first step we define a signaling scheme $\gamma^*$ supported in $\Xi_q$ as follows:
	\begin{equation*}
	\gamma_{\tilde{\xi}}^*=\sum_{\xi \in \text{supp}(\gamma)} \gamma_\xi  \gamma^{\xi}_{\tilde{\xi}} \hspace{4mm} \forall \tilde{\xi} \in \Xi_q,
	\end{equation*}
	where $\gamma^{\xi} \in \Delta_{\Xi_q}$ is the signaling scheme satisfying Lemma \ref{lem:actions_quniform} with $q=\frac{2\log(2m/\alpha)}{\epsilon^2}$. First we observe that $\gamma^* \in \Delta_{\Xi_q}$ satisfies the consistency constraints, indeed we have:
	\begin{equation*}
	\sum_{\tilde \xi \in \Xi_q}\gamma_{\tilde{\xi}}^* \, \tilde{\xi_\theta}  
	= \sum_{\xi \in \text{supp}(\gamma)} \gamma_\xi \sum_{\tilde \xi \in \Xi_q}  \gamma^{\xi}_{\tilde{\xi}} \tilde{\xi}_\theta
	= \sum_{ \xi \in \text{supp}(\gamma)} \gamma_{\xi} \xi_\theta 
	= \mu_\theta \,\,\,\ \forall \theta \in \Theta.
	\end{equation*}
	 Moreover, let $\pi^*:\Delta_\Theta \times \mathcal{A} \to \mathbb{R}_{+}$ be the optimal payment function in each $\tilde \xi \in \Xi_q$. We show that the protocol $(\gamma^*, \pi^*, 0)$ is $(\alpha+2\sqrt{\eps})$-optimal.
	 Let $\pi'':\Delta_\Theta \times \mathcal{A} \to \mathbb{R}_{+}$ be the optimal payment function in each $\xi \in \Xi_q$ when the buyer is playing an $\epsilon$-best response, \emph{i.e.}, 
	 \[\pi''(\xi,\cdot)\in \arg \max_{\tilde \pi(\xi,\cdot)} \sum_{k \in \K} \lambda_k [\sum_{\theta} \xi_\theta u^s_\theta(b^{k,\epsilon}_{\xi,\tilde \pi})- \tilde \pi(\xi)].\] 
	 Moreover, let $\pi':\Delta_\Theta \times \mathcal{A} \to \mathbb{R}_{+}$ be the payment function such that $\pi'(\xi,a)=(1-\sqrt{\epsilon})\pi''(\xi,a)+ \sqrt{\epsilon} \sum_{\theta \in \theta} \xi_\theta u^s_\theta(a)$ for each $\xi$ and $\A$. 
	 Then, we have:
	\begin{align*} 
	\sum_{\tilde{\xi} \in \Xi_q}\gamma_{\tilde{\xi}}^* &\Big( \sum_{\theta \in \Theta} \tilde \xi_\theta  u_\theta^s(b^{k}_{\tilde \xi,\pi^*})- \pi^*(\tilde{\xi},
	b^{k}_{\tilde \xi,\pi^*}) \Big) 
	\quad  \quad\quad\quad\quad\quad\quad\quad\quad\quad\quad\quad\quad\quad\quad\quad\quad\quad\quad\quad\quad\quad\quad\quad
	\\ &\specialcell{\ge \sum_{\tilde{\xi} \in \Xi_q}\gamma^*_{\tilde{\xi}} \Big( \sum_{\theta \in \Theta} \tilde \xi_\theta u_\theta^s(b^{k}_{\tilde \xi,\pi'})-\pi'(\tilde{\xi},
	b^{k}_{\tilde \xi,\pi'}) \Big) \hfill
	 \textnormal{(Optimality of $\pi^*$)}}
	 \\
	& \specialcell{ \ge   \sum_{\tilde{\xi} \in \Xi_q}\gamma_{\tilde{\xi}}^* \Big( \sum_{\theta \in \Theta} \tilde{\xi}_\theta u_\theta^s(b^{k,\epsilon}_{\tilde \xi,\pi''})-\pi''(\tilde{\xi},
	b^{k,\epsilon}_{\tilde \xi,\pi''})\Big) - 2 \sqrt{\epsilon} 
	\hfill \textnormal{(By Proposition \ref{thm:dutting})}}
	 \\
	& \specialcell{\ge \sum_{\xi \in \text{supp}(\gamma)} \gamma_\xi \Big( \sum_{\tilde{\xi} \in \Xi_q }
	\gamma^{\xi}_{\tilde{\xi}}  \big(\sum_{\theta \in \Theta} \tilde{\xi}_\theta u_\theta^s(b^{k,\epsilon}_{\tilde \xi,\pi''})-\pi''(\tilde \xi, b^{k,\epsilon}_{\tilde \xi,\pi''})\big) \Big) - 2 \sqrt{\epsilon} \hfill \textnormal{(By defintion of $\gamma^*$)} }
	\\
	& \specialcell{\ge \sum_{\xi \in \text{supp}(\gamma)} \gamma_{\xi} \Big(  \sum_{\theta \in \Theta} \xi_\theta u_\theta^s(b^{k}_{ \xi,\pi})-\pi(\xi, b^{k}_{ \xi,\pi})\Big) -\alpha - 2 \sqrt{\epsilon} 
	\hfill \textnormal{(By Lemma \ref{lem:actions_quniform})}}.
	\end{align*}
	%

	Notice that the optimal payment $\pi^*:\Delta_\Theta \times \mathcal{A} \to \mathbb{R}_{+}$ in each $\xi \in \Xi_q$ can be computed in polynomial time employing Lemma~\ref{lem:price_fun}.
	 Hence, to compute the optimal signaling scheme $\gamma^* \in \Delta_{\Xi_q}$ we can solve the following LP:
	\begin{align*}
	& \sum_{k \in \K} \lambda_k \sum_{\xi \in \Xi_q} \gamma_\xi \sum_{\theta \in \Theta} \xi_\theta u^s_\theta(b^{k}_{ \xi,\pi^*}) - \pi^*(\xi, b^{k}_{ \xi,\pi^*}) \,\,\ \text{s.t.} \\
	& \sum_{ \xi \in \text{supp}(\gamma)} \gamma_\xi \xi_\theta = \mu_\theta \,\,\,\ \forall \theta \in \Theta.
	\end{align*}
	Note that since $|\Xi_q|=O(d^q)$, all the payment function $\pi^*:\Xi_q \times \mathcal{A} \to \mathbb{R}_{+}$ can be precomputed in polynomial time. Moreover, the LP has polynomially many variables and constraints and can be solved efficiently. 
	Finally, the  solution returned by the LP is $\alpha+2\sqrt{\eps}$-optimal.
	This concludes the proof.
%
\end{proof}

\QPTAS*

\begin{proof}
	The proof follows the same steps of Theorem~\ref{le:sp}. However, it relies on Theorem~\ref{thm:linear} instead of Lemma~\ref{lem:price_fun} to compute an approximate payment function for all $q$-uniform posteriors. 
\end{proof}

%% file: content/appendix_const_states.tex
\conststatesfirst*


\begin{proof}
	
	We define a payment function $\pi' : \Delta_\Theta \times \mathcal{A} \to \mathbb{R}_{+}$ as follows: $\pi'(\xi,a)=\pi(\xi^*,a)$ for each $a \in \mathcal{A}$ and $\xi \in \Delta_\Theta$.
	Furthermore, we define: 
	\[I_{\alpha}(\xi^*)=\left\{ \xi \in \Delta_{\Theta} : \lVert \xi-\xi^*\rVert_{\infty} \le \frac{\alpha^2}{18d}\right\},\] as the neighborhood of the given posterior $\xi^* \in \Delta_{\Theta}$ and $\Xi(\xi^*)= I_{\epsilon} (\xi^*)\cap \Xi_q $ its intersection with the set $\Xi_q$. Notice that if $\ q \ge \frac{18d}{\alpha^2}$, it holds $\xi^* \in \text{co}(\Xi(\xi^*))$.
	\footnote{Given a finite set $A$ we denote with $\text{co}(A)$ the set containing all the convex combination of elements in $A$.}

	We show that for each $\xi \in I_{\alpha}(\xi^*)$, it holds $ b_{\xi^*, \pi'}^{k} \in \mathcal{B}_{\xi, \pi'}^{k,\epsilon}$, where $\epsilon \coloneqq  \alpha^2/9$. As a first step, by Hölder's inequality we have that
	\begin{equation*}
	\sum_{\theta \in \Theta}|(\xi_\theta - \xi^*_\theta) u_\theta^s(a)| \le d || \xi-\xi^*||_{\infty} = \epsilon/2\,\,\ \forall a \in \mathcal{A},
	\end{equation*}
	 Moreover, by the definition of best response and the previous inequality, we have that:
	\begin{align*}
	\sum_{\theta \in \Theta}\xi_\theta u_\theta^k(b^{k}_{\xi^*,\pi'})+\pi'(\xi, b^{k}_{\xi^*,\pi'})
	& \ge \sum_{\theta \in \Theta}\xi^*_\theta u_\theta^k(b^{k}_{\xi^*,\pi'})+\pi'(\xi^*, b^{k}_{\xi^*,\pi'})  -\epsilon/2 \\
	& \ge \sum_{\theta \in \Theta}\xi^*_\theta u_\theta^k(b^{k}_{\xi,\pi'})+\pi'(\xi^*, b^{k}_{\xi,\pi'})  -\epsilon/2 \\
	& \ge \sum_{\theta \in \Theta}\xi_\theta u_\theta^k(b^{k}_{\xi,\pi'})+\pi'(\xi, b^{k}_{\xi,\pi'}) - \epsilon \\
	& \ge \sum_{\theta \in \Theta}\xi_\theta u_\theta^k(a')+\pi'(\xi, a')  - \epsilon,
	\end{align*}
	for each $a' \in \A$.
	This shows that $ b_{\xi^*, \pi'}^{k} \in \mathcal{B}_{\xi, \pi'}^{k,\epsilon}$. 
	
	Let $\pi^*: \Delta_\Theta \times \mathcal{A} \to \mathbb{R}_{+}$ be the payment function prescribed by Proposition \ref{thm:dutting}. Then, we have that:

	%
	\begin{align*}
	\sum_{\theta \in \Theta}\xi_\theta u_\theta^s(b^{k}_{\xi,\pi})-\pi(\xi, b^k_{\xi,\pi})
	& \ge	\sum_{\theta \in \Theta}\xi_\theta u_\theta^s(b^{k}_{\xi,\pi^*})-\pi^*(\xi, b^k_{\xi,\pi^*}) 
	& (\text{Optimality of $\pi'$}) \\
	& \ge 	\sum_{\theta \in \Theta}\xi_\theta u_\theta^s(b^{k,\epsilon}_{\xi,\pi'})-\pi'(\xi, b^{k,\epsilon}_{\xi,\pi'}) - 2 \sqrt{\epsilon} 
	& (\text{By Proposition \ref{thm:dutting} }) \\
	& \ge \sum_{\theta \in \Theta}\xi_\theta u_\theta^s(b^{k}_{\xi^*,\pi'})-\pi'(\xi, b^k_{\xi^*,\pi'}) - 2 \sqrt{\epsilon}  \\
	& \ge \sum_{\theta \in \Theta}\xi^*_\theta u_\theta^s(b^{k}_{\xi^*,\pi'})-\pi'(\xi^*, b^k_{\xi^*,\pi'}) - 2 \sqrt{\epsilon} - \epsilon \\
	& = \sum_{\theta \in \Theta}\xi^*_\theta u_\theta^s(b^{k}_{\xi^*,\pi'})-\pi'(\xi^*, b^k_{\xi^*,\pi'}) - 3 \sqrt{\epsilon} \\
	& = \sum_{\theta \in \Theta}\xi^*_\theta u_\theta^s(b^{k}_{\xi^*,\pi})-\pi(\xi^*, b^k_{\xi^*,\pi}) - 3 \sqrt{\epsilon}.
	\end{align*} 
	This shows that the expected seller's utility decreases of at most $3\sqrt{\epsilon}$  when we consider sufficiently close posteriors. 	
	Hence, by Caratheodory's theorem we can decompose $\xi^*$ as follows:
	\begin{equation*}
	\sum_{{\xi'} \in \Xi(\xi) } \gamma^{\xi^*}_{\xi'}{\xi'_\theta} = \xi^*_\theta \hspace{3mm} \forall \theta \in \Theta
	\end{equation*}
	with $\gamma^{\xi^*} \in \Delta_{\Xi(\xi^*)}$, where we recall  that $\xi^* \in \text{co}(\Xi(\xi^*))$. We show now that such a decomposition decreases the expected seller's utility only by the desired amount. Formally, we have that:
	\begin{align*}
	\sum_{{\xi'} \in {\Xi(\xi)} } \gamma^{\xi^*}_{\xi'} \Big(\sum_{\theta \in \Theta}\xi'_\theta u_\theta^s(b^{k}_{\xi',\pi})+\pi(\xi', b^k_{\xi',\pi})  \Big) 
	&  \ge   \sum_{{\xi'} \in \Xi(\xi)}  \gamma^{\xi^*}_{\xi'} \Big(\sum_{\theta \in \Theta}\xi^*_\theta u_\theta^s(b^{k}_{\xi^*,\pi})+\pi(\xi, b^{k}_{\xi^*,\pi}) - 3 \sqrt{\epsilon} \Big)   \\
	& = \sum_{\theta \in \Theta}\xi^*_\theta u_\theta^s( b^k_{\xi^*,\pi})+\pi(\xi^*,  b^k_{\xi^*,\pi}) - 3 \sqrt{\epsilon} .
	\end{align*}
	Since $3 \sqrt{\epsilon} \le \alpha$, this concludes the proof.
\end{proof}

\conststatethmsecond*
\begin{proof}
	Given a constant $\alpha >0$ we let $(\gamma, \pi)$ be an optimal protocol. As a first step, we show that there exists a protocol $(\gamma^*, \pi^*)$ achieving a seller's expected utility of at least $\mathsf{APX}\ge \rho \mathsf{OPT}-2^{-\Omega(1/\rho)} -\alpha$, where $\mathsf{OPT}$ is the utility achieved with $(\gamma, \pi)$. Moreover, the payment function $\pi^*$ is a linear function with parameter $\beta \in \{1-2^{-i}\}_{i \in \{i,\dots,\lfloor \rho/2\rfloor\}}$.  
	We define a signaling scheme $\gamma^*$ supported in $\Xi_q$ as follows:
	\begin{equation*}
	\gamma_{\tilde{\xi}}^*=\sum_{\xi \in \text{supp}(\gamma)} \gamma_\xi  \gamma^{\xi}_{\tilde{\xi}} \hspace{4mm} \forall \tilde{\xi} \in \Xi_q,
	\end{equation*}
	where $\gamma^{\xi} \in \Delta_{\Xi_q}$ is the signaling scheme satisfying Lemma \ref{lem:actions_quniform} with $q= \frac{18d}{\alpha^2} $. First we observe that $\gamma^* \in \Delta_{\Xi_q}$ satisfies the consistency constraints, indeed we have:
	\begin{equation*}
	\sum_{\tilde \xi \in \Xi_q}\gamma_{\tilde{\xi}}^* \, \tilde{\xi_\theta}  
	= \sum_{\xi \in \text{supp}(\gamma)} \gamma_\xi \sum_{\tilde \xi \in \Xi_q}  \gamma^{\xi}_{\tilde{\xi}} \tilde{\xi}_\theta
	= \sum_{ \xi \in \text{supp}(\gamma)} \gamma_{\xi} \xi_\theta 
	= \mu_\theta \,\,\,\ \forall \theta \in \Theta.
	\end{equation*}
	Moreover, we can define as  $\pi^*:\Delta_\Theta \times \mathcal{A} \to \mathbb{R}_{+}$ as the payment function computed (in polynomial time) with Corollary~\ref{cor:linear} in each q-uniform posterior.
 	
 	Let $\pi'$ be the optimal payment function. We show that the protocol $(\gamma^*, \pi^*)$ achieves the desired approximation. Formally:
	\begin{align*} 
	\sum_{\tilde{\xi} \in \Xi_q}  \gamma_{\tilde{\xi}}^* \Big( \sum_{\theta \in \Theta}& \tilde \xi_\theta  u_\theta^s(b^{k}_{\tilde \xi,\pi^*})- \pi^*(\tilde{\xi},
		b^{k}_{\tilde \xi,\pi^*}) \Big) \quad\quad\quad\quad\quad\quad\quad\quad\quad\quad\quad\quad\quad\quad\quad\quad\quad\quad\quad\quad\quad\\
	&\specialcell{ \ge \rho \left( \sum_{\tilde{\xi} \in \Xi_q}\gamma^*_{\tilde{\xi}} \Big( \sum_{\theta \in \Theta} \tilde \xi_\theta u_\theta^s(b^{k}_{\tilde \xi,\pi'})-\pi'(\tilde{\xi},
		b^{k}_{\tilde \xi,\pi'}) \Big) \right) - 2^{\Omega(1/\rho)} \hfill
		\textnormal{(Corollary~\ref{cor:linear})}}
	\\
	& \specialcell{= \sum_{\xi \in \text{supp}(\gamma)} \gamma_\xi \Big( \sum_{\tilde{\xi} \in \Xi_q }
		\gamma^{\xi}_{\tilde{\xi}}  \big(\sum_{\theta \in \Theta} \tilde{\xi}_\theta u_\theta^s(b^{k,\epsilon}_{\tilde \xi,\pi})-\pi(\tilde \xi, b^{k,\epsilon}_{\tilde \xi,\pi})\big) \Big) \hfill \textnormal{(By defintion of $\gamma^*$)} }
	\\
	& \specialcell{\ge \sum_{\xi \in \text{supp}(\gamma)} \gamma_{\xi} \Big(  \sum_{\theta \in \Theta} \xi_\theta u_\theta^s(b^{k}_{ \xi,\pi})-\pi(\xi, b^{k}_{ \xi,\pi})\Big) -\alpha 
		\hfill \textnormal{(By Lemma \ref{lem:states_quniform})}}.
	\end{align*}
	%
	
	 
	 This implies that since $(\gamma^*,\pi^*)$ is feasible for the following LP, it has value at least $\rho \textnormal{OPT}- 2^{\Omega(1/\rho)} -\alpha$.

	\begin{align*}
	\max_{\gamma\ge 0}& \sum_{k \in \K} \lambda_k \sum_{\xi \in \Xi_q} \gamma_\xi \sum_{\theta \in \Theta} \xi_\theta u^s_\theta(b^{k}_{ \xi,\pi^*}) - \pi^*(\xi, b^{k}_{ \xi,\pi^*}) \,\,\ \text{s.t.} \\
	& \sum_{ \xi \in \text{supp}(\gamma)} \gamma_\xi \xi_\theta = \mu_\theta \,\,\,\ \forall \theta \in \Theta.
	\end{align*}

	Hence, to find the desired approximation it is sufficient to compute $\pi^*$ in each $q$-uniform posterior and solve the LP. 
	Note that since $|\Xi_q|=O(q^d)$, the computation of the payment function $\pi^*:\Delta_\Theta \times \mathcal{A} \to \mathbb{R}_{+}$ and the computation of the previous LP require polynomial time for each fixed $\alpha>0$. 
	
	Finally, to prove the second part of the statement it is sufficient to notice that $\pi^*$ is optimal with respect to the desired set of linear payment functions.
\end{proof}

%% file: content/appendix_hardness_second.tex
\section{Proofs Omitted from Section~\ref{sec:selection_no_limited_liability}}
\hardnessfirst*
	\begin{proof}
	We introduce a reduction from \textsc{Lineq-Ma}$(1 -\epsLM, \delta)$ to the design of the optimal protocol, showing that for $\epsLM$ and $\deltaLM$ small enough, the following holds:
	\begin{itemize}
		\item \emph{Completeness}: If an instance of \textsc{Lineq-Ma}$(1 -\epsLM, \delta)$ admits a $1-\epsLM$ fraction of satisfiable equations when variables are restricted to lie in the hypercube $\{0,1\}^{\nVar}$, then there exists a protocol that provides to the seller's expected utility at least of $\eta$, where $\eta$ will be defined in the following;
		\item \emph{Soundness}: If at most a $\deltaLM$ fraction of the equations can be satisfied, then every protocol provides to the seller's expected utility at most $\eta-c$, where $c$ is a \emph{constant} defined in the following.
	\end{itemize}
	In the rest of the proof, given a vector of variables $\mathbf{x} \in \mathbb{Q}^{\nVar}$, for $i \in [\nVar]$, we denote with $x_i$ the component corresponding to the $i$-th variable.
	Similarly, for $j \in [\nEq]$, $c_j$ is the $j$-th component of the vector $\mathbf{c}$, whereas, for $i \in [\nVar]$ and $j \in [\nEq]$, the $(j,i)$-entry of $\mathbf{A}$ is denoted by $A_{ji}$.
	
	\paragraph{Reduction}
	As a preliminary step, we normalize the coefficients by letting $\mathbf{\bar A} \defeq \frac{1}{\normaliz} \mathbf{A}$ and $\mathbf{\bar c} \defeq \frac{1}{\normaliz^2} \mathbf{c}$, where we let $\normaliz \defeq 2 M\max\left\{ \max_{i \in [\nVar], j \in [\nEq]} \ {A}_{ji}, \max_{j \in [\nEq]} \ c_j, \nVar^2 \right\}$ and $M$ will be defined in the following.
	It is easy to see that the normalization preserves the number of satisfiable equations.
	Formally, the number of satisfied equations of $\mathbf{A}\mathbf{x}=\mathbf{c}$ is equal to the number of satisfied equations of  $\mathbf{\bar A} \mathbf{\bar x} = \mathbf{\bar c}$, where $\mathbf{\bar x }=\frac{1}{\normaliz} \mathbf{x}$.
	For every variable $i \in [\nVar]$, we define a state of nature $\theta_{i} \in \Theta$.
	Moreover, we introduce three additional states $ \theta_0, \theta_1,\theta_2 \in \Theta$.
	The prior distribution $\mu \in \textnormal{int}(\Delta_\Theta)$ is defined in such a way that $\mu_{\theta_i} =\frac{1}{2\nVar^2}$ for every $i \in [\nVar]$, while $\mu_{ \theta_0} = \frac{\nVar-1}{2\nVar}$, $\mu_{ \theta_1} = \frac{1}{4}$, and $\mu_{ \theta_2} = \frac{1}{4}$  (notice that $\sum_{\theta \in \Theta} \mu_\theta = 1$).
	We define four buyer's types $k^1_j,k^2_j,k^3_j,k^4_j \in \K$ for each equation $j \in [\nEq]$, where the probability of observing each buyer's type is $\frac{1}{8\nEq}$. Moreover, we define an additional type $k^\star$. All the types $k \in \K$ have budget $b_k=\nu/2$, where $\nu$ will be defined in the following.
	The buyer has $9$ actions available, namely $\A \defeq \{ a_0, a_1, a_2, a_3, a_4, a_5,a_6, a_7,a_8 \}$. 
	Then, we define the utilities of the players, where the utility is $0$ when not specified.
	For each $k^1_j, j \in [\nEq]$, the utilities are:
	\begin{itemize}
		\setlength\itemsep{0.4em}
		\item $u^{k^1_j}_{\theta_i}(a_0)=\frac{1}{2}$ for each $i \in [\nVar]$, 
		\item$u^{k^1_j}_{\theta_i}(a_1)=\frac{1}{2}-\bar A_{j i}+\bar c_{j}$ for each $i \in [\nVar]$,
		\item$u^{k^1_j}_{\theta_i}(a_2)=\frac{1}{2}+\bar A_{j i}-\bar c_{j}$ for each $i \in [\nVar]$
		\item $u^{k^1_j}_{\theta_0}(a_0)=\frac{1}{2}$, 
		\item $u^{k^1_j}_{\theta_0}(a_1)=\frac{1}{2}+\bar c_{j}$,  
		\item $u^{k^1_j}_{\theta_0}(a_2)=\frac{1}{2}-\bar c_{j}$.
		\item $u^{k^1_j}_{\theta_1}(a_3)=\frac{1}{2}+2 \nu$, 
	\end{itemize}
	For each $k^2_j, j \in [\nEq]$, the utilities are:
	\begin{itemize}
			\setlength\itemsep{0.4em}
		\item$u^{k^2_j}_{\theta_i}(a_0)=\frac{1}{2}-\bar A_{j i}+\bar c_{j}$ for each $i\in [\nVar]$,
		\item$u^{k^2_j}_{\theta_i}(a_7)=\frac{1}{2}$ for each $i \in [\nVar]$
		\item $u^{k^2_j}_{\theta_0}(a_0)=\frac{1}{2}+ \bar c_{j}$, 
		\item $u^{k^2_j}_{\theta_1}(a_3)=\frac{1}{2}+2 \nu$, 
	\end{itemize}
	For each type $k^3_j$, $j \in [\nEq]$ the utilities are:
	\begin{itemize}
			\setlength\itemsep{0.4em}
		\item$u^{k^3_j}_{\theta_i}(a_0)=\frac{1}{2}+\bar A_{j i}-\bar c_{j}$ for each $i \in [\nVar]$,
		\item$u^{k^3_j}_{\theta_i}(a_7)=\frac{1}{2}$ for each $i \in [\nVar]$
		\item $u^{k^3_j}_{\theta_0}(a_1)=\frac{1}{2}-\bar c_{j}$,  
		\item $u^{k^3_j}_{\theta_1}(a_3)=\frac{1}{2}+2 \nu$, 
	\end{itemize}
For each type $k^4_j$, $j \in [\nEq]$ the utilities are equivalent to the one of type $k^1_j$ but with the following differences:
	\begin{itemize}
			\setlength\itemsep{0.4em}
		\item  $u^{k_j}_{\theta}(a_5)=\frac{1}{2}$ for each $\theta \in \Theta$,
		\item  $u^{k_j}_{\theta_i}(a_7)=\frac{1}{2}$ for each $i \in [\nVar]$,
	\end{itemize}
	Finally, the utilities of type $k^\star$ are:
	\begin{itemize}
			\setlength\itemsep{0.4em}
		\item $u^{k^\star}_{\theta_1}(a_6)=1$,
		\item $u^{k^\star}_{\theta}(a_1)=1$ for each $\theta \in \Theta$.
	\end{itemize}
	Moreover, we let $u^{k}_{\theta}(a_8)=\frac{1}{2}$ for every $k \in K$ and $\theta \in \Theta$. %
	%
	Finally, the utility of the seller is:
	\begin{itemize}
			\setlength\itemsep{0.4em}
		\item $u^\mathsf{s}_{\theta} (a_6) = \frac{1}{4}$ for each $\theta \in \Theta$,
		\item$u^\mathsf{s}_{\theta} (a_5) = 4\nu$ for each $\theta \in \Theta$,
		\item $u^\mathsf{s}_{\theta} (a_0) = \nu$ for each $\theta \in \Theta$,
		\item $u^\mathsf{s}_{\theta} (a_7) = 2 \nu$ for each $\theta \in \Theta$.
	\end{itemize}
	We recall that the utility is $0$ when not defined explicitly.
	
	\paragraph{Completeness.}
	Suppose that there exists a vector $\mathbf{\hat x} \in \{0,1\}^{\nVar}$ such that at least a fraction $1-\epsLM$ of the equations in $\mathbf{A} \mathbf{\hat x}=\mathbf{c}$ are satisfied.
	Let $X^1\subseteq [\nVar]$ be the set of variables $i \in [\nVar]$ with $\hat x_i = 1$, while $X^0 \defeq [\nVar]\setminus X^1$.
	Given the definition of $\mathbf{\bar A}$ and $\mathbf{\bar c}$, there exists a vector $\mathbf{\bar x} \in \{0,\frac{1}{\normaliz}\}^{\nVar}$ such that at least a fraction $1-\epsLM$ of the equations in $\mathbf{\bar A} \mathbf{\bar x}= \mathbf{\bar c}$ are satisfied, and, additionally, $\bar x_i = \frac{1}{\normaliz}$ for all the variables in $i \in X^1$, while $\bar x_i = 0$ whenever $i \in X^0$.
	Let us consider an (indirect) signaling scheme $\phi : \Theta \to \Delta_\sset$ where the set of signals is $\sset \defeq \{s_1, s_2,s_3\}$.
	Let $q \defeq \frac{\nVar(\nVar-1)}{\normaliz-|X^1|}$.
	For each $i \in [\nVar]$, let $\phi_{\theta_i}(s_1) = q$ and  $\phi_{\theta_i}(s_2)=1-q$ if $i \in X^1$, while $\phi_{\theta_i}(s_2)=1$ otherwise.
	Moreover, let $\phi_{\theta_0}(s_1)=1$, $\phi_{\theta_1}(s_3)=1$ and $\phi_{\theta_2}(s_2)=1$.
	Then, all the other probabilities $\phi_\theta(s)$ are set to $0$. It is easy to see that the signaling scheme is feasible. 
	Moreover, we set the price $p=\nu/2$.
	Finally, we set $\pi(s_3,a_6)=2\nu$ and all the other payments $\pi(s,a)=0$.
	
	%
	
	Now, we compute the expected seller's utility due of each type of buyer.
	\begin{itemize}
		\item The buyer of type $k^\star$ in the posterior $\xi^{s_3}$ plays the action $a_6$ and gets utility $\sum_{\theta \in \theta}\xi^{s_3}_\theta u^{k^\star}_\theta(a_6)+\pi(s_3,a_6)=1+2\nu$. Moreover, in the other posteriors $\xi^{s_1}$ and $\xi^{s_2}$ the seller's utility is at least $0$. Finally, the protocol is IR for the buyer since the expected utility declining the protocol is $1$ while accepting it is $-\pi/2+1\cdot \frac{3}{4}+(1+2\nu)\frac{1}{4}=1$. Hence, the expected principal utility when the buyer's type is $k^\star$ is at least $\sum_{s \in \sset}\sum_{\theta \in \theta}\xi^{s}_\theta u^{k^\star}_\theta(b^{k^\star}_{\xi^s,\pi})=\frac{1}{16}$.
	\item Consider a buyer $k^1_j$, $j \in [\nEq]$, such that the j-th equality is satisfied by the vector $\hat \xvec$.
	Now, let us take the buyer's posterior $\xi^{s_1} \in \Delta_\Theta$ induced by the signal $s_1$. 
	Let $h \defeq \frac{\frac{q}{\nVar^2}}{\sum_{i \in X^1}\frac{q}{\nVar^2}+\frac{\nVar-1}{\nVar}}$. 
	Then, using the definition of $\xi^{s_1}$, it is easy to check that $\p^1_{\theta_i}=h$ for every $i \in X^{s_1}$, $\p^{s_1}_{\theta_i}=0$ for every $i \in X^0$, while $\p^{s_1}_{ \theta_0} = \frac{\frac{\nVar-1}{\nVar}}{\sum_{i \in X^1}\frac{q}{\nVar^2}+\frac{\nVar-1}{\nVar}} = 1 - h \left| X^1 \right|$.
	%
	%
	The buyer of type $k_j \in\K$ experiences a utility of $\sum_{\theta \in \Theta} \p^{s_1}_\theta u^{k_j}_{\theta} (a_0)= \frac{1}{2}$ by playing action $a_0$.
	Instead, the utility she gets by playing $a_1$ is defined as follows:
	\begin{align*}
	\sum_{\theta \in \Theta} \p^{s_1}_\theta u^{k_j}_{\theta} (a_1) & = \sum_{i \in X^1} h \left( \frac{1}{2} - \bar A_{j i} + \bar c_{j} \right)+ \p^1_{ \theta_0} \left( \frac{1}{2} + \bar c_j \right)= \\
	& = h \left|X^1 \right| \left(\frac{1}{2} + \bar c_j \right) - h \sum_{i \in X^1} \bar A_{j i} +  \left(1- h \left| X^1 \right|  \right) \left(\frac{1}{2} + \bar c_j \right) = \\
	& = \frac{1}{2}+ \bar c_{j}-h\sum_{i \in X^1} \bar A_{j i } = \frac{1}{2}+\bar c_{j}-\frac{1}{\normaliz}\sum_{i \in X^1} \bar A_{j i }= \frac{1}{2},
	\end{align*}
	where the second to last equality holds since $h = \frac{1}{\normaliz}$ (by definition of $h$ and $q$), while the last equality follows from the fact that the $j$-th equation is satisfied, and, thus, $\frac{1}{\normaliz}\sum_{i \in X^1} \bar A_{j i } =  \bar c_j$ (recall that $\bar x_i = \frac{1}{\normaliz}$ for all $i \in X^1$).
	Using similar arguments, we can write $\sum_{\theta \in \Theta} \p^{s_1}_\theta u^{k_j}_{\theta} (a_2) = \frac{1}{2}$.
	Moreover, all the other actions have utility $0$.
	Hence, the buyer plays $a_0$ in the posterior $\xi^{s_1}$.
	In  posterior $\xi^{s_2}$ induced by signal $s_2$, the utility of each action different from $a_8$ is strictly smaller than $\frac{1}{2}$. Hence, the buyer will play $a_8$,
	while in posterior $\xi^{s_3}$ induced by signal $s_3$, the utility of action $a_3$ is $\frac{1}{2}+2 \nu $ and the buyer will play $a_3$.
	Hence the expected utility of the buyer is $\frac{3}{4}\frac{1}{2}+ \frac{1}{4} (\frac{1}{2}+2 \nu)-p=\frac{1}{2}$.
	Moreover, the protocol is IR for the buyer since if she declines the protocol the utility is $\frac{1}{2}$ while if she accepts the protocol the utility is $\frac{1}{2}$.  Hence, when the buyer's type is $k^1_j$ the expected seller's utility is $\nu \mu_{\theta_0}+\nu/2$.
	
	\item	Consider a buyer $k^2_j$ or $k^3_j$ , $j \in \nEq$ such that the j-th equality is satisfied.
	A similar argument as before shows that in posterior $\xi^{s_1}$ the buyer's optimal action is $a_7$, while in posterior $\xi^{s_2}$, the optimal action is $a_8$.
	In posterior $\xi^{s_3}$, the optimal action is $a_3$.
	Hence,  the expected buyer's utility is $\frac{3}{4}\frac{1}{2}+ \frac{1}{4} (\frac{1}{2}+2 \nu)-p=\frac{1}{2}$.
	Hence, the protocol is IR for the buyer and provides expected seller's utility at least $2\nu \mu_{\theta_0}+\nu/2$.
	
	\item Consider a buyer $k^4_j$, $j \in [\nEq]$ such that the j-th equality is satisfied.
	The buyer has an utility similar to $k^1_j$ and plays the same best responses. Hence, it is indifferent in participating or not participating to the protocol. We assume that they brake ties in favor of the seller and does not accept. 
	She plays action $a_5$ and the expected seller's utility is $2\nu$.
	\end{itemize}
	
	Since all the other buyer's types provide positive utility ---it never happens that the expected payment from the seller to the buyer exceeds the payment from the buyer to the seller---, the expected seller's utility is at least 
	\[\eta=\frac{1}{32}+(1-\epsLM)\frac{1}{8}( \mu_{\theta_0}\nu+ \nu/2)+(1-\epsLM)\frac{1}{4} ( \mu_{\theta_0} 2\nu+ \nu/2)+ (1-\epsLM)\frac{1}{8} 2\nu\]

	%

	%

	\paragraph{Soundness}

	As a first step, we upperbound the expected seller's utility from each type.
	It is easy to see that the maximum expected utility that the seller can extract from the buyer's type $k^\star$ is at most $\frac{1}{32}$.
	Moreover, the maximum expected utility that the seller can extract from a buyer of type $k^1_j$, $j \in [\nEq]$, is at most $\frac{1}{8}(\nu)$.
	The maximum  expected utility that the seller can extract from a buyer of type $k^2_j$ or $k^3_j$, $j \in [\nEq]$  is at most  $\frac{1}{4}\frac{3}{2}\nu$.
	Finally, the maximum  expected utility that the seller can extract from a buyer of type $k^4_j$, $j \in [\nEq]$, is $\frac{1}{8}  2\nu$.

	%
	Using the previous upperbounds, we can bound the component of the utility due to each set of types.
	For each constant $t<1$, there exist constants $c=c(t)$, $\epsLM=\epsLM(t)$ such that if the expected utility is greater than $\eta-c$ then the expected utility from types $k^1_j$, $j \in [\nEq]$, is at least  $t \frac{1}{8}(\nu)$, the expected utility from types $k^2_j$, $j \in [\nEq]$, and $k^3_j$, $j \in [\nEq]$, is at least $t \frac{1}{4}\frac{3}{2}\nu$, and the expected utility from types $k^4_j$, $j \in [\nEq]$, is at least $t \frac{1}{8}  2\nu$.
	To see that, consider for instance the types $k^1_j$, $j \in [\nEq]$. It must hold:
	\[\frac{1}{32}+ \bar t\frac{1}{8}\nu+\frac{1}{4} \frac{3}{2}\nu+ \frac{1}{8} 2\nu\ge \frac{1}{32}+(1-\epsLM)\frac{1}{8}( \mu_{\theta_0}\nu+ \nu/2)+(1-\epsLM)\frac{1}{4} ( \mu_{\theta_0} 2\nu+ \nu/2)+ (1-\epsLM)\frac{1}{8} 2\nu \]
	Since for $\nVar$ large enough $\mu_{\theta_0}$ is close to $\frac{1}{2}$, for $c(t)$, $\epsLM(t)$ small enough constant the equation is satisfied for $\bar t\ge t$.
	A similar result holds for every other set of types $k^2_j$ with $j \in [\nEq]$, $k^3_j$ with $j \in [\nEq]$, and $k^4_j$ with $j \in [\nEq]$.
	
	The next step is to show the existence of a posterior in which a $t$ fraction of agent of types $k^1_j$, $j \in [\nEq]$, play $a_0$ and the the same holds for each other set of types $k^2_j$,$k^3_j$ with action $a_7$.
	Suppose by contradiction that there is no posterior in which a $t$ fraction of $k^1_j$, $j \in [\nEq]$,  plays $a_0$.
	First, notice that the maximum payment is at most $p=\nu/2+\frac{1}{M}$, otherwise all the buyer's types $k^1_j$ are not IR.
	Moreover, the seller's utility minus payment is greater than $0$ in a posterior only if the agent plays $a_0$.
	Finally, it is easy to see that it is sufficient to consider signaling schemes that induce posteriors such that if  $\xi_{\theta_i}>0$, then $\xi_{\theta_1}=0$ and $\xi_{\theta_2}=0$ since states $\xi_{\theta_1}$ and $\xi_{\theta_2}$ disincentivize the actions with high seller's utility.
	Hence, the maximal utility from agents of types $k^1_j$ is  at most
	\[ \frac{1}{8} \left[\nu/2+\frac{1}{M}+ (t-1/\nEq) \frac{1}{2} \nu \right] < t \frac{1}{8}\nu,\]
	for $M$ large enough, reaching a contradiction.
	A similar argument holds for the other types.
%
%
	%
	%
	%
	%
	%
	%
	This implies that there exists a set $Q \subseteq [\nEq]$ and a posterior $\xi$ such that for each $j \in Q$ all the buyers $k^1_j$, $k^2_j$, and $k^3_j$ in the posterior play $a_0$,$a_7$, and $a_7$, respectively.
	Notice that $|Q|\ge 1-3(1-t)$ and for $t$ large enough $|Q|>\delta$.

	Suppose that there exists a signal inducing a posterior $\xi \in \Delta_\Theta$ in which all the buyer's types $k^1_j$, $j \in Q$ best respond by playing action $a_0$.
	We show that there exists at least one $j \in Q$ such that  it holds $\sum_{\theta \in \Theta} \p_\theta u^{k^1_j}_{\theta} (a_1) > \sum_{\theta \in \Theta} \p_\theta u^{k^1_j}_{\theta} (a_0)$ or $\sum_{\theta \in \Theta} \p_\theta u^{k^1_j}_{\theta} (a_2) > \sum_{\theta \in \Theta} \p_\theta u^{k^1_j}_{\theta} (a_0)$.
	For every buyer's type $k^1_j \in \K$, it holds $\sum_{\theta \in \Theta} \p_\theta u^{k_j}_{\theta} (a_0)=\frac{1}{2}$.
	Moreover, it is the case that:
	\[
	\sum_{\theta \in \Theta} \p_\theta u^{k^1_j}_{\theta} (a_1)=\sum_{i \in [\nVar]} \p_{\theta_i} \left( \frac{1}{2}-\bar A_{j i}+\bar c_{j} \right) + \p_{ \theta_0}\left(\frac{1}{2}+\bar c_{j} \right)=\frac{1}{2}+\bar c_{j}-\sum_{i \in [\nVar]} \p_{\theta_i}	\bar A_{j i}.
	\]
	Similarly, it holds:
	\[ 
	\sum_{\theta \in \Theta} \p_\theta u^{k^1_j}_{\theta} (a_2)=\frac{1}{2}-	\bar c_{j}+\sum_{i \in [\nVar]} \p_{\theta_i} \bar A_{j i}.
	\]
	Suppose by contradiction that for every type $k^1_j,$ $j \in Q$, it is the case that $\sum_{\theta \in \Theta} \p_\theta u^{k^1_j}_{\theta} (a_0) \ge  \sum_{\theta \in \Theta} \p_\theta u^{k^1_j}_{\theta} (a_1)$, which implies that $\bar c_{j}-\sum_{i \in [\nVar]} \p_{\theta_i}\bar A_{j i}\le0$, whereas it holds $\sum_{\theta \in \Theta} \p_\theta u^{k^1_j}_{\theta} (a_0)\ge \sum_{\theta \in \Theta} \p_\theta u^{k_j}_{\theta} (a_2)$, implying $-\bar c_{j}+\sum_{i \in [\nVar]} \p_{\theta_i}\bar A_{j i}\le0$.
	Thus, $\sum_{i \in [\nVar]} \p_{\theta_i}\bar A_{j i} =\bar c_{j}$ for every $j \in Q$ and the vector $\mathbf{\hat x} \in \mathbb{Q}^{\nVar}$ with $\hat x_i = \p_{\theta_{i}}$ for all $i \in [\nVar]$ satisfies at least a fraction $\delta$ of the equations, reaching a contradiction.
	Since we have that $t$ types $k^1_j$ play $a_0$, this implies that $\pi(\xi,a_0)>0$. However, at the same time we have that the  buyers of type $k^2_j$ and $k^3_j$ plays action $a_7$. Consider a $j^* \in Q$ such that  $\sum_{\theta \in \Theta} \p_\theta u^{k^1_{j^*}}_{\theta} (a_1) > \sum_{\theta \in \Theta} \p_\theta u^{k^1_{j^*}}_{\theta} (a_0)$ or $\sum_{\theta \in \Theta} \p_\theta u^{k^1_{j^*}}_{\theta} (a_2) > \sum_{\theta \in \Theta} \p_\theta u^{k^1_{j^*}}_{\theta} (a_0)$.  Recall that this buyer must play $a_7$. If the first inequality holds then it must hold $\sum_{\theta \in \Theta} \p_\theta u^{k^2_{j^*}}_{\theta} (a_7) + \pi(\xi,a_7) \ge \sum_{\theta \in \Theta} \p_\theta u^{k^2_{j^*}}_{\theta} (a_0) + \pi(\xi,a_0)$. Moreover, $\sum_{\theta \in \Theta} \p_\theta u^{k^2_{j^*}}_{\theta} (a_7) = \sum_{\theta \in \Theta} \p_\theta u^{k^1_{j^*}}_{\theta} (a_0)  < \sum_{\theta \in \Theta} \p_\theta u^{k^1_{j^*}}_{\theta} (a_1)  = \sum_{\theta \in \Theta} \p_\theta u^{k^2_{j^*}}_{\theta} (a_0) $, implying  $\pi(\xi,a_7)>\pi(\xi,a_0)$.
	A similar argument holds for the buyer $k^3_{j^*}$ if the second inequality is satisfied.
	This implies that type $k^4_{j^*}$ can play the same best responses of player $k^1_j$ in any posterior different from $\xi$ and play action $a_7$ in $\xi$.
	Hence, the expected utility of buyer $k^4_{j^*}$ is strictly greater than the one of $k^1_{j^*}$ (that is IR), and hence it is strictly IR. 
	
	We conclude the proof showing that the utility of this buyer's type is too small, reaching a contradiction.
	First, notice that the seller must induce a posterior with $\xi_{\theta_1}\ge \frac{3}{4}$ with probability at least $\frac{1}{8}$. In all the other posteriors the seller's utility from type $k^*$ is $0$.
	However, it must hold that the utility from type $k^\star$ is at least $\frac{1}{64}$ for $\nu$ small enough. Hence, playing posteriors with $\xi_{\theta_1}\ge \frac{3}{4}$ with probability smaller than $\frac{1}{8}$ the seller's utility form type $k^*$ is at most $\frac{1}{2}\frac{1}{4} \frac{1}{8} <\frac{1}{64}$.
	Now consider the type $k^4_{j^*}$ that is IR.
	In a posterior $\xi$ with $\xi_{\theta_1}\ge \frac{3}{4}$, the seller's utility when the type is $k^4_{j^*}$ is at most $0$. Hence, the total utility from this type is at most $p+ \frac{7}{8}4\nu\le\nu+1/M$, where the last inequality follows by the fact that the payment is at most $\frac{\nu}{2}+1/M$.
	For $|Q|$ large enough, we have that a $|Q|/\nEq-\delta$ fraction of types $k^4_{j}$ provide seller's utility at most $\nu+1/M$. Hence, the total utility from type $k^4_{j}$ is at most $\frac{1}{8} [(|Q|/\nEq-\delta) (\nu+1/M)+  (1- (|Q|/\nEq-\delta)) 2\nu ] \le t \frac{1}{8} 2\nu$.
	Thus, we reach a contradiction.
\end{proof}

%% file: content/appendixGeneral.tex
\constactionsthmsecond*
\begin{proof}
	Let $(\phi,p,\pi)$ be an optimal protocol.
	 Then, the seller's expected utility is given by:
	\[\sum_{k \notin \R_{\phi,p,\pi}} \lambda_k \sum_{\theta \in \Theta} \mu_\theta u^s_\theta(b^k_\mu)+\sum_{k \in \R_{\phi,p,\pi}} \lambda_k \left[  \sum_{s \in \mathcal{S}} \sum_{\theta \in \Theta} \mu_\theta \phi_\theta(s) \left(u^s_\theta(b^k_{\xi^s,\pi}) -\pi(s,b^k_{\xi^s,\pi})\right) +p \right], \]
	where we recall that $\R_{\phi,p,\pi}$ is the set of buyer's types for which the IR constraint is satisfied under protocol $(\phi,p,\pi)$.
	Given a signal $s \in \mathcal{S}$ and a type $k\in \K$, let $b^k_{\xi^s}\in \arg \max_{a \in \mathcal{A}} \sum_{\theta \in \Theta} \mu_\theta \phi_\theta(s) u^{k}_\theta(a)$. Intuitively, $b^k_{\xi^s}$ is an optimal action for the buyer without considering the payment function.
 	Then, the seller's utility can be spitted in three components:
 	\begin{enumerate}
 		\item[(i)] The utility from the buyer's types that are not IR
	\[U_1\defeq \sum_{k \notin \mathcal{R}_{\phi,p,\pi}}  \lambda_{k}  \sum_{\theta \in \Theta} \mu_\theta u^s_\theta(b^k_\mu) ;\]
	\item[(ii)] The maximum seller's utility deriving from the buyer's action \[U_2\defeq \sum_{k \in \R_{\phi,p,\pi}} \lambda_k  \left[ \sum_{s \in \mathcal{S}} \sum_{\theta \in \Theta} \mu_\theta \phi_\theta(s)  \left( u^s_\theta(b^k_{\xi^s}) + u^k_\theta(b^k_{\xi^s,\pi}) - u^k_\theta(b^k_{\xi^s})\right)\right] ,\]
	where we use the fact that to incentivize action $b^k_{\xi^s,\pi}$ over  $b^k_{\xi^s}$ the payment must be at least $\frac{\sum_{s \in \mathcal{S}} \sum_{\theta \in \Theta} \mu_\theta \phi_\theta(s)\left(u^k_\theta(b^k_{\xi^s})-u^k_\theta(b^k_{\xi^s,\pi})\right)}{\sum_{\theta \in \Theta} \mu_\theta \phi_\theta(s)}$;
	\item[(iii)]  The utility related to the overall payment that the seller's can extract from the buyer given the price function $\pi$
	 \[U_3\defeq \sum_{k \in \R_{\phi,p,\pi}}  \lambda_k \left[  p- \sum_s \sum_\theta \mu_\theta \phi_\theta(s) \left[\pi(s,b^k_{s,\pi})  + u^k_\theta(b^k_{s,\pi}) - u^k_\theta(b^k_{\xi^s})\right]\right].\]
\end{enumerate}

	Notice that the term $U_2+U_3$ is the utility deriving from buyer's types for which the IR constraint is satisfied, where we add, respectively subtract, the term
	\[\sum_{k \in \R_{\phi,p,\pi}}\lambda_k  \left[ \sum_{s \in \mathcal{S}} \sum_{\theta \in \Theta} \mu_\theta \phi_\theta(s)  \left(u^k_\theta(b^k_{\xi^s,\pi}) - u^k_\theta(b^k_{\xi^s})\right)\right] \]
	  to $U_2$, respectively $U_3$.

	In the following, we design three protocols $(\phi^1,p^1,\pi^1)$, $(\phi^2,p^2,\pi^2)$, and $(\phi^3,p^3,\pi^3)$, each with seller's utility that approximates the corresponding utility terms  $U_1$, $U_2$, and $U_3$. We will show that this will implies that at least one protocol provides a good approximation of the overall seller's utility, \emph{i.e.}, of $U_1+ U_2+ U_3$.
	
	\paragraph{Approximate $U_1$.}
	The protocol $(\phi^1,p^1,\pi^1)$ that provides no information, charges no price, and does not provides any payment has seller's utility 
	\[\sum_{k \in \K} \sum_{\theta} \mu_\theta u^s_\theta(b^k_\mu) \ge \sum_{k \notin \R_{\phi,p,\pi}} \sum_{\theta} \mu_\theta u^s_\theta(b^k_\mu) =U_1\]
	
	\paragraph{Approximate $U_2$.}
	By Corollary~\ref{cor:linear}, we know that for each signal $s \in \mathcal{S}$ (inducing a posterior $\xi^s$) and $\rho\in (0,1/2]$, there exists a linear contract $\pi'(s,\cdot)$ such that $\pi'(s,a)= \beta \sum_{\theta \in \Theta}\xi^s_{\theta} u^s_\theta(a)$ with parameter $\beta=1-2^{-i}$, $i \in \{1, \dots,\lfloor\frac{1}{2\rho}\rfloor \}$ that guarantees:
	\begin{subequations}\label{U2}
	\begin{align}
	 \sum_{k\in \K}  \lambda_k& \Big[\sum_{\theta \in \Theta} \xi^s_\theta \left( u^s_\theta(b^k_{\xi^s, \pi'})-  \pi'(\xi^s,b^k_{ \xi^s,\pi'}) \right) \Big]   \\
	 & \ge \rho \sum_{k \in \K}  \lambda_k  \left[ \sum_\theta \xi^s_\theta  \left( u^s_\theta(b^k_{\xi^s,\pi})  + u^k_\theta(b^k_{\xi^s,\pi}) - u^k_\theta(b^k_{\xi^s})  \right) \right]  -2^{-\Omega(1/\rho)} \\
	&\ge \rho \sum_{k \in \mathcal{R}_{\phi,p,\pi}}  \lambda_k  \left[ \sum_\theta \xi^s_\theta \left( u^s_\theta(b^k_{\xi^s,\pi})  + u^k_\theta(b^k_{\xi^s,\pi}) - u^k_\theta(b^k_{\xi^s})  \right) \right]  -2^{-\Omega(1/\rho)} 
	\end{align}
	\end{subequations}
	where the first inequality comes from Corollary~\ref{cor:linear}, and the last one since we restrict the elements in the first summation.
	
	Now, we need a protocol $(\phi^2,p^2,\pi^2)$ that approximate the utility obtained by the optimal protocol that uses only linear payment functions.
	When the number of states is fixed, we can approximate the optimal protocol that uses linear payment functions using Theorem~\ref{thm:fixedStates} with an additive loss $\alpha$.
	Otherwise, we can use Theorem \ref{QPTAS} that is polynomial time when the number of actions is fixed, while it runs in quasi-polynomial time and provides a loss $\alpha$ when instantiated with sufficiently small parameters.
	Hence, protocol $(\phi^2,p^2,\pi^2)$  can be computed in time $\textnormal{poly}(\min\{\mathcal{I}^d,\mathcal{I}^{log(m)}\})$. 
	Notice that both the algorithms returns a protocol such that $p=0$ and hence $p^2=0$.
	Then, we can show that the protocol $(\phi^2,p^2,\pi^2)$ has seller's utility  
	\begin{align*}
	&\sum_{k \in \K} \lambda_k \left[\sum_{s \in \sset} \sum_{\theta} \mu_\theta \phi_{\theta}(s) \left( u^s_\theta(b^k_{\xi^s, \pi^2})- \pi^2(s,b^k_{\xi^s, \pi^2})\right)\right] \\
	&\hspace{0.5cm}\ge \sum_{k \in \K} \lambda_k \left[\sum_{s \in \sset} \sum_{\theta} \mu_\theta \phi_{\theta}(s) \left(u^s_\theta(b^k_{\xi^s, \pi'})- \pi'(s,b^k_{\xi^s, \pi'}\right)\right]-\alpha   \\
	&\hspace{0.5cm}= \sum_{k \in \K}  \lambda_k \left[\sum_{s \in \sset} \left(\sum_{\theta} \mu_\theta \phi_{\theta}(s)\right) \sum_{\theta}  \xi^s_\theta \left(u^s_\theta(b^k_{\xi^s, \pi'})- \pi'(s,b^k_{\xi^s, \pi'}\right)\right]-\alpha   \\
	& \hspace{0.5cm}= \sum_{s \in \sset} \left(\sum_{\theta} \mu_\theta \phi_{\theta}(s)\right) \sum_{k \in \K}  \lambda_k \left[\sum_{\theta}  \xi^s_\theta \left(u^s_\theta(b^k_{\xi^s, \pi'})- \pi'(s,b^k_{\xi^s, \pi'}\right)\right]-\alpha   \\
	& \hspace{0.5cm}\ge \sum_{s \in \sset} \left(\sum_{\theta} \mu_\theta \phi_{\theta}(s)\right) \left[ \rho \sum_{k \in \mathcal{R}}  \lambda_k   \sum_\theta \xi^s_\theta \left( u^s_\theta(b^k_{\xi^s,\pi})  + u^k_\theta(b^k_{\xi^s,\pi}) - u^k_\theta(b^k_{\xi^s})  \right)   -2^{-\Omega(1/\rho)} \right]-\alpha   \\
	& \hspace{0.5cm}= \sum_{s \in \sset} \left(\sum_{\theta} \mu_\theta \phi_{\theta}(s)\right) \left[ \rho \sum_{k \in \mathcal{R}_{\phi,p,\pi}}  \lambda_k   \sum_\theta \xi^s_\theta \left( u^s_\theta(b^k_{\xi^s,\pi})  + u^k_\theta(b^k_{\xi^s,\pi}) - u^k_\theta(b^k_{\xi^s})  \right) \right]  -2^{-\Omega(1/\rho)} -\alpha   \\
	& \hspace{0.5cm}=\rho \sum_{k \in \mathcal{R}_{\phi,p,\pi}}  \lambda_k    \sum_{s \in \sset} \left(\sum_{\theta} \mu_\theta \phi_{\theta}(s)\right) \left[  \sum_\theta \xi^s_\theta \left( u^s_\theta(b^k_{\xi^s,\pi})  + u^k_\theta(b^k_{\xi^s,\pi}) - u^k_\theta(b^k_{\xi^s})  \right)  \right] -2^{-\Omega(1/\rho)} -\alpha   \\
	&\hspace{0.5cm}=\rho \sum_{k \in \mathcal{R}_{\phi,p,\pi}}  \lambda_k    \sum_{s \in \sset} \sum_{\theta} \mu_\theta \phi_{\theta}(s) \left( u^s_\theta(b^k_{\xi^s,\pi})  + u^k_\theta(b^k_{\xi^s,\pi}) - u^k_\theta(b^k_{\xi^s})  \right)   -2^{-\Omega(1/\rho)} -\alpha, 
	\end{align*}
	where the first inequality holds since $\pi'$ employs linear payments functions and $(\phi^2,p^2,\pi^2)$ has an additive loss $\alpha$ w.r.t. any protocol that employs linear payments functions, while the second inequality comes from Equation~\eqref{U2}.

	
	\paragraph{Approximate $U_3$.}
	Let $\delta_{k}\defeq \sum_{\theta} \mu_\theta u^k_\theta(b^k_\theta)- \sum_{\theta} \mu_\theta u^k_\theta(b^k_\mu)$ for each $k \in \K$, where $b^k_\theta$ is the best response of agent of type $k \in \K$ when the state of nature  is $\theta$.
	For each $k \in \mathcal{R}_{\phi,p,\pi}$, by the definition of IR it holds 
	\begin{align}\label{eq:second}
	  \sum_{s \in \sset} \sum_{\theta \in \Theta} \mu_\theta \phi_\theta(s)[\pi(s,b^k_{\xi^s,\pi})  + u^k_\theta(b^k_{\xi^s,\pi})] -p\ge \sum_{\theta} \mu_\theta u^k_\theta(b^k_\mu),  
	 \end{align}
	Hence, 
	\begin{subequations}
	\begin{align*}
	&p-\sum_{s\in \sset} \sum_{\theta \in \Theta} \mu_\theta \phi_\theta(s)\left[\pi(s,b^k_{\xi^s,\pi})  + u^k_\theta(b^k_{\xi^s,\pi})-u^k_\theta(b^k_{\xi^s})\right] \\
	&\hspace{2cm} = p-\sum_{s\in \sset} \sum_{\theta \in \Theta} \mu_\theta \phi_\theta(s)\left[\pi(s,b^k_{\xi^s,\pi})  + u^k_\theta(b^k_{\xi^s,\pi})\right] + \sum_{s\in \sset} \sum_{\theta \in \Theta} \mu_\theta \phi_\theta(s) u^k_\theta(b^k_{\xi^s}) \\
	&\hspace{2cm} \le  p-\sum_{s\in \sset} \sum_{\theta \in \Theta} \mu_\theta \phi_\theta(s)\left[\pi(s,b^k_{\xi^s,\pi})  + u^k_\theta(b^k_{\xi^s,\pi})\right] + \sum_{s\in \sset} \sum_{\theta \in \Theta} \mu_\theta \phi_\theta(s) u^k_\theta(b^k_{\theta}) \\
	&\hspace{2cm} =  p-\sum_{s\in \sset} \sum_{\theta \in \Theta} \mu_\theta \phi_\theta(s)\left[\pi(s,b^k_{\xi^s,\pi})  + u^k_\theta(b^k_{\xi^s,\pi})\right] +\sum_{\theta \in \Theta} \mu _\theta  u^k_\theta(b^k_{\theta}) \\
	&\hspace{2cm} \le -\sum_{\theta \in \Theta} \mu _\theta u^k_{\theta}(b^k_\theta) +\sum_{\theta \in \Theta} \mu _\theta  u^k_\theta(b^k_{\theta}) \\
	&\hspace{2cm}  \le \delta_k ,
	\end{align*}
	\end{subequations}
	where the first inequality follows by the optimality of action $b^k_\theta$ in state $\theta$, and the second one by Equation~\eqref{eq:second}.
	
	Next, we show that for each $\zeta \in [0,1]$ we can design a protocol with seller's utility of at least $\frac{\zeta}{2}\sum_{k \in \K} \delta_k -2^{-{1/\zeta}}$.
	Let $ P_{\zeta} \defeq \{2^{-i}\}_{i \in \{1,\dots, \lfloor1/\zeta\rfloor\}}\cup \{0\}$, and for each $k \in \K$ let $p^k$ be the greatest $p \in P_{\zeta}$ such that $p\le \delta_k$. Then, 
	\[\sum_{k \in \K} \lambda_k p^k \ge 	\sum_{k \in \K} 	\lambda_k \left(\delta_k/2-2^{-\lfloor1/\zeta\rfloor}\right)= \sum_{k \in \K} 	\lambda_k \delta_k/2-2^{-\lfloor1/\zeta\rfloor},\]
	where the inequality holds since either $p^k\ge \delta_k/2$ or $p^k\le 2^{-\lfloor1/\zeta\rfloor}$
	 
	Hence, 
	$\sum_{p \in P_{\zeta}} p \sum_{k \in \K:p^k=p} \lambda_k \ge \sum_{k \in K} 	\lambda_k \delta_k/2-2^{-\lfloor1/\zeta\rfloor}$, implying
	\[\max_{p \in P_{\zeta}} \, p \sum_{k \in \K:p^k=p} \lambda_k \ge  \frac{1}{2|P_\zeta|}\sum_{k \in \K} \lambda_k \delta_k-2^{-\lfloor1/\zeta\rfloor}\ge\frac{\zeta}{2}\sum_{k \in \K} \lambda_k \delta_k-2^{-\lfloor1/\zeta\rfloor} .\]
	Let $p^*=\argmax_{p \in  P_\zeta} p \sum_{k \in K:p^k=p} \lambda_k $. Consider the protocol $(\phi^3,p^3,\pi^3)$ that charges payment $p^3=p^*$, reveals all information with $\phi^3$ and set payment $\pi^3(s,a)=0$ for each $s\in \sset$ and $a \in \A$.
	We show that this protocol satisfies the IR constraint for all the players such that $p^k=p^*$.
	Indeed, for all these types it holds   
	\begin{align} 
	\sum_{\theta \in \Theta} \mu_\theta  u^k_\theta(b^k_{\theta}) -p^*  &\ge  \sum_{\theta \in \Theta} \mu_\theta  u^k_\theta(b^k_{\theta}) - \delta_k \nonumber\\
	& = \sum_{\theta \in \Theta} \mu_\theta  u^k_\theta(b^k_{\theta}) -\left( \sum_{\theta} \mu_\theta u^k_\theta(b^k_\theta)- \sum_{\theta} \mu_\theta u^k_\theta(b^k_\mu)\right) \ge 0.\label{eq:3}
	\end{align}
	Then, the utility of the protocol is at least the payment obtained by the buyers' type in $\R_{\phi^3,p^3,\pi^3} \supseteq \{k \in \K:p^k=p^*\} $. In particular, it is at least
	\begin{align*}
	p^* \sum_{k \in \K:p^k=p^*} \lambda_k &\ge \frac{\zeta}{2} \sum_{k \in \K} \lambda_k \delta_k-2^{-\lfloor 1/\zeta\rfloor} \\
	&\ge \frac{\zeta}{2} \sum_{k \in \R_{\phi,p,\pi}} \lambda_k \delta_k-2^{-\lfloor 1/\zeta\rfloor}\\
	&\ge \frac{\zeta}{2} \sum_{k \in \R_{\phi,p,\pi}} \lambda_k \left[p-\sum_{s\in \sset} \sum_\theta \mu_\theta \phi_\theta(s)[\pi(s,b^k_{\xi^s,\pi})  + u^k_\theta(b^k_{\xi^s,\pi})-u^k_\theta(b^k_{\xi^s})]\right] -2^{-\lfloor 1/\zeta\rfloor}\\
	&= \frac{\zeta}{2} U_3 -2^{-\lfloor 1/\zeta\rfloor} ,
	\end{align*}
	where in the the first  inequality we use Equation~\eqref{eq:3}, and in the third inequality we use Equation~\eqref{eq:4}.
	Equivalently, setting $\rho=\zeta/2$, we obtain that for each $\rho\in [0,1/2]$ there exists a protocol $(\phi^3,p^3,\pi^3)$ that has seller's utility at least $\rho U_3 -2^{-\Omega(1/\rho)}$.
	
	\paragraph{Wrapping up.} Let $i=\arg \max_{j \in \{1,2,3\}} U_j$ and $\textnormal{OPT}$ be the seller's utility with the optimal protocol $(\phi,p,\pi)$. Then, since $U_1+U_2+U_3=\textnormal{OPT}$, we have that $U_i\ge \frac{1}{3}\textnormal{OPT}$. Moreover, since for each $\rho \in [0,1/2]$ we can  approximate each utility $U_i$, $i \in \{1,2,3\}$ with a protocol with utility at least $\rho U_i-2^{-\Omega(1/\rho)}-\alpha$, the seller's utility of our approximation algorithm is at least $\rho U_i-2^{-\Omega(1/\rho)}-\alpha\ge \rho \textnormal{OPT}/3-2^{-\Omega(1/\rho)}-\alpha$. Finally, setting $\rho'=\rho/3$, we obtain that for each $\rho'\in [0,1/6]$ the utility of the designed protocol is at least $\textnormal{OPT}-2^{-\Omega(1/\rho)}-\alpha$. 
	This concludes the proof.
\end{proof}

%% file: content/appendix_const_types.tex
\constypesfirst*
\begin{proof}
	Let $(\phi,\pi,p)$ be a protocol and let be $s_1,s_2 \in \mathcal{S}$ be two signals such that $b^k_{\xi^{s_1}}=b^k_{\xi^{s_2}}$ for each receiver's type $k \in \K$. We show that it is always possible to define a new protocol $(\phi^*,\pi^*,p)$ that employs a single signal $s^*$ instead of $s_1$ and $s_2$ achieving the same seller's expected utility while satisfying the constraints.
	Formally, we  define a new signaling scheme $\phi^*$ as follows:

	$$\begin{cases}
	\phi_\theta^*(s^*)=\phi_\theta(s_1)+\phi_\theta(s_2) \,\,\, \forall \theta \in \Theta \\
	\phi_\theta^*(s)=\phi_\theta(s) \,\,\, \forall \theta \in \Theta, \,\,\, \forall s \in \mathcal {S}  \setminus \{s_1,s_2\}\end{cases}$$
	and a new payment function $\pi^*$ as follows:
	$$\begin{cases}
	\pi^*(s^*, a ) =z\pi(s_1, a )+ (1-z)\pi(s_2, a ) \,\,\,\quad  \forall a \in \A\\
	\pi^*(s, a ) =\pi(s, a ) \,\,\, \forall a \in \A, \,\,\quad \quad \quad\quad\quad \forall s \in \mathcal {S}  \setminus \{s_1,s_2\}
	\end{cases}$$
	with $z=\sum_{\theta \in \Theta}\mu_\theta\phi_\theta(s_1)/(\sum_{\theta \in\Theta}\mu_\theta(\phi_\theta(s_1)+\phi_\theta(s_2))$. As a first step, we observe that for each $k \in \K$ it holds:
	\begin{align*}
	 & \sum_{\theta \in \Theta} \mu_\theta \Big[ \phi_\theta(s_1) \left[ u_\theta^{s}(b^k_{\xi^{s_1},\pi}) -\pi(s_1, b^k_{\xi^{s_1},\pi}) \right] 
	+ \phi_\theta(s_2) \left[ u_\theta^{s}(b^k_{\xi^{s_2},\pi}) -\pi(s_2, b^k_{\xi^{s_2},\pi}) \right] \Big ] = \\ & \hspace{7cm}\sum_{\theta \in \Theta} \mu_\theta \phi^*_\theta(s^*) \left[ u_\theta^{s}(b^k_{\xi^{s^*},\pi^*}) -\pi^*(s^*, b^k_{\xi^{s^*},\pi^*}) \right].
	\end{align*}
	Moreover, for each $k \in \K$ it holds:
	\begin{align*}
	& \sum_{\theta \in \Theta} \mu_\theta \Big[ \phi_\theta(s_1) \left[ u_\theta^{k}(b^k_{\xi^{s_1},\pi}) +\pi(s_1, b^k_{\xi^{s_1},\pi}) \right] 
	+ \phi_\theta(s_2) \left[ u_\theta^{k}(b^k_{\xi^{s_2},\pi}) +\pi(s_2, b^k_{\xi^{s_2},\pi}) \right] \Big ] = \\ & \hspace{7cm}\sum_{\theta \in \Theta} \mu_\theta \phi^*_\theta(s^*) \left[ u_\theta^{k}(b^k_{\xi^{s^*},\pi^*}) +\pi^*(s^*, b^k_{\xi^{s^*},\pi^*}) \right].
	\end{align*}
	 Hence, noticing that for each signal $s\in \mathcal{S}\setminus\{s_1,s_2\}$ the seller's utility and the buyer's utility does not change from $(\phi,\pi,p)$ to $(\phi^*,\pi^*,p)$, the set $\mathcal{R}$ of buyer's type for which the  IR is satisfied does not change.
	As a consequence, the two protocols achieve the same seller's expected utility. 
	
	Applying this procedure to all the couples of signals that induces the same vector of best responses, we obtain a generalized-direct and generalized-persuasive protocol providing the same seller's expected utility.
\end{proof}

\constypesecond*
\begin{proof}
	Let $(\phi , \pi , p )$ be a protocol.
	We show that there exists a $\hat k  \in \mathcal{K}$ and a payment function $\hat \pi $ such that the protocol $(\phi, \hat\pi , b_{\hat k} )$ provides the same seller's expected  utility. Let 
	\[\hat k \in \arg\min_{k \in \mathcal{R}_{\phi,\pi,p}: b_k\ge p } \{  b_k \}.\] We observe that all the buyer's types $k \in \mathcal{R}_{\phi,\pi,p}$ have enough budget to participate in the protocol,\emph{i.e.}, $b_k\ge b_{\hat k}$. Furthermore, we define $\hat \pi(s,a) = \pi(s,a) + b_{\hat k} - p$ for each $s \in \mathcal{S}$ and $a \in \A$. 
	
	Then, we show that the set of types  $\mathcal{R}_{\phi,\pi,p}=\mathcal{R}_{\phi,\hat \pi,\hat p}$. Indeed, for each type $k \in \mathcal{R}_{\phi, \hat\pi, \hat p}$ it holds
		\begin{align*}
	 &\sum_{\theta \in \Theta} \sum_{s \in \mathcal{S}} \mu_\theta \phi_\theta(s)  \left[u_\theta^{k}(b^k_{\xi^{s},\hat \pi}) +	\hat \pi(s, b^k_{\xi^{s},\hat \pi})\right] -b_{\hat k}  \\
	 &\hspace{5cm} = \sum_{\theta \in \Theta} \sum_{s \in \mathcal{S}} \mu_\theta \phi_\theta(s)  \left[ u_\theta^{k}(b^k_{\xi^{s},\hat\pi}) +	 \pi(s, b^k_{\xi^{s},\hat\pi})+b_{\hat k}-p\right] -b_{\hat k} \\
	 &\hspace{5cm} =	\sum_{\theta \in \Theta} \sum_{s \in \mathcal{S}} \mu_\theta \phi_\theta(s)  \left[u_\theta^{k}(b^k_{\xi^{s},\pi}) +\pi(s, b^k_{\xi^{s},\pi})   \right]-p,
	\end{align*}
	and hence $k \in \mathcal{R}_{\phi, \pi, p}$.
	Similarly, we can prove that each buyer's type $k \notin \mathcal{R}_{\phi,\hat \pi,\hat p}$ does not belong to $\mathcal{R}_{\phi, \pi, p}$.
	It follows that $\mathcal{R}_{\phi,\pi,p}=\mathcal{R}_{\phi,\hat \pi,\hat p}$.
	
	Finally, we can show that the seller's utility results equal to the one in $(\phi, \pi , p)$. Indeed, we have:
	\begin{align*}
		&\sum_{k \in \mathcal{R}_{\phi,p,\pi}} \lambda_k \Big[ \sum_{\theta \in \Theta} \sum_{s \in \mathcal{S}} \mu_\theta \phi_\theta(s)  \left[u_\theta^{s}(b^k_{\xi^{s},\pi}) -\pi(s, b^k_{\xi^{s},\pi}) \right]  + p \Big] + \sum_{k  \notin \mathcal{R}_{\phi,p,\pi}} \lambda_k  \sum_{\theta \in \Theta}  \mu_\theta   u_\theta^{s}(b^k_{\mu})    \\
		& \hspace{0.6cm}=	\sum_{k \in \mathcal{R}_{\phi, \hat p,\hat \pi}} \lambda_k \Big[ \sum_{\theta \in \Theta} \sum_{s \in \mathcal{S}} \mu_\theta \phi_\theta(s) \left[ u_\theta^{s}(b^k_{\xi^{s},\hat \pi}) -\hat \pi(s, b^k_{\xi^{s},\hat \pi}) \right]  + b_{\hat k} \Big] + \sum_{k  \notin \mathcal{R}_{\phi,\hat p,\hat \pi}} \lambda_k \Big[ \sum_{\theta \in \Theta}  \mu_\theta   u_\theta^{s}(b^k_{\mu}) \Big] 
	\end{align*}
	This concludes the proof.
	%
\end{proof}

\fixedtypes*
\begin{proof}
	In the following, we present an algorithm to compute an optimal protocol that works in polynomial time when the number of buyer's types is fixed. As a first step, we observe that, thanks to Lemma \ref{lem:constypes2}, the initial payment required by the seller coincides with $b_k$ for some $k \in \mathcal{K}$. 
	Furthermore, we can focus on direct protocols by Lemma~\ref{lem:constypes1}.
	Then, given a price $p \in \{b_k\}_{k \in \K}$ and a set of buyer's types $\mathcal{R}\subseteq \K \cap \{k \in \K : b_k\ge p\} $ for which the IR constraint is satisfied, the the problem of computing the optimal protocol can be formulated as Problem~\eqref{quad:LP_kfixed}.
	Similarly to Section~\ref{sec:protocol_selection}, we can provide a  linear relaxation of Problem~\eqref{quad:LP_kfixed}  introducing a variable $l(\avec,a')$ that replaces $ \sum_{\theta \in \Theta} \mu_\theta \phi_\theta(\avec)  \pi(\avec, a')$ for each $\avec \in \A^{n}$ and $a' \in \A$.
	Then, we obtain the following LP.

	\begin{subequations}\label{eqn:LP_kfixed}
		\begin{align}
		&\max_{\phi \ge 0,l \ge 0}
		\sum_{k \in \mathcal{R}} \lambda_k \sum_{\avec \in \mathcal{A}^{n}} \left[\sum_{\theta \in \Theta} \mu_\theta \phi_\theta(\avec) u_\theta^s(a_k) -l(\avec, a_k)\right] + \sum_{k \notin \mathcal{R}} \lambda_k\sum_{\theta \in \Theta} \mu_\theta u_\theta^s(b^k_{\mu}) \quad\quad\quad\quad\quad\,\,\,\,\,\,\\ 
		&\hspace{0.5cm} \sum_{\theta \in \Theta} \mu_\theta \phi_{\theta}(\avec) u^k_\theta(a_k) + l(\avec,a_k)  \ge  \sum_{\theta \in \Theta} \mu_\theta \phi_{\theta}(\avec) u^k_\theta(a') + l(\avec,a') \nonumber\\
		&\specialcell{\hfill \forall k \in \mathcal{R}, \forall \avec \in \mathcal{A}^{n},  \forall a' \not = a_k \in \mathcal{A} }\label{17b}\\
		& \specialcell{ \hspace{0.5cm}\sum_{\avec \in \mathcal{A}^{n} }\left[\sum_{\theta \in \Theta} \mu_\theta \phi_{\theta}(\avec) u^k_\theta(a_k) + l(\avec,a_k)\right]  - b_k \ge \sum_{\theta \in \Theta} \mu_\theta  u^k_\theta(b^k_\mu) \hfill \forall k \in \mathcal{R}} \label{17c}\\
		&\specialcell{\hspace{0.5cm} \sum_{\avec \in \mathcal{A}^{n} } \left[\sum_{\theta \in \Theta} \mu_\theta \phi_{\theta}(\avec) u^k_\theta(a_k) + l(\avec,a_k)\right] - b_k \le \sum_{\theta \in \Theta} \mu_\theta  u^k_\theta(b_\mu^k)\hfill \forall k \not \in \mathcal{R}}\label{17d}\\
		&\specialcell{\hspace{0.5cm} \sum_{\avec \in \mathcal{A}^{n} } \phi_\theta(\avec)=1 \hfill \forall \theta \in \Theta.}
		\end{align}
	\end{subequations}
	Hence, once we fix $b_k$ and $\mathcal{R}$, LP~\eqref{eqn:LP_kfixed} returns a solution that has the same value of the optimal protocol.

	To compute the optimal protocol we can iterate over all the possible prices $p \in \{b_k\}_{k \in \K}$ and all the possible  subsets $\R \subseteq \K\cap \{k \in \K : b_k\ge p\} $ of receivers types for which the IR constraint is satisfied. Notice that, given a price $p$, the IR constraint can be satisfied only the buyer's type $k \in \K$  with $b_k\ge p$. Then, we solve LP~\eqref{eqn:LP_kfixed}.
	Finally, we return the solution with highest value.
	As we show in the first part of the proof, this solution has the same value of the optimal protocol.
	Moreover, it is easy to check that the overall procedure requires to solve $O(n2^{n})$ LPs, showing that the algorithm runs in polynomial time.
	
	To conclude the proof, we need to show how to modify the solution of LP~\ref{eqn:LP_kfixed} to obtain a protocol, \emph{i.e.}, a solution to Problem~\eqref{quad:LP_kfixed}, with at least the same value.
	To do so, we exploit a similar approach to the one presented in Section~\ref{sec:protocol_selection}.
	Let $(\phi,l)$ be the solution to LP~\eqref{eqn:LP_kfixed} returned by the algorithm.
	Suppose that there exists a couple $(\bar \avec,\bar k)$ such that $l(\bar \avec,a_{\bar k})>0$ and $\sum_{\theta \in \Theta} \mu_\theta \phi_\theta(\bar \avec)=0$.
	We show how to obtain a solution such that $l(\bar \avec,a)=0$ for each $a\in \A$.
	Notice that by Constraint~\eqref{17b}, it holds $l(\bar \avec,\bar a_{ k})\ge l(\bar\avec,a)$ for each $k\in \K$, $a \in \A$.
	This implies that $l(\bar \avec,\bar a_k)=l(\bar \avec,\bar a_{k'})$ for each $k \neq k'$. We denote this value with $l(\bar \avec)$.
	Let $\hat \avec \in \A^n$ be any signal such that  $\sum_{\theta \in \Theta} \mu_\theta \phi_\theta(\hat\avec)>0$.
	Consider a assignment $(\phi,l')$ to the variables such that
	\begin{itemize}
		\item $l'(\bar \avec,a)=0$ for each $a \in \A$;
		\item $l'(\hat\avec,a)=l(\hat\avec,a)+ l(\bar \avec)$ for each $a \in \A$;
		\item $l'(\avec)= l(\avec)$ for each $\avec\notin \{\bar \avec,\hat \avec\}$.
	\end{itemize}
	We show that this solution is feasible to LP~\eqref{eqn:LP_kfixed} and has the same objective value of $(\phi,l)$.
	Indeed, it holds 
	\begin{align*}
	&\sum_{k \in \mathcal{R}} \lambda_k \sum_{\avec \in \mathcal{A}^{n}} \left[\sum_{\theta \in \Theta} \mu_\theta \phi_\theta(\avec) u_\theta^s(a_k) -l'(\avec, a_k)\right] + \sum_{k \notin \mathcal{R}} \lambda_k\sum_{\theta \in \Theta} \mu_\theta u_\theta^s(b^k_{\mu}) \\
	&\hspace{2cm}=\sum_{k \in \mathcal{R}} \lambda_k \bigg[\sum_{\avec \in \mathcal{A}^{n}\setminus \{\bar\avec,\hat \avec\}} \left(\sum_{\theta \in \Theta} \mu_\theta \phi_\theta(\avec) u_\theta^s(a_k) -l'(\avec, a_k)\right)+\sum_{\theta \in \Theta} \mu_\theta \phi_\theta(\bar \avec) u_\theta^s(\bar a_k)\\
	&\hspace{5cm} +\sum_{\theta \in \Theta} \mu_\theta \phi_\theta(\hat \avec) u_\theta^s(\hat a_k) -\left(l(\hat \avec, \hat a_k)-l(\bar \avec)\right)\bigg] + \sum_{k \notin \mathcal{R}} \lambda_k\sum_{\theta \in \Theta} \mu_\theta u_\theta^s(b^k_{\mu})\\
	&\hspace{2cm}=\sum_{k \in \mathcal{R}} \lambda_k \bigg[\sum_{\avec \in \mathcal{A}^{n}\setminus \{\bar\avec,\hat \avec\}} \left(\sum_{\theta \in \Theta} \mu_\theta \phi_\theta(\avec) u_\theta^s(a_k) -l(\avec, a_k)	\right)+\sum_{\theta \in \Theta} \mu_\theta \phi_\theta(\bar \avec) u_\theta^s(\bar a_k)-l(\bar \avec,\bar a_k) \\
	&\hspace{5cm} +\sum_{\theta \in \Theta} \mu_\theta \phi_\theta(\hat \avec) u_\theta^s(\hat a_k) -l(\hat \avec, \hat a_k)\bigg] + \sum_{k \notin \mathcal{R}} \lambda_k\sum_{\theta \in \Theta} \mu_\theta u_\theta^s(b^k_{\mu})\\
	&\hspace{2cm}=\sum_{k \in \mathcal{R}} \lambda_k \sum_{\avec \in \mathcal{A}^{n}} \left[\sum_{\theta \in \Theta} \mu_\theta \phi_\theta(\avec) u_\theta^s(a_k) -l(\avec, a_k)\right] + \sum_{k \notin \mathcal{R}} \lambda_k\sum_{\theta \in \Theta} \mu_\theta u_\theta^s(b^k_{\mu}),
	\end{align*}
	showing that the seller's utility does not change.
	Moreover, Constraints~\eqref{17b} relative to $\bar \avec$ are satisfied since have the form $0\ge 0$.
	The Constraints~\eqref{17b} relative to $\hat \avec$ continue to be satisfied since we add a term $l(\bar \avec)$ on both sides of the inequality.
	Finally, all the other Constraint~\eqref{17b} are unchanged.
	Consider Constraint~\eqref{17c} relative to a buyer's type $k\in \K$.
	It holds 
	\begin{align*}
	&\sum_{\avec \in \mathcal{A}^{n} }\left[ \sum_{\theta \in \Theta} \mu_\theta \phi_{\theta}(\avec) u^k_\theta(a_k) + l'(\avec,a_k)\right] - b_k \\
	&\hspace{4cm}= \sum_{\avec \in \mathcal{A}^{n}\setminus\{\bar \avec,\hat \avec\}} \left[\sum_{\theta \in \Theta} \mu_\theta \phi_{\theta}(\avec) u^k_\theta(a_k) + l(\avec,a_k) \right]  +  \sum_{\theta \in \Theta} \mu_\theta \phi_{\theta}(\bar \avec) u^k_\theta(\bar a_k) \\
	&\hspace{7.5cm}+\sum_{\theta \in \Theta} \mu_\theta \phi_{\theta}(\hat \avec) u^k_\theta(\hat a_k) +l(\hat\avec,\hat a_k)  + l(\bar \avec)  - b_k   \\
	&\hspace{4cm}= \sum_{\avec \in \mathcal{A}^{n}\setminus\{\bar \avec,\hat \avec\}} \left[\sum_{\theta \in \Theta} \mu_\theta \phi_{\theta}(\avec) u^k_\theta(a_k) + l(\avec,a_k) \right]  +  \sum_{\theta \in \Theta} \mu_\theta \phi_{\theta}(\bar \avec) u^k_\theta(\bar a_k) \\
	&\hspace{7cm}+l(\bar \avec, \bar a_k)+\sum_{\theta \in \Theta} \mu_\theta \phi_{\theta}(\hat \avec) u^k_\theta(\hat a_k) +l(\hat\avec,\hat a_k)  - b_k   \\
	&\hspace{4cm}= \sum_{\avec \in \mathcal{A}^{n} }\left[ \sum_{\theta \in \Theta} \mu_\theta \phi_{\theta}(\avec) u^k_\theta(a_k) + l'(\avec,a_k)\right] - b_k \\
	&\hspace{4cm}\ge \sum_{\theta \in \Theta} \mu_\theta  u^k_\theta(b^k_\mu) 
	\end{align*}
	Similarly, we can show that Constraints~\eqref{17d} continue to hold.
	Hence, iteratively applying this procedure we obtain a solution with the same value of the optimal protocol and such that for each tuple $(\avec,k)$ if  $l(\avec,a_k)>0$ and $\sum_{\theta \in \Theta} \mu_\theta \phi_\theta(\avec)>0$.
	We can convert this solution into an optimal protocol, \emph{i.e.}, an optimal solution to Problem~\eqref{quad:LP_kfixed} setting $\pi(\avec,a_k)=\frac{l(\avec,a_k)}{\sum_{\theta \in \Theta} \mu_\theta \phi_\theta(\avec)}$ for each $\avec \in \A^n$ such that $\sum_{\theta \in \Theta} \mu_\theta \phi_\theta(\avec)=0$ and $k\in \K$.
	Moreover, we set all the other payments to $0$.
	It is easy to see that the obtained protocol  is a feasible optimal solution to Problem~\eqref{quad:LP_kfixed}.
	This concludes the proof.
%
%
\end{proof}

%% file: paper.bbl

\begin{thebibliography}{21}


\ifx \showCODEN    \undefined \def \showCODEN     #1{\unskip}     \fi
\ifx \showDOI      \undefined \def \showDOI       #1{#1}\fi
\ifx \showISBNx    \undefined \def \showISBNx     #1{\unskip}     \fi
\ifx \showISBNxiii \undefined \def \showISBNxiii  #1{\unskip}     \fi
\ifx \showISSN     \undefined \def \showISSN      #1{\unskip}     \fi
\ifx \showLCCN     \undefined \def \showLCCN      #1{\unskip}     \fi
\ifx \shownote     \undefined \def \shownote      #1{#1}          \fi
\ifx \showarticletitle \undefined \def \showarticletitle #1{#1}   \fi
\ifx \showURL      \undefined \def \showURL       {\relax}        \fi
\providecommand\bibfield[2]{#2}
\providecommand\bibinfo[2]{#2}
\providecommand\natexlab[1]{#1}
\providecommand\showeprint[2][]{arXiv:#2}

\bibitem[\protect\citeauthoryear{Alimonti and Kann}{Alimonti and Kann}{2000}]%
        {APXAlimonti}
\bibfield{author}{\bibinfo{person}{Paola Alimonti} {and} \bibinfo{person}{Viggo
  Kann}.} \bibinfo{year}{2000}\natexlab{}.
\newblock \showarticletitle{Some APX-completeness results for cubic graphs}.
\newblock \bibinfo{journal}{\emph{Theoretical Computer Science}}
  \bibinfo{volume}{237} (\bibinfo{date}{04} \bibinfo{year}{2000}),
  \bibinfo{pages}{123--134}.
\newblock
\urldef\tempurl%
\url{https://doi.org/10.1016/S0304-3975(98)00158-3}
\showDOI{\tempurl}


\bibitem[\protect\citeauthoryear{Alon, D{\"u}tting, and Talgam-Cohen}{Alon
  et~al\mbox{.}}{2021}]%
        {alon2021contracts}
\bibfield{author}{\bibinfo{person}{Tal Alon}, \bibinfo{person}{Paul
  D{\"u}tting}, {and} \bibinfo{person}{Inbal Talgam-Cohen}.}
  \bibinfo{year}{2021}\natexlab{}.
\newblock \showarticletitle{Contracts with Private Cost per Unit-of-Effort}. In
  \bibinfo{booktitle}{\emph{Proceedings of the 22nd ACM Conference on Economics
  and Computation}}. \bibinfo{pages}{52--69}.
\newblock


\bibitem[\protect\citeauthoryear{Alon, Dütting, Li, and Talgam-Cohen}{Alon
  et~al\mbox{.}}{2022}]%
        {alon2022bayesian}
\bibfield{author}{\bibinfo{person}{Tal Alon}, \bibinfo{person}{Paul Dütting},
  \bibinfo{person}{Yingkai Li}, {and} \bibinfo{person}{Inbal Talgam-Cohen}.}
  \bibinfo{year}{2022}\natexlab{}.
\newblock \bibinfo{title}{Bayesian Analysis of Linear Contracts}.
\newblock
\newblock
\showeprint[arxiv]{cs.GT/2211.06850}


\bibitem[\protect\citeauthoryear{Babaioff, Kleinberg, and Paes~Leme}{Babaioff
  et~al\mbox{.}}{2012}]%
        {babaioff2012}
\bibfield{author}{\bibinfo{person}{Moshe Babaioff}, \bibinfo{person}{Robert
  Kleinberg}, {and} \bibinfo{person}{Renato Paes~Leme}.}
  \bibinfo{year}{2012}\natexlab{}.
\newblock \showarticletitle{Optimal mechanisms for selling information}. In
  \bibinfo{booktitle}{\emph{Proceedings of the 13th ACM Conference on
  Electronic Commerce}}. \bibinfo{pages}{92--109}.
\newblock


\bibitem[\protect\citeauthoryear{Bergemann, Bonatti, and Smolin}{Bergemann
  et~al\mbox{.}}{2018}]%
        {bergemann2018}
\bibfield{author}{\bibinfo{person}{Dirk Bergemann}, \bibinfo{person}{Alessandro
  Bonatti}, {and} \bibinfo{person}{Alex Smolin}.}
  \bibinfo{year}{2018}\natexlab{}.
\newblock \showarticletitle{The Design and Price of Information}.
\newblock \bibinfo{journal}{\emph{American Economic Review}}
  \bibinfo{volume}{108}, \bibinfo{number}{1} (\bibinfo{date}{January}
  \bibinfo{year}{2018}), \bibinfo{pages}{1--48}.
\newblock
\urldef\tempurl%
\url{https://doi.org/10.1257/aer.20161079}
\showDOI{\tempurl}


\bibitem[\protect\citeauthoryear{Bergemann, Cai, Velegkas, and Zhao}{Bergemann
  et~al\mbox{.}}{2022}]%
        {bergemann2022}
\bibfield{author}{\bibinfo{person}{Dirk Bergemann}, \bibinfo{person}{Yang Cai},
  \bibinfo{person}{Grigoris Velegkas}, {and} \bibinfo{person}{Mingfei Zhao}.}
  \bibinfo{year}{2022}\natexlab{}.
\newblock \showarticletitle{Is Selling Complete Information (Approximately)
  Optimal?}. In \bibinfo{booktitle}{\emph{Proceedings of the 23rd ACM
  Conference on Economics and Computation}} \emph{(\bibinfo{series}{EC '22})}.
  \bibinfo{publisher}{Association for Computing Machinery},
  \bibinfo{address}{New York, NY, USA}, \bibinfo{pages}{608–663}.
\newblock
\showISBNx{9781450391504}
\urldef\tempurl%
\url{https://doi.org/10.1145/3490486.3538304}
\showDOI{\tempurl}


\bibitem[\protect\citeauthoryear{Castiglioni, Marchesi, and Gatti}{Castiglioni
  et~al\mbox{.}}{2022a}]%
        {castiglioni2021Contract}
\bibfield{author}{\bibinfo{person}{Matteo Castiglioni},
  \bibinfo{person}{Alberto Marchesi}, {and} \bibinfo{person}{Nicola Gatti}.}
  \bibinfo{year}{2022}\natexlab{a}.
\newblock \showarticletitle{Bayesian agency: Linear versus tractable
  contracts}.
\newblock \bibinfo{journal}{\emph{Artificial Intelligence}}
  \bibinfo{volume}{307} (\bibinfo{year}{2022}), \bibinfo{pages}{103684}.
\newblock


\bibitem[\protect\citeauthoryear{Castiglioni, Marchesi, and Gatti}{Castiglioni
  et~al\mbox{.}}{2022b}]%
        {Castiglioni2022Randomized}
\bibfield{author}{\bibinfo{person}{Matteo Castiglioni},
  \bibinfo{person}{Alberto Marchesi}, {and} \bibinfo{person}{Nicola Gatti}.}
  \bibinfo{year}{2022}\natexlab{b}.
\newblock \showarticletitle{Designing Menus of Contracts Efficiently: The Power
  of Randomization}.
\newblock \bibinfo{journal}{\emph{CoRR}}  \bibinfo{volume}{abs/2202.10966}
  (\bibinfo{year}{2022}).
\newblock
\urldef\tempurl%
\url{https://arxiv.org/abs/2202.10966}
\showURL{%
\tempurl}


\bibitem[\protect\citeauthoryear{Castiglioni, Marchesi, and Gatti}{Castiglioni
  et~al\mbox{.}}{2022c}]%
        {DBLP:conf/sigecom/CastiglioniM022}
\bibfield{author}{\bibinfo{person}{Matteo Castiglioni},
  \bibinfo{person}{Alberto Marchesi}, {and} \bibinfo{person}{Nicola Gatti}.}
  \bibinfo{year}{2022}\natexlab{c}.
\newblock \showarticletitle{Designing Menus of Contracts Efficiently: The Power
  of Randomization}. In \bibinfo{booktitle}{\emph{{EC} '22: The 23rd {ACM}
  Conference on Economics and Computation}}. \bibinfo{pages}{705--735}.
\newblock


\bibitem[\protect\citeauthoryear{Chen, Xu, and Zheng}{Chen
  et~al\mbox{.}}{2020}]%
        {xu2020}
\bibfield{author}{\bibinfo{person}{Yiling Chen}, \bibinfo{person}{Haifeng Xu},
  {and} \bibinfo{person}{Shuran Zheng}.} \bibinfo{year}{2020}\natexlab{}.
\newblock \showarticletitle{Selling information through consulting}. In
  \bibinfo{booktitle}{\emph{Proceedings of the Fourteenth Annual ACM-SIAM
  Symposium on Discrete Algorithms}}. SIAM, \bibinfo{pages}{2412--2431}.
\newblock


\bibitem[\protect\citeauthoryear{Daskalakis and Syrgkanis}{Daskalakis and
  Syrgkanis}{2022}]%
        {daskalakis2022learning}
\bibfield{author}{\bibinfo{person}{Constantinos Daskalakis} {and}
  \bibinfo{person}{Vasilis Syrgkanis}.} \bibinfo{year}{2022}\natexlab{}.
\newblock \showarticletitle{Learning in auctions: Regret is hard, envy is
  easy}.
\newblock \bibinfo{journal}{\emph{Games and Economic Behavior}}
  (\bibinfo{year}{2022}).
\newblock


\bibitem[\protect\citeauthoryear{Dughmi, Niazadeh, Psomas, and Weinberg}{Dughmi
  et~al\mbox{.}}{2019}]%
        {dughmi2019}
\bibfield{author}{\bibinfo{person}{Shaddin Dughmi}, \bibinfo{person}{Rad
  Niazadeh}, \bibinfo{person}{Alexandros Psomas}, {and}
  \bibinfo{person}{S~Matthew Weinberg}.} \bibinfo{year}{2019}\natexlab{}.
\newblock \showarticletitle{Persuasion and incentives through the lens of
  duality}. In \bibinfo{booktitle}{\emph{International Conference on Web and
  Internet Economics}}. Springer, \bibinfo{pages}{142--155}.
\newblock


\bibitem[\protect\citeauthoryear{Dughmi and Xu}{Dughmi and Xu}{2019}]%
        {dughmi2019algorithmic}
\bibfield{author}{\bibinfo{person}{Shaddin Dughmi} {and}
  \bibinfo{person}{Haifeng Xu}.} \bibinfo{year}{2019}\natexlab{}.
\newblock \showarticletitle{Algorithmic bayesian persuasion}.
\newblock \bibinfo{journal}{\emph{SIAM J. Comput.}} \bibinfo{volume}{50},
  \bibinfo{number}{3} (\bibinfo{year}{2019}), \bibinfo{pages}{STOC16--68}.
\newblock


\bibitem[\protect\citeauthoryear{D{\"u}tting, Roughgarden, and
  Talgam-Cohen}{D{\"u}tting et~al\mbox{.}}{2019}]%
        {dutting2019simple}
\bibfield{author}{\bibinfo{person}{Paul D{\"u}tting}, \bibinfo{person}{Tim
  Roughgarden}, {and} \bibinfo{person}{Inbal Talgam-Cohen}.}
  \bibinfo{year}{2019}\natexlab{}.
\newblock \showarticletitle{Simple versus optimal contracts}. In
  \bibinfo{booktitle}{\emph{Proceedings of the 2019 ACM Conference on Economics
  and Computation}}. \bibinfo{pages}{369--387}.
\newblock


\bibitem[\protect\citeauthoryear{Dutting, Roughgarden, and
  Talgam-Cohen}{Dutting et~al\mbox{.}}{2021}]%
        {dutting2021complexity}
\bibfield{author}{\bibinfo{person}{Paul Dutting}, \bibinfo{person}{Tim
  Roughgarden}, {and} \bibinfo{person}{Inbal Talgam-Cohen}.}
  \bibinfo{year}{2021}\natexlab{}.
\newblock \showarticletitle{The complexity of contracts}.
\newblock \bibinfo{journal}{\emph{SIAM J. Comput.}} \bibinfo{volume}{50},
  \bibinfo{number}{1} (\bibinfo{year}{2021}), \bibinfo{pages}{211--254}.
\newblock


\bibitem[\protect\citeauthoryear{Gan, Han, Wu, and Xu}{Gan
  et~al\mbox{.}}{2022}]%
        {gan2022optimal}
\bibfield{author}{\bibinfo{person}{Jiarui Gan}, \bibinfo{person}{Minbiao Han},
  \bibinfo{person}{Jibang Wu}, {and} \bibinfo{person}{Haifeng Xu}.}
  \bibinfo{year}{2022}\natexlab{}.
\newblock \showarticletitle{Optimal Coordination in Generalized Principal-Agent
  Problems: A Revisit and Extensions}.
\newblock \bibinfo{journal}{\emph{arXiv preprint arXiv:2209.01146}}
  (\bibinfo{year}{2022}).
\newblock


\bibitem[\protect\citeauthoryear{Guruganesh, Schneider, and Wang}{Guruganesh
  et~al\mbox{.}}{2021}]%
        {guruganesh2021contracts}
\bibfield{author}{\bibinfo{person}{Guru Guruganesh}, \bibinfo{person}{Jon
  Schneider}, {and} \bibinfo{person}{Joshua~R Wang}.}
  \bibinfo{year}{2021}\natexlab{}.
\newblock \showarticletitle{Contracts under moral hazard and adverse
  selection}. In \bibinfo{booktitle}{\emph{{EC} '21: The 22nd {ACM} Conference
  on Economics and Computation}}. \bibinfo{pages}{563--582}.
\newblock


\bibitem[\protect\citeauthoryear{Guruswami and Raghavendra}{Guruswami and
  Raghavendra}{2009}]%
        {Guruswami2009}
\bibfield{author}{\bibinfo{person}{Venkatesan Guruswami} {and}
  \bibinfo{person}{Prasad Raghavendra}.} \bibinfo{year}{2009}\natexlab{}.
\newblock \showarticletitle{Hardness of Learning Halfspaces with Noise}.
\newblock \bibinfo{journal}{\emph{SIAM J. Comput.}} \bibinfo{volume}{39},
  \bibinfo{number}{2} (\bibinfo{year}{2009}), \bibinfo{pages}{742--765}.
\newblock
\urldef\tempurl%
\url{https://doi.org/10.1137/070685798}
\showDOI{\tempurl}


\bibitem[\protect\citeauthoryear{Kamenica and Gentzkow}{Kamenica and
  Gentzkow}{2011}]%
        {kamenica2011bayesian}
\bibfield{author}{\bibinfo{person}{Emir Kamenica} {and}
  \bibinfo{person}{Matthew Gentzkow}.} \bibinfo{year}{2011}\natexlab{}.
\newblock \showarticletitle{Bayesian persuasion}.
\newblock \bibinfo{journal}{\emph{American Economic Review}}
  \bibinfo{volume}{101}, \bibinfo{number}{6} (\bibinfo{year}{2011}),
  \bibinfo{pages}{2590--2615}.
\newblock


\bibitem[\protect\citeauthoryear{Liu, Shen, and Xu}{Liu et~al\mbox{.}}{2021}]%
        {Liu2021}
\bibfield{author}{\bibinfo{person}{Shuze Liu}, \bibinfo{person}{Weiran Shen},
  {and} \bibinfo{person}{Haifeng Xu}.} \bibinfo{year}{2021}\natexlab{}.
\newblock \showarticletitle{Optimal Pricing of Information}.
\newblock \bibinfo{journal}{\emph{Proceedings of the 22nd ACM Conference on
  Economics and Computation}} (\bibinfo{year}{2021}).
\newblock


\bibitem[\protect\citeauthoryear{Shoham and Leyton-Brown}{Shoham and
  Leyton-Brown}{2008}]%
        {shoham2008multiagent}
\bibfield{author}{\bibinfo{person}{Yoav Shoham} {and} \bibinfo{person}{Kevin
  Leyton-Brown}.} \bibinfo{year}{2008}\natexlab{}.
\newblock \bibinfo{booktitle}{\emph{Multiagent systems: Algorithmic,
  game-theoretic, and logical foundations}}.
\newblock \bibinfo{publisher}{Cambridge University Press}.
\newblock


\end{thebibliography}
